\documentclass[11pt, twoside]{IEEEtran} \onecolumn 

\pdfoutput=1

\title{Online Matrix Completion and Online Robust PCA}

\author{
Brian~Lois and Namrata~Vaswani
\thanks{B. Lois is with the Mathematics and ECE departments, and N. Vaswani is with the ECE department at Iowa State University.
Email: \{blois,namrata\}@iastate.edu. An early version of this work will be presented at ICASSP 2015 \cite{rrpcp_icassp15} and another part (with the same title as this paper) is under submission to ISIT 2015. This work was partly supported by NSF grant CCF-1117125.
}
}

\usepackage{epsfig,color,epstopdf}
\usepackage{amsmath, amssymb, amsthm, algorithm, algorithmic,graphicx}
\usepackage{bm}
\usepackage{datetime}

\usepackage{caption}
\usepackage{subcaption}

\usepackage{tikz}
\usetikzlibrary{decorations.pathreplacing}

\usepackage{enumerate}
\usepackage[normalem]{ulem}

\usepackage{color}

\newtheorem{theorem}{Theorem}[section]
\newtheorem{lem}[theorem]{Lemma}
\newtheorem{sigmodel}[theorem]{Model}
\newtheorem{corollary}[theorem]{Corollary}
\newtheorem{definition}[theorem]{Definition}
\newtheorem{remark}[theorem]{Remark}

\newtheorem{fact}[theorem]{Fact}

\newcommand{\ds}{\displaystyle}

\newcommand{\R}{\mathbb{R}}

\renewcommand{\Pr}{\mathbb{P}}

\newcommand{\bi}{\begin{itemize}}
\newcommand{\ei}{\end{itemize}}
\newcommand{\ben}{\begin{enumerate}}
\newcommand{\een}{\end{enumerate}}

\newcommand{\bean}{\begin{eqnarray*} }
\newcommand{\eean}{\end{eqnarray*} }
\newcommand{\bea}{\begin{eqnarray} }
\newcommand{\eea}{\end{eqnarray} }
\newcommand{\nn}{\nonumber}
\newcommand{\ba}{\begin{array} }
\newcommand{\ea}{\end{array} }
\newcommand{\beq}{\begin{equation}}
\newcommand{\eeq}{\end{equation}}

\newcommand{\vect}[2]{\left[\begin{array}{cccccc}
     #1 \\
     #2
   \end{array}
  \right]
  }

\renewcommand\thetheorem{\arabic{section}.\arabic{theorem}}

\newcommand{\mt}{\bm{m}_t}
\newcommand{\xt}{\bm{x}_t}
\newcommand{\xhatt}{\hat{\bm{x}}_t}
\newcommand{\lt}{\bm{\ell}_t}
\renewcommand{\l}{\bm{\ell}}
\newcommand{\e}{\bm{e}}
\newcommand{\lhatt}{\hat{\bm{\ell}}_t}
\newcommand{\et}{\bm{e}_t}
\newcommand{\Pt}{\bm{P}_t}

\newcommand{\Pjs}{\bm{P}_{(j),*}}

\newcommand{\Pjnew}{\bm{P}_{(j),\mathrm{new}}}
\newcommand{\at}{\bm{a}_t}
\newcommand{\ats}{\bm{a}_{t,*}}
\newcommand{\atnew}{\bm{a}_{t,\mathrm{new}}}
\newcommand{\I}{\bm{I}}
\newcommand{\Lamt}{\bm{\Lambda}_t}
\newcommand{\Lamtnew}{\bm{\Lambda}_{t,\mathrm{new}}}
\newcommand{\T}{\mathcal{T}}
\newcommand{\J}{\mathcal{J}}

\newcommand{\D}{\bm{D}}

\newcommand{\ttrain}{t_{\mathrm{train}}}
\newcommand{\rmnew}{\mathrm{new}}
\newcommand{\new}{\mathrm{new}}
\newcommand{\cs}{\text{cs}}

\newcommand{\bigo}{\mathcal{O}}

\newcommand{\Lhat}{\hat{\bm{\ell}}}
\newcommand{\lhat}{\hat{\bm{\ell}}}

\newcommand{\Ltil}{\tilde{\bm{\ell}}}
\renewcommand{\P}{\bm{P}}
\newcommand{\Phat}{\hat{\bm{P}}}

\newcommand{\Span}{\operatorname{range}}
\newcommand{\del}{\mathrm{del}}

\newcommand{\add}{\mathrm{add}}

\newcommand{\rank}{\operatorname{rank}}
\newcommand{\E}{\mathbb{E}}

\newcommand{\calc}{\mathcal{C}}
\newcommand{\full}{\mathrm{full}}
\newcommand{\cov}{\operatorname{Cov}}

\newcommand{\train}{\mathrm{train}}
\newcommand{\thresh}{\mathrm{thresh}}  %{\hat{\lambda}^-}
\newcommand{\That}{\hat{\mathcal{T}}}

\newcommand{\SE}{\mathrm{SE}}

\newcommand{\that}{{\hat{t}}}
\newcommand{\rmend}{\mathrm{end}}
\newcommand{\M}{\bm{\mathcal{M}}}

\newcommand{\llceil}{\left\lceil}
\newcommand{\rrceil}{\right\rceil}
\newcommand{\llfloor}{\left\lfloor}
\newcommand{\rrfloor}{\right\rfloor}
\newcommand{\uhat}{{\hat{u}}}
\newcommand{\lammin}{{\hat{\lambda}_{\train}^-}}
\newcommand{\jhat}{\hat{\jmath}}

\newcommand{\rrho}{\rho}

\begin{document}

\maketitle

% major edits done by NV to:
%1. rename \lambda_\thresh as  \lammin (defined above in \newcommand)
%2. replace lambda_new^- by \lammin  in subspace change model (and corresponding changes) and in the choice of \epsilon
%3. Appendix B and h+ lemma and general support model
%4. proof of zeta_k^+ lemma in Appendix D
%5. bound on \mathcal{H}_k
%6. minor changes also to other things and to discussion

%questions for Brian
%0. x Move preliminaries to just before the proof section -- keep proof in appendix only??
%1. x the random variable $X$ etc: you do not use bold -- should you use bold??
%2. should we remove subscript j also from $\Phat_*$, $P_*$ etc ??

% further to-do
% check everything
% simulations and add back figures - easy
% discussion and related work section, contributions section, abstract - is more needed ?
% Extensions section: add back the support change at every frame model and maybe also the random support change model. do we need to explain how object moving down is a special case of sbyrho model?

%
%{\LARGE
%To Do:
%\begin{enumerate}
%\item Check new support change model
%\item Figures
%\item Simulations
%\item Proofreading
%\end{enumerate}
%}

\begin{abstract}
 This work studies two interrelated problems - online robust PCA (RPCA) and online low-rank matrix completion (MC).  In recent work by Cand\`{e}s et al., RPCA has been defined as a problem of separating a low-rank matrix (true data), $L:=[\ell_1, \ell_2, \dots \ell_{t}, \dots , \ell_{t_{\max}}]$ and a sparse matrix (outliers), $S:=[x_1, x_2, \dots x_{t}, \dots, x_{t_{\max}}]$ from their sum, $M:=L+S$. Our work uses this definition of RPCA. An important application where both these problems occur is in video analytics in trying to separate sparse foregrounds (e.g., moving objects) and slowly changing backgrounds.

While there has been a large amount of recent work on both developing and analyzing batch RPCA and batch MC algorithms, the online problem is largely open. In this work, we develop a practical modification of our recently proposed algorithm to solve both the online RPCA and online MC problems.  The main contribution of this work is that we obtain correctness results for the proposed algorithms under mild assumptions.
The assumptions that we need are: (a) a good estimate of the initial subspace is available (easy to obtain using a short sequence of background-only frames in video surveillance); (b) the $\ell_t$'s obey a `slow subspace change' assumption; (c) the basis vectors for the subspace from which $\ell_t$ is generated are dense (non-sparse); (d) the support of $x_t$ changes by at least a certain amount at least every so often; and (e) algorithm parameters are appropriately set.
%The assumption on support change that we require is weaker than that required by the batch approaches; but this is because we exploit extra assumptions such as initial subspace knowledge and slow subspace change.
\end{abstract}

\section{Introduction}
Principal Components Analysis (PCA) is a tool that is frequently used for dimension reduction.  Given a matrix of data $\D$, PCA
computes a small number of orthogonal directions, called principal components, that contain most of the variability of the data.
For relatively noise-free data that lies close to a low-dimensional subspace, PCA is easily accomplished via singular value decomposition (SVD). %of $\D$. % and retaining the singular vectors corresponding to the largest singular values.
The problem of PCA in the presence of outliers is referred to as robust PCA (RPCA).
%A limitation of SVD is that it is highly sensitive to outliers in the data set.
%Recently there has been much work done to develop and analyze algorithms for PCA that are robust with respect to outliers.
%A common way to model outliers is as sparse vectors \cite{error_correction_PCP_l1}.
In recent work, Cand\`{e}s et al. \cite{rpca} posed RPCA as a problem of separating a low-rank matrix, $\bm{L}$, and a sparse matrix, $\bm{S}$, from their sum, $\bm{M}:=\bm{L}+\bm{S}$. They proposed a convex program called principal components' pursuit (PCP) that provided a provably correct batch solution to this problem under mild assumptions.
PCP solves
\[
\min_{\tilde{\bm{L}},\tilde{\bm{S}}} \|\tilde{\bm{L}}\|_* + \lambda\|\tilde{\bm{S}}\|_{\operatorname{sum}} \quad \text{subject to} \quad
\tilde{\bm{L}} + \tilde{\bm{S}} = \bm{M},
\]
where $\|\cdot\|_*$ is the nuclear norm (sum of singular values), $\|\cdot\|_{\operatorname{sum}}$ is the sum of the absolute values of the entries, and $\lambda$ is an appropriately chosen scalar.
The same program was analyzed in parallel by Chandrasekharan et al. \cite{rpca2} and later by Hsu et al. \cite{rpca_zhang}. Since these works, there has been a large amount of work on batch approaches for RPCA and their performance guarantees.

%In seminal papers Cand\`{e}s et. al. and Chandrasekaran et. al. introduced the Principal Components Pursuit (PCP) convex program and proved its robustness to sparse outliers \cite{rpca}, \cite{rpca2}. Principal Components Pursuit poses the robust PCA problem as identifying a low rank matrix and a sparse matrix from their sum. The program is to minimize a weighted sum of the nuclear norm of the low rank matrix and the vector $\ell_1$ norm of the sparse matrix subject to their sum being equal to the observed data matrix. The results in \cite{chen2011low} improve upon those in \cite{rpca2}.
%Other methods such as \cite{outlier_pursuit} model the entire column vector as being either correct or an outlier.  Some other works on the performance guarantees for batch robust PCA include \cite{rpca_tropp}, \cite{linear_inverse_prob}, and \cite{noisy_undersampled_yuan}. All of these methods require waiting until all of the data has been acquired before performing the optimization.

When RPCA needs to be solved in a recursive fashion for sequentially arriving data vectors it is referred to as online (or recursive) RPCA. %By online or recursive we mean obtain estimates of the sparse matrix columns and of the low-rank matrix columns and its subspace using the previous estimates, rather than re-solving the entire problem at each time $t$.
Online RPCA assumes that a short sequence of outlier-free (sparse component free) data vectors is available. %which can be used to obtain an accurate estimate of the initial subspace of the true data vectors (which lie in a low-dimensional subspace). % An application where this type of problem is useful is in video analysis \cite{Torre03aframework}.
An example application where this problem occurs is the problem of separating a video sequence into foreground and background layers (video layering) on-the-fly \cite{rpca}. Video layering is a key first step for automatic video surveillance and many other streaming video analytics tasks. In videos, the foreground usually consists of one or more moving persons or objects and hence is a sparse image. The background images usually change only gradually over time \cite{rpca}, e.g., moving lake waters or moving trees in a forest, and hence are well modeled as lying in a low-dimensional subspace that is fixed or slowly changing. Also, the changes are global (dense) \cite{rpca}. In most video applications, it is valid to assume that an initial short sequence of background-only frames is available and this can be used to estimate the initial subspace via SVD.

Often in video applications the sparse foreground $\xt$ is actually the signal of interest, and the background $\lt$ is the noise.  In this case, the problem can be interpreted as one of recursive sparse recovery in (potentially) large but structured noise.  Our result allows for $\lt$ to be large in magnitude as long as it is structured.  The structure we impose is that the $\lt$'s lie in a low dimensional subspace that changes {\em slowly} over time.

%In case of a static camera video,
%Video layering is a key first step to simplifying many video analytics and computer vision tasks, e.g., video surveillance (to track moving foreground objects), background video recovery and subspace tracking in the presence of frequent foreground occlusions or low-bandwidth mobile video chats or video conferencing (can transmit only the foreground layer).

In some other applications, instead of there being outliers, parts of a data vector may be missing entirely. When the (unknown) complete data vector is a column of a low-rank matrix, the problem of recovering it is referred to as matrix completion (MC). For example, recovering video sequences and tracking their subspace changes in the presence of easily detectable foreground occlusions. If the occluding object's intensity is known and is significantly different from that of the background, its support can be obtained by simple thresholding. The background video recovery problem then becomes an MC problem. A nuclear norm minimization (NNM) based solution for MC was introduced in \cite{matcomp_first} and studied in \cite{Candes_Recht}.  The convex program here is to minimize the nuclear norm of $\tilde{\bm{M}}$ subject to $\tilde{\bm{M}}$ and $\bm{M}$ agreeing on all observed entries.
Since then there has been a large amount of work on batch methods for MC and their correctness results.

\subsection{Problem Definition}
Consider the {\em online MC problem.}
Let $\T_t$ denote the set of missing entries at time $t$. We observe a vector $\mt \in \R^n$ that satisfies
\bea
\mt = \I_{\overline{\T}_t}{\I_{\overline{\T}_t}}' \lt \qquad \text{ for } t = t_{\train}+1, t_{\train}+2, \dots, t_{\max}.
\label{omc_eq}
\eea
with the possibility that $t_{\max}$ can be infinity too. Here  $\lt$ is such that, for $t$ large enough (quantified in Model \ref{exp_model}), the matrix $\bm{L}_t:=[\l_1, \l_2, \dots, \lt]$ is a low-rank matrix.
%from a slowly changing low-dimensional subspace. An important special case of this is random vectors $\lt$ that have a covariance matrix that is low-rank at each time and slowly changing
%The model ensures that the resulting matrix $\L_t:=[\l_1, \l_2, \dots \lt]$ is low-rank for any $t$. We also assume that the missing sets $\T_t$ have {\em some} changes over time. We provide one model for this in Model \ref{sbyrho}.
Notice that by defining $\mt$ as above, we are setting to zero the entries that are missed (see the notation section on page \pageref{notation}).

Consider the {\em online RPCA} problem.  At time $t$ we observe a vector $\mt \in \R^n$ that satisfies %that is the sum of a vector from a slowly changing low-dimensional subspace $\lt$ and a sparse vector $\xt$. So
\bea
\mt = \lt + \xt \qquad \text{ for } t = t_{\train}+1, t_{\train}+2, \dots, t_{\max}.
\label{orpca_eq}
\eea
%with the possibility that $t_{\max} = \infty$. We model the low-dimensional $\lt$'s as above. We use $\T_t$ to denote the support set of $\xt$.
Here $\lt$ is as defined above and $\xt$ is the sparse (outlier) vector. We use $\T_t$ to denote the support set of $\xt$.

For both problems above, for $t=1,2,\dots, t_\train$, we are given complete outlier-free measurements $\mt = \lt$ so that it is possible to estimate the initial subspace. For the video surveillance application, this would correspond to having a short initial sequence of background only images, which can often be obtained.
For $t > t_\train$, the goal is to estimate $\lt$ (or $\lt$ and $\xt$ in case of RPCA) as soon as $\mt$ arrives and to periodically update the estimate of $\Span(\bm{L}_t)$.

In the rest of the paper, we refer to $\T_t$ as the {\em missing/corrupted entries' set}.

\subsection{Related Work}
Some other work that also studies the online MC problem (defined differently from above) includes \cite{sequentialSVD,grouse,petrels,local_conv_grouse}. We discuss the connection with the idea from \cite{sequentialSVD} in Section \ref{algosubsec}. The algorithm from \cite{grouse}, GROUSE, is a first order stochastic gradient method; a result for its convergence to the local minimum of the cost function it optimizes is obtained in \cite{local_conv_grouse}. The algorithm of \cite{petrels}, PETRELS, is a second order stochastic gradient method. It is shown in \cite{petrels} that PETRELS converges to the stationary point of  the cost function it optimizes. The advantage of PETRELS and GROUSE is that they do not need initial subspace knowledge. Another somewhat related work is \cite{aarti_adaptive_mc}.
% Our algorithms could also benefit from using their approach for initialization. %For our approach proposed above, when the initial subspace knowledge is not available or initial complete data is not available, we can also use the PETRELS or GROUSE ideas for initialization.
%
%Grassmanian Rank-One Update Subspace Estimation (GROUSE) \cite{grouse} and  Parallel Subspace Estimation and Tracking by Recursive Least Squares From Partial Observations (PETRELS) \cite{petrels}.

Partial results have been provided for ReProCS for online RPCA in our older work \cite{ReProCS_IT}. In other more recent work \cite{xu_nips2013_1} another partial result is obtained for online RPCA defined differently from above. Neither of these is a correctness result. Both require an assumption that depends on intermediate algorithm estimates. Another somewhat related work is \cite{OnlinePCA_ContaminatedData} on online PCA with contaminated data. This does not model the outlier as a sparse vector but defines anything that is far from the data subspace as an outlier.

%and follow up papers \cite{zhan_correlated, rrpcp_globalsip};  however, none of these results is a correctness result. All require an assumption that depends on intermediate algorithm estimates. In other more recent work, Feng et al. provide either a partial result \cite{xu_nips2013_1} or provide a complete but asymptotic result for online robust PCA defined differently \cite{xu_nips2013_2}.
%Recent work of Feng et. al. \cite{OnlinePCA_ContaminatedData,rpca_stochatistic_optimization} provides partial results for online RPCA. One of these papers, \cite{OnlinePCA_ContaminatedData}, does not model the outlier as a sparse vector;  \cite{rpca_stochatistic_optimization} does, but it again contains a partial result. Moreover the theorems in both papers solve a very different formulation of online robust PCA.
Some other works only provide an algorithm without proving any performance results, e.g., \cite{grass_undersampled}.

We discuss the most related works in detail in Sec \ref{discussion_online}.

%Let $\Phat_0$ be a matrix whose columns form an orthonormal basis for this estimate.
%rrpcp_perf,rrpcp_icassp15
\subsection{Contributions}
In this work we develop and study a practical modification of the Recursive Projected Compressive Sensing (ReProCS) algorithm introduced and studied in our earlier work \cite{ReProCS_IT} for online RPCA. We also develop a special case of it that solves the online MC problem. The main contribution of this work is that we obtain a {\em complete correctness result} for ReProCS-based algorithms for both online MC and online RPCA (or more generally, online sparse plus low-rank matrix recovery). Online algorithms are useful because they are causal (needed for applications like video surveillance) and, in most cases, are faster and need less storage compared to most batch techniques (we should mention here that there is some recent work on faster batch techniques as well, e.g., \cite{lowrank_altmin}).
To the best of our knowledge, this work and an earlier conference version of this \cite{rrpcp_icassp15} may be among the first correctness results for online RPCA. The algorithm studied in \cite{rrpcp_icassp15} is more restrictive. %because it assumes knowledge of subspace change times and numbers of new directions added.
%To the best of our knowledge, this is among the first works that provides a correctness result for online RPCA, online MC and more generally online sparse plus low-rank matrix recovery.  %As shown in \cite{han_tsp}, with practical heuristics used to set its parameters, ReProCS has significantly improved recovery performance compared to many other recursive (\cite{GRASTA,mateos2012robust, mardani2013dynamic}) and batch methods (\cite{rpca, Torre03aframework, mateos2012robust}) for robust PCA. %many simulated and real video datasets.
%

Moreover, as we will see, by exploiting temporal dependencies, such as slow subspace change, and initial subspace knowledge, our result is able to allow for a more correlated set of missing/corrupted entries than do the various results for  PCP \cite{rpca,rpca2,rpca_zhang} or NNM \cite{Candes_Recht} (see Sec. \ref{discussion}).
%For experimental comparisons and time comparisons of ReProCS with batch methods and with , see \cite{han_tsp}.

Our result uses the overall proof approach introduced in our earlier work \cite{ReProCS_IT} that provided a partial result for online RPCA. The most important new insight needed to get a complete result is described in Section \ref{insight}. Also see Sec. \ref{discussion_online}. New proof techniques are needed for this line of work because almost all existing works only analyze batch algorithms that solve a different problem. Also, as explained in Section \ref{algosubsec}, the standard PCA procedure cannot be used in the subspace update step and hence results for it are not applicable.

As shown in \cite{han_tsp}, because it exploits initial subspace knowledge and slow subspace change, ReProCS has significantly improved recovery performance compared with  batch RPCA algorithms - PCP \cite{rpca} and \cite{Torre03aframework} -  as well as with the online algorithm of \cite{grass_undersampled} for foreground and background extraction in many simulated and real video sequences; it is also faster than the batch methods but slower than \cite{grass_undersampled}.

%We show that as long as algorithm parameters are set appropriately, a good-enough estimate of the initial subspace is available, a slow subspace change assumption holds, the subspaces are dense enough for a given maximum support size and a given maximum rank, and there is a certain amount of support change at least every so often, then the support can be exactly recovered with high probability. Also the sparse and low-rank matrix columns can be recovered with bounded and small error.
\subsection{Notation}\label{notation}
We use $'$ to denote transpose.
The  2-norm of a vector and the induced 2-norm of a matrix are denoted by $\| \cdot \|_2$.
For a set $\mathcal{T}$ of integers, $|\mathcal{T}|$ denotes its cardinality and $\overline{\mathcal{T}}$ denotes its complement set. We use $\emptyset$ to denote the empty set.
For a vector $\bm{x}$, $\bm{x}_{\mathcal{T}}$ is a smaller vector containing the entries of $\bm{x}$ indexed by $\mathcal{T}$.
Define $\I_{\mathcal{T}}$ to be an $n \times |\mathcal{T}|$ matrix of those columns of the identity matrix indexed by $\mathcal{T}$. For a matrix $\bm{A}$, define $\bm{A}_{\mathcal{T}} := \bm{AI}_{\mathcal{T}}$.
For matrices $\bm{P}$ and $\bm{Q}$ where the columns of $\bm{Q}$ are a subset of the columns of $\bm{P}$, $\bm{P} \setminus \bm{Q}$ refers to the matrix of columns in $\bm{P}$ and not in $\bm{Q}$.

For an $n\times n$ Hermitian matrix $\bm{H}$, $\bm{H}\overset{\mathrm{EVD}}= \bm{U\Lambda U}'$ denotes an eigenvalue decomposition.  That is, $\bm{U}$ has orthonormal columns, and $\bm{\Lambda}$ is a diagonal matrix of size at least $\rank(\bm{H})\times\rank(\bm{H})$. (If $\bm{H}$ is rank deficient, then $\bm{\Lambda}$ can have any size between $\rank(\bm{H})$ and $n$.)
For Hermitian matrices $\bm{A}$ and $\bm{B}$, the notation $\bm{A}\preceq\bm{B}$ means that $\bm{B}-\bm{A}$ is positive semi-definite.
We order the eigenvalues of an Hermitian matrix in decreasing order.  So $\lambda_1 \geq \lambda_2 \geq \dots \geq \lambda_n$.

For integers $a$ and $b$, we use the interval notation $[a, b]$ to mean all of the integers between $a$ and $b$, inclusive, and similarly for $[a,b)$ etc.

\begin{definition}
For a matrix $\bm{A}$, the restricted isometry constant (RIC) $\delta_s(\bm{A})$ is the smallest real number $\delta_s$ such that
\[
(1-\delta_s)\|\bm{x}\|_2^2 \leq \|\bm{Ax}\|_2^2 \leq (1+\delta_s) \|\bm{x}\|_2^2
\]
for all $s$-sparse vectors $\bm{x}$ \cite{candes_rip}.  A vector $\bm{x}$ is $s$-sparse if it has $s$ or fewer non-zero entries.
\end{definition}

\begin{definition}
We refer to a matrix with orthonormal columns as a {\em basis matrix}. Notice that if $\bm{P}$ is a basis matrix, then ${\bm{P}}'\bm{P} = \bm{I}$.
\end{definition}
\begin{definition}
For basis matrices $\Phat$ and $\bm{P}$, define $\mathrm{dif}(\Phat,\bm{P}):= \|(\I - \Phat \Phat') \bm{P} \|_2$. This quantifies the difference between their range spaces. If $\Phat$ and $\bm{P}$ have the same number of columns, then $\mathrm{dif}(\Phat,\bm{P}) = \mathrm{dif}(\bm{P},\Phat)$, otherwise the function is not necessarily symmetric.
\end{definition}

%%\begin{definition}
%For ease of notation, define the function
%$
%\digamma(\alpha,\epsilon,b_1) := \exp\left( \frac{-\alpha\epsilon^2}{8(b_1)^2} \right)
%$

\subsection{Organization}
The remainder of the paper is organized as follows.  In Section \ref{Problem Definition and Assumptions} we give the model and main result for both online MC and online RPCA.
%Section \ref{Problem Definition and Assumptions_rPCA} is short, and explains how the matrix completion result extends to the robust PCA problem by adding a step to the algorithm and an additional assumption.
Next we discuss our main results in Section \ref{discussion}. The algorithms for solving both problems are given and discussed in Section \ref{algosubsec}. The discussion also explains why the proof of our main result should go through. Section \ref{insight} within this section describes the key insight needed by the proof and Section \ref{outline} gives the proof outline.
The most general form of our model on the missing entries set, $\T_t$, is described in Section \ref{gen_supch_sec}.  A key new lemma for proving our main results is also given in this section.  The proof of our main results can be found in Section \ref{pf_thm}. Proofs of three long lemmas needed for proving the lemmas leading to the main theorem are postponed until Section \ref{3_pfs}. Section \ref{sims} shows numerical experiments backing up our claims. We discuss some extensions in Section \ref{extensions} and give conclusions in Section \ref{conclusions}
%of the main theorems (Theorem \ref{thm1} and Theorem \ref{thm1_mc})

%\newpage

\section{Online Matrix Completion: Assumptions and Main Result} \label{Problem Definition and Assumptions}

Before we give our model on $\lt$, we need the following definition.
\begin{definition} \label{lammin_def}
Recall that $\mt = \lt$ for $t = 1,\dots,t_{\train}$ is the training data.
Let $\lammin$ be the minimum non-zero eigenvalue of $\frac{1}{t_{\train}} \sum_{t=1}^{t_{\train}} \mt{\mt}' $.  That is
\[
\lammin := \min_{\lambda_i > 0} \lambda_i \left( \frac{1}{t_{\train}}\sum_{t=1}^{t_{\train}} \mt{\mt}'  \right)
\]
Define $\Phat_{t_\train}$ to be the matrix containing the eigenvectors of $\frac{1}{t_{\train}} \sum_{t=1}^{t_{\train}} \mt{\mt}'$, with eigenvalues larger than or equal to $\lammin$, as its columns.
\end{definition}
We will use $\Phat_{t_\train}$ as the initial subspace knowledge in the algorithms.
We will use $\lammin$ in our algorithms to set the eigenvalue threshold to both detect subspace change and estimate the number of newly added directions. We also use $\lammin$ to state the slow subspace change assumption below {\color{blue} We use this to state the most general version of the slow subspace change assumption in Model \ref{exp_model}. However, as explained in the footnote in the line below (\ref{slow_lam_new}), we can get a slightly more restrictive model without using $\lammin$}.

\subsection{Model on \texorpdfstring{$\lt$}{l_t}} \label{ltmodel}
We assume that $\lt$ is a vector from a fixed or slowly changing low-dimensional subspace that changes in such a way that the matrix $\bm{L}_t:=[\l_1, \l_2, \dots \l_t]$ is low rank for $t$ large enough. This can be modeled in various ways.
{\color{blue}
The simplest and most commonly used model for data that lies in a low-dimensional subspace is to assume that at all times, it is is independent and identically distributed (iid) with zero mean and a fixed covariance matrix $\bm\Sigma$ that is low rank. However this is often impractical since, in most applications, data statistics change with time, albeit ``slowly". To model this perfectly, one would need to assume that $\lt$ is zero mean with covariance $\bm{\Sigma}_t$ that changes at each time. Let $\bm{\Sigma}_t = \P_t \Lamt \P_t'$ denote its diagonalization (with $\P_t$ tall); then this means that both $\P_t$ and $\Lamt$ can change at each time $t$. This is the most general case but it but it has an identifiability problem for estimating the subspace of $\lt$. The subspace spanned by the columns of $\P_t$ cannot be estimated with one data point. If $\P_t$ has $r_t$ columns, one needs $r_t$ or more data points for its accurate estimation.
So, if $\P_t$ changes at each time, it is not clear how all the subspaces can be accurately estimated. Moreover, in general (without specific assumptions), this will not ensure that the matrix $\bm{L}_t$ is low rank.
To resolve this issue, a general enough but tractable option is to assume that $\P_t$ is constant for a certain period of time and then changes and $\Lamt$ can change at each time. To ensure that $\bm{\Sigma}_t$ changes ``slowly", we assume that, when $\P_t$ changes, the eigenvalues along the newly added subspace directions are small for some time ($d$ frames) after the change. One precise model for this is specified next. We also assumed boundedness of $\lt$. This is more practically valid rather than the usual Gaussian assumption (often made for simplicity) since most sensor data or noise is bounded. %In our case, as we will see boundedness is necessary to analyze the projected sparse recovery step.
%One possible model is given below. It assumes that $\lt$'s are zero mean, bounded and mutually independent random variables with a covariance matrix that is low-rank at each time and that changes ``slowly" in the following fashion: (a) its column subspace remains constant for a long enough time and then changes; (b) when it changes, the number of newly added directions is small and the eigenvalues along the newly added directions are small for $d$ frames after the change.
}

\begin{figure}
\includegraphics[width=\textwidth]{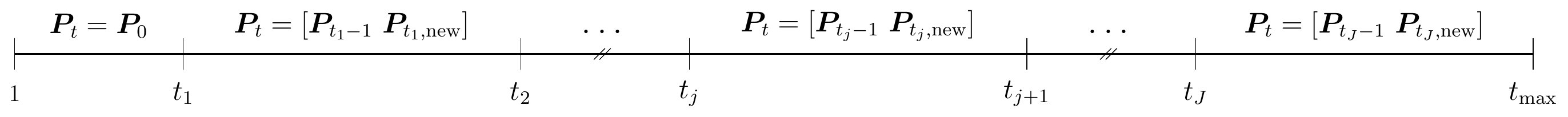}
\caption{A diagram of Model \ref{exp_model} \label{modelfig}}
\end{figure}

\begin{sigmodel}[Model on $\lt$]\label{exp_model}
Assume that the $\lt$ are zero mean and bounded random vectors in $\R^n$ that are mutually independent over time. Also assume that their covariance matrix $\bm{\Sigma}_t$ has an eigenvalue decomposition
\[
\E[\lt{\lt}'] = \bm{\Sigma}_t \overset{\mathrm{EVD}}{=} \Pt \Lamt {\Pt}'
\]
% that have not appeared before
where $\Pt$ changes as
\begin{align}\label{Pt_def}
\bm{P}_t =
\begin{cases}
 [\bm{P}_{t-1}  \ \bm{P}_{t,\new}]  & \text{if} \ t = t_1 \text{ or } t_2 \text{ or } \dots \ t_J \\
\bm{P}_{t-1}  &   \text{otherwise.}
\end{cases}
\end{align}
%(here $t_1 < t_2 \dots < t_J$ are the subspace change times)
and $\Lamt$ changes as follows. For $t\in[t_{j},t_{j+1})$, define  $\bm{\Lambda}_{t,\new} := {\bm{P}_{t_j,\new}}'\bm{\Sigma}_{t}\bm{P}_{t_j,\new}$ and assume that
\begin{equation}\label{slow_lam_new}
(\bm{\Lambda}_{t,\new})_{i,i} = (v_i)^{t-t_j} q_i  \lammin  \ \text{ for } \ i = 1,\dots,r_{j,\new}
\end{equation}
where $q_i \ge 1$ and $v_i>1$ but not too large
\footnote{Our result would still hold if the $v_i$ were different for each change time (i.e. $v_{j,i}$).  We let them be the same to reduce notation.
{\color{blue} If we do not want to use $\lammin$ here in the model on $\lt$, we can replace (\ref{slow_lam_new}) by $(\bm{\Lambda}_{t,\new})_{i,i} = (v_i)^{t-t_j} q_i \lambda_{\mathrm{bnd}}$ (for a positive constant $\lambda_{\mathrm{bnd}}$) instead and assume in the theorem that $\lammin$ is close to $\lambda_{\mathrm{bnd}}$, e.g. $0.9 \lambda_{\mathrm{bnd}} \le \lammin \le 1.1 \lambda_{\mathrm{bnd}}$ will suffice.}
}.
We assume that (a) $t_{j+1} - t_{j} \geq d$ for a $d\geq(K+2)\alpha$; and (b) for all $i$, $q_i (v_i)^d \leq 3$.
Here $K$ and $\alpha$ are algorithm parameters that are set in Theorem \ref{thm1}.

Other minor assumptions are as follows. (i) Define $t_0:= 1$ and assume that $t_\train \in [t_0, t_1)$.
(ii) For $j=0,1,2, \dots, J$, define
\[
r_j: = \rank(\bm{P}_{t_j}) \ \text{ and } \ r_{j,\new}:=\rank(\bm{P}_{t_j,\new}).
\]
and assume that $r_J < \min(n, t_{j+1}-t_j)$. This ensures that, for all $t > r_J$, the matrix $\bm{L}_t$ is low-rank.
(iii) Define
\[\lambda^+ := \sup_t \lambda_{\max}(\bm{\Lambda}_t)\]
as the maximum eigenvalue at any time and assume that $\lambda^+ < \infty$.

Observe from the above that $\P_t$ is a basis matrix and $\Lamt$ is diagonal. We refer to the $t_j$'s as the subspace change times.
\end{sigmodel}
A visual depiction of the above model can be found in Figure \ref{modelfig}.

Define the largest and smallest eigenvalues along the new directions for the first $d$ frames after a subspace change as
\[
\lambda_{\new}^+ := \max_j \max_{t\in [t_j,t_j + d]} \lambda_{\max}\left( \bm{\Lambda}_{t,\new} \right) \ \text{and} \  \lambda_{\new}^- := \min_j \min_{t\in [t_j,t_j + d]} \lambda_{\min}\left( \bm{\Lambda}_{t,\new} \right)
\]
The slow change model on $\bm{\Lambda}_{t,\new}$ is one way to ensure that
\begin{eqnarray} \label{anew_small}
\lammin \leq \lambda_{\new}^- \leq \lambda_{\new}^+  \leq 3 \lammin
\end{eqnarray}
i.e. the maximum variance of the projection of $\lt$ along the new directions is small enough for the first $d$ frames after a change. Also the minimum variance is larger than a constant greater than zero (and hence detectable). The proof of our main result only relies on (\ref{anew_small}) and does not use the actual slow increase model in any other way. The above inequality along with $t_{j+1} - t_{j} \geq d \geq(K+2)\alpha$ quantifies {\em ``slow subspace change".}

Notice that the above model does not put any assumption on the eigenvalues along the existing directions. In particular, they do not need to be greater than zero and hence the model automatically allows existing directions (columns of $\bm{P}_{t_j-1}$ for $t \in [t_j ,t_{j+1})$) to drop out of the current subspace. It could be the case that for some time period, $(\bm{\Lambda}_{t})_{i,i} = 0$ (for an $i$ corresponding to a column of $\bm{P}_{t_j-1}$), so that the $i^{\text{th}}$ column of $\bm{P}_{t_j-1}$ is not contributing anything to $\lt$ at that time. For the same index $i$, $(\bm{\Lambda}_{t})_{i,i}$ could also later increase again to a nonzero value.
Therefore $r_0 + \sum_{i=1}^{j} r_{i,\new}$ is only a bound on the rank of $\bm{\Sigma}_t$ for $t\in[t_j,t_{j+1})$, and not necessarily the rank itself.
A more explicit model for deletion of directions is to let $\bm{P}_t$ change as
\bea
\label{del_model}
\bm{P}_t =
\begin{cases}
 [(\bm{P}_{t-1} \setminus  \bm{P}_{t,\del})  \  \ \bm{P}_{t,\new}]  & \text{if} \ t = t_1 \text{ or } t_2 \text{ or } \dots \ t_J \\
\bm{P}_{t-1}  &   \text{otherwise.}
\end{cases}
\eea
where $\bm{P}_{t,\del}$ contains the columns of $\bm{P}_{t-1}$ for which the variance is zero. If we add the assumption that $[\bm{P}_{t_1-1} \ \bm{P}_{t_1,\new} \ \bm{P}_{t_2,\new} \ \dots \ \bm{P}_{t_J,\new}]$ be a basis matrix (i.e. deleted directions cannot be part of a later $\bm{P}_{t_j,\new}$), then this is a special case of Model \ref{exp_model} above.  We say special case because this only allows deletions at times $t_j$, whereas Model \ref{exp_model} allows deletion of old directions at any time.

{\color{blue}
The above model assumes that $\lt$'s are zero mean and mutually independent over time. In the video analytics application, zero mean is easy to ensure by letting $\lt$ be the background image at time $t$ with an empirical `mean background image' (computed using the training data) subtracted out. The independence assumption then models independent background variations around a common mean. As we explain in Section \ref{extensions}, this can be easily relaxed and we can get a result very similar to the current one under a first order autoregressive model on the $\lt$'s.
}

For $t \in [t_j,t_{j+1})$, let $\bm{P}_{t,*} := \bm{P}_{t_j-1}$ and $\bm{\Lambda}_{t,*} := {\bm{P}_{t,*}}' \bm{\Sigma}_t \bm{P}_{t,*}$.
Observe that Model \ref{exp_model} does not have any constraint on $\bm{\Lambda}_{t,*}$. Thus if we assume that its entries are such that their changes from $t$ to $t+1$ are smaller than or equal to $\|\bm{\Lambda}_{t,\new} - \bm{\Lambda}_{t+1,\new}\|_2$, then clearly, $\frac{\|\bm{\Sigma}_{t+1}-\bm{\Sigma}_t  \|_2}{\|\bm{\Sigma}_t  \|_2} \le (3^{1/d}-1)$ for all $t \in [t_j,t_j+d]$ and all $j$ \footnote{This follows because $\|\bm{\Sigma}_t \|_2 \ge \|\Lamtnew\|_2 = \max_i (v_i)^{t-t_j} q_i \lammin$ and $\|\bm{\Sigma}_{t+1} - \bm{\Sigma}_{t}\|_2 \le \|\bm{\Lambda}_{t+1,\new} - \bm{\Lambda}_{t,\new}\|_2 \le \max_i  (v_i)^{t-t_j} q_i \lammin (v_i - 1) \le \max_i  (v_i)^{t-t_j} q_i  \lammin \max_i (v_i-1)$. Thus the ratio is bounded by $\max_i (v_i-1) \le (3/q_i)^{1/d}-1  < (3^{1/d} - 1)$ since $q_i \ge 1$.}. Since $d$ is large, the upper bound is a small quantity, i.e. the covariance matrix changes slowly. For later time instants, we do not have any requirement (and so in particular $\bm{\Sigma}_t  $ could still change slowly). Hence the above model includes ``slow changing" and low-rank $\bm{\Sigma}_t$ as a special case.

\begin{figure}[t]
\centering
\begin{subfigure}{.5\linewidth}
\centering
\includegraphics[height=3in]{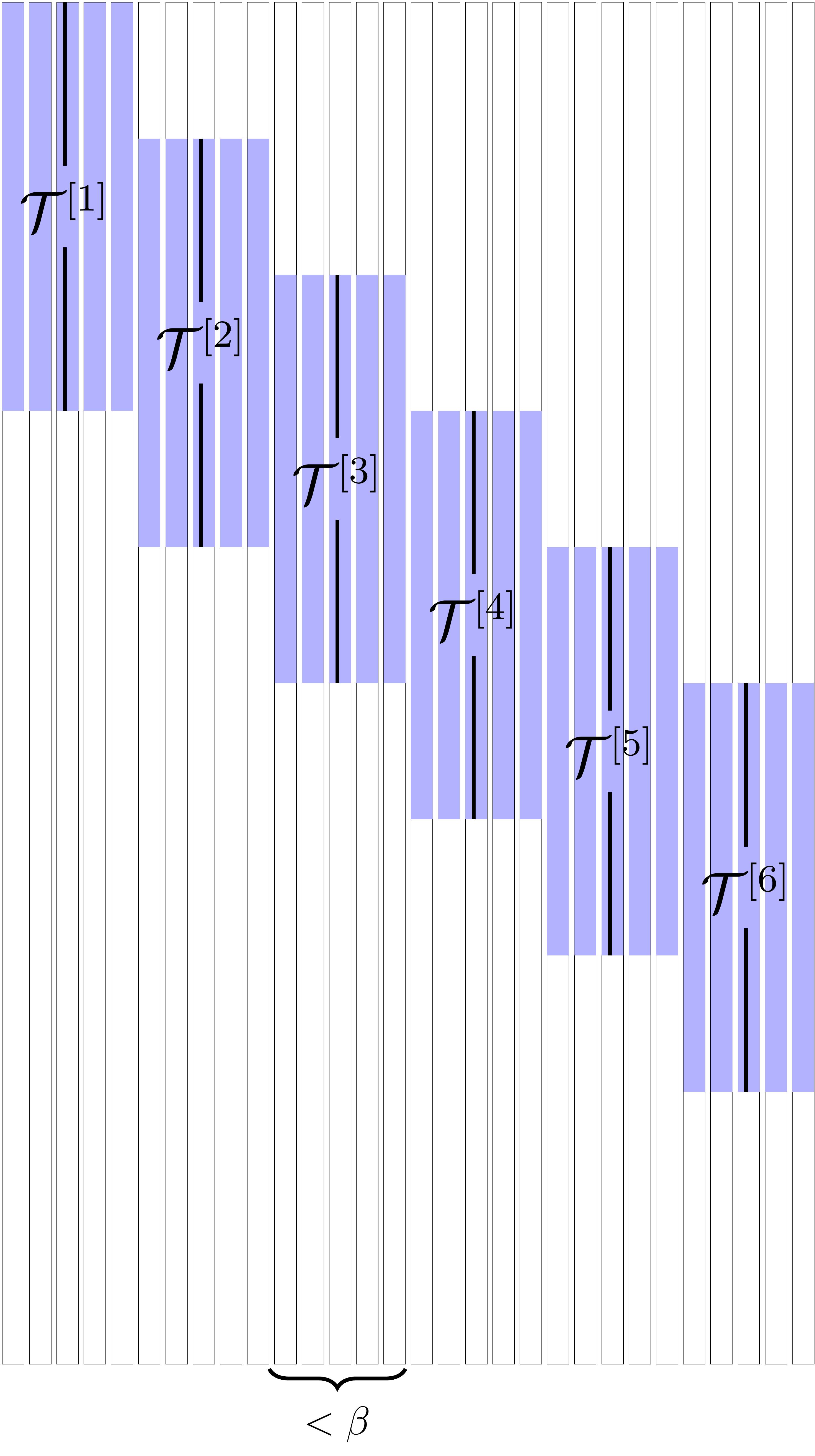}
\caption{$\varrho=3$ and $\beta=5$ case}
\end{subfigure}
\begin{subfigure}{.3\linewidth}
\centering
\includegraphics[height=3in]{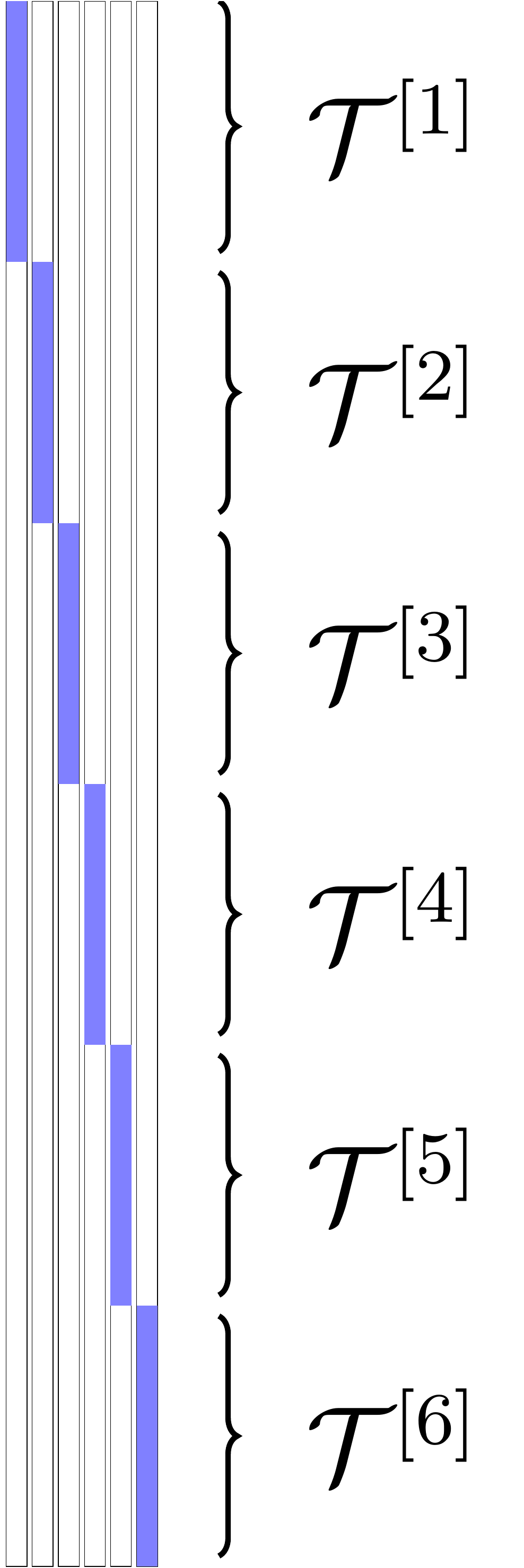}
\caption{$\varrho=1$ and $\beta=1$ case}
\end{subfigure}
\caption{\label{fig1} Examples of Model \ref{sbyrho}. (a) shows a 1D video object of length $s$ that moves by at least $s/3$ pixels once every $5$ frames. (b) shows the object moving by $s$ at every frame. (b) is an example of the best case for our result - the case with smallest $\rho,\beta$ ($\T_t$'s  mutually disjoint)}
\label{supportfig}
\end{figure}

\subsection{Model on the set of missing entries or the outlier support set, $\T_t$} \label{xtmodel}
Our result requires that the set of missing entries (or the outlier support sets), $\T_t$, have {\em some} changes over time. We give one simple model for it below. One example that satisfies this model is a video application consisting of a foreground with one object of length $s$ or less that can remain static for at most $\beta$ frames at a time. When it moves, it moves {\em downwards} (or upwards, but always in one direction) by at least $s/\rrho$ pixels, and at most $s/\rho_2$ pixels. Once it reaches the bottom of the scene, it disappears. The maximum motion is such that, if the object were to move at each frame, it still does not go from the top to the bottom of the scene in a time interval of length $\alpha$, i.e. $\frac{s}{\rho_2}\alpha \le n$. Anytime after it has disappeared another object could appear. We show this example in Fig. \ref{supportfig}.
\begin{sigmodel}[model on $\T_t$]\label{sbyrho}
Let $t^k$, with $t^k < t^{k+1}$, denote the times at which $\T_t$ changes and let $ \T^{[k]}$ denote the distinct sets.
For an integer $\alpha$ (we set its value in Theorem \ref{thm1}), assume the following.
\ben
\item Assume that $\T_t = \T^{[k]}$ for all times $t \in [t^k, t^{k+1})$ with $(t^{k+1} - t^k) < \beta$ and $|\T^{[k]}| \le s$.
\item Let $\rrho$ be a positive integer so that for any $k$,
\[
\T^{[k]} \cap \T^{[k+\rrho ]} = \emptyset;
\]
assume that
\[
{\rrho}^2 \beta \leq  0.01 \alpha.
\]

\item For any $k$,
\[
\sum_{i=k+1}^{k+\alpha} \left|\T^{[i]} \setminus \T^{[i+1]}\right| \le n
\]
and for any $k < i \le k+\alpha$,
\[
(\T^{[k]} \setminus \T^{[k+1]}) \cap (\T^{[i]} \setminus \T^{[i+1]}) = \emptyset.
\]
(One way to ensure $\sum_{i=k+1}^{k+\alpha}|\T^{[i]} \setminus \T^{[i+1]}| \le n$ is to require that  for all $i$, $|\T^{[i]} \setminus \T^{[i+1]}| \le \frac{s}{\rho_2}$ with $\frac{s}{\rho_2}\alpha \le n$.)

%\item For any $k$ and $i$ such that $k < i \le k+\alpha$, $$(\T^{[k]} \setminus \T^{[k+1]}) \cap (\T^{[i]} \setminus \T^{[i+1]}) = \emptyset$$
%(an implicit extra requirement for the above to hold is $\sum_{i=k+1}^{k+\alpha}|\T^{[i]} \setminus \T^{[i+1]}| \le n$; one way to ensure this holds is
%$$|\T^{[i]} \setminus \T^{[i+1]}| \le \frac{s}{\rho_2} \text{ and } \frac{s}{\rho_2}\alpha \le n.)$$ %$\sum_{i=k+1}^{k+\alpha}|\T^{[i]} \setminus \T^{[i+1]}| \le n$
\een
In this model, $k$ takes values $1,2, \dots $; the largest value it can take is $t_{\max}$ (this will happen if $\T_t$ changes at every time).
\end{sigmodel}
Clearly the video moving object example satisfies the above model as long as $\rho^2 \beta \le 0.01 \alpha$. \footnote{Let $\T_t$ be the support set of the object (set of pixels containing the object). The first condition holds since there is at most one object of size $s$ or less and the object cannot remain static for more than $\beta$ frames. Since it moves in one direction by at least $s/\rrho$ each time it moves, this means that definitely after it moves $\rho$ times, the supports will be disjoint (second condition). The third condition holds because it moves in one direction and by at most $s/\rho_2$ with $\frac{s}{\rho_2}\alpha \le n$ (so even if it were to move at each $t$, i.e. if $t_{k+1}=t_k+1$ for all $k$, the third condition will hold). Also see Fig. \ref{supportfig}.} This becomes clearer from  Fig. \ref{supportfig}. %it will not come back to where it started at the beginning of the $\alpha$ frame interval

\subsection{Denseness}
In order to recover the $\lt$'s from missing data or to separate them from the sparse outliers, the basis vectors for the subspace from which they are generated cannot be sparse.  We quantify this using the incoherence condition from \cite{rpca}.  Let $\mu$ be the smallest real number such that
\begin{equation}\label{candes_dense}
\max_{i} \|{\bm{P}_{t_0}}' \I_i\|_2^2 \leq \frac{\mu r_0}{n}
\quad \text{ and } \quad
\max_{i} \|{\bm{P}_{t_j,\new}}' \I_i\|_2^2 \leq \frac{\mu r_{j,\new}}{n} \text{ for all } j
\end{equation}
Recall from the notation section that $\I_i$ is the $i^{\text{th}}$ column of the identity matrix (or $i^{\text{th}}$ standard basis vector). We bound $\mu r_0$ and $\mu r_{j,\new}$ in the theorem.

\subsection{Main Result for Online Matrix Completion}\label{results}
%We state below our main result for the ReProCS algorithm (Algorithm \ref{reprocsdet}) explained in Section \ref{algosubsec}.

%Define  $r_{\new} := \max_j \rank(\bm{P}_{t_j,\new})$ and
\begin{definition}
Recall that $r_{j,\new}:= \rank(\P_{t_j,\new})$ and $r_{j}:= \rank(\P_{t_j})$. Define $r_{\new} := \max_j r_{j,\new}$,   and $r = r_0 + Jr_{\new}$.

Also define $\at := {\bm{P}_t}' \lt$, and for $t \in [t_j, t_{j+1})$, $\atnew := {\bm{P}_{t_j,\new}}'\lt$.
Let
\[
\gamma := \max_t \|\bm{a}_t\|_{\infty} \quad\text{ and }\quad \gamma_{\new} := \max_j \max_{t\in[t_j,t_j+d]} \|\bm{a}_{t,\new}\|_{\infty}
\]
Notice that $\rank(\bm{L}) = \rank(\bm{P}_{t_{\max}}) \leq r$. Also, $\|\at\|_2 \leq \sqrt{r}\gamma$ and for $t\in[t_j,t_j+d]$, $\|\bm{a}_{t,\new}\|_2 \le \sqrt{r_\new} \gamma_\new$.
\end{definition}

The following theorem gives a correctness result for Algorithm \ref{reprocsdet} given and explained in Section \ref{algosubsec}. The algorithm has two parameters - $\alpha$ and $K$. The parameter $\alpha$ is the number of consecutive time instants that are used to obtain {\em an} estimate of the new subspace, and $K$ is the total number of times the new subspace is estimated before we get an accurate enough estimate of it. The algorithm uses $\lammin$ and $\Phat_{t_\train}$ defined in Definition \ref{lammin_def} and $\mt$ as inputs.
%In the theorem below, we set $\alpha$ so that the probability of everything being good (all $J$ new subspaces getting accurately recovered within $d$ frames of a subspace change) is at least $1-n^{-10}$; and we set $K$ so that the $K$-th estimate of the new subspace is accurate enough (its error is below $r_\new \zeta$ for a small enough $\zeta$).

%NV
\begin{theorem}\label{thm1_mc}
Consider Algorithm \ref{reprocsdet}. Assume that $\mt$ satisfies \eqref{omc_eq}. Pick a $\zeta$ that satisfies
\[
\zeta \leq \min \left\{ \frac{10^{-4}}{r^2} , \frac{0.03\lammin}{r^2 \lambda^+} , \frac{1}{r^3\gamma^2}, \frac{\lammin}{r^3\gamma^2} \right\}.
\]
Suppose that the following hold.
\begin{enumerate}

%\item $\ds\|(\I - \Phat_{\ttrain} \Phat_{\ttrain}{}') \bm{P}_{\ttrain}\|_2 \le r_0 \zeta$ (notice from Model \ref{exp_model} that $\bm{P}_{\ttrain} = \bm{P}_{t_0} = \bm{P}_1$  );

\item $\mathrm{dif}( \Phat_{\ttrain}, \bm{P}_{\ttrain} ) \le r_0 \zeta$ (notice from Model \ref{exp_model} that $\bm{P}_{\ttrain} = \bm{P}_{t_0} = \bm{P}_1$);

\item The algorithm parameters are set as: \\
%$\thresh=\frac{\lammin}{2}$;
 $K = \left\lceil \frac{\log(0.16r_{\new}\zeta)}{\log(0.83)}\right\rceil$; and $\alpha = C (\log(6(K+1)J) + 11\log(n))$ for a constant %$C \geq C_{\add}$ with
\begin{equation}\label{alpha}
C \ge C_{\add} := {32\cdot100^2}\frac{\max\{16 , 1.2(\sqrt{\zeta} + \sqrt{r_\new}\gamma_\new)^4\}}{ \left(r_{\new} \zeta \lammin \right)^2};
\end{equation}

\item (Subspace change) Model \ref{exp_model} on $\bm{\ell}_t$ holds; % with $\lambda_{\new}^- \geq \lammin$;

\item \label{supchass} (Changes in the missing/corrupted sets $\T_t$)
Model \ref{sbyrho} on $\T_t$ holds or its generalization, Model \ref{general_model} (given in Section \ref{gen_supch_sec}), holds;

\item (Denseness and bound on $s$, $r_0$, $r_\new$) the bounds in \eqref{candes_dense} hold with $2s(r_0+Jr_{\new})\mu\leq{0.09n}$ and  $2sr_{\new}\mu\leq{0.0004n}$;
%The low dimensional subspace is dense enough and the rank and support size are small enough ?? so that

\end{enumerate}

Then, with probability at least $1 - n^{-10}$, at all times $t$,
\ben
\item $\|\lhat_t - \lt \|_2 \leq 1.2 \left(\sqrt{\zeta} +  \sqrt{r_{\new}}\gamma_{\rmnew}  \right)$

\item the subspace error $\SE_t := \|( \I - \hat{\bm{P}}_t \hat{\bm{P}}_t{}' ) \Pt \|_2$ is bounded above by $10^{-2} \sqrt{\zeta}$ for $t\in[t_{j}+d,t_{j+1})$.

\een
\end{theorem}

\begin{proof}
The proof is given in Sections \ref{pf_thm} and \ref{3_pfs}. As shown in Lemma \ref{spc_case}, Model \ref{sbyrho} is a special case of Model \ref{general_model} (Model \ref{general_model} is more general) on $\T_t$. Hence we prove the result only using Model \ref{general_model}.
\end{proof}
Theorem \ref{thm1_mc} says that if an accurate estimate of the initial subspace is available; the two algorithm parameters are set appropriately; the $\lt$'s are mutually independent over time and the low-dimensional subspace from which $\lt$ is generated changes ``slowly" enough, i.e. (a) the delay between change times is large enough ($d \ge (K+2)\alpha$) and (b) the eigenvalues along the newly added directions are small enough for $d$ frames after a subspace change (so that \eqref{anew_small} holds); the set of missing entries at time $t$, $\T_t$, has enough changes; and the basis vectors that span the low-dimensional subspaces are dense enough; then, with high probability (w.h.p.), the error in estimating $\lt$ will be small at all times $t$.  Also, the error in estimating the low-dimensional subspace will be initially large when new directions are added, but will decay to a small constant times $\sqrt{\zeta}$ within a finite delay.

 %$t_\train$ is large enough (if $\lt$'s are continuous random vectors, then $t_\train \ge r_0$ will suffice almost surely),
Consider the accurate initial subspace assumption. If the training data truly satisfies $\mt=\lt$ (without any noise or modeling error) and if we have at least $r_0$ linearly independent $\lt$'s (if $\lt$'s are continuous random vectors, this corresponds to needing $t_\train \ge r_0$ almost surely),  then the estimate of $\Span(\bm{P}_{t_\train})$ obtained from training data will actually be exact, i.e. we will have $\mathrm{dif}( \Phat_{\ttrain}, \bm{P}_{\ttrain} ) = 0$.
The theorem assumption that $\mathrm{dif}( \Phat_{\ttrain}, \bm{P}_{\ttrain} ) \le r_0 \zeta$ allows for the initial training data to be noisy or not exactly satisfying the model. If the training data is noisy, we need to know $r_0$ (in practice this is computed by thresholding to retain a certain percentage of largest eigenvalues). In this case we can let $\lammin$ be the $r_0$-th eigenvalue of $\frac{1}{\alpha} \sum_{t=1}^{t_\train} \mt {\mt}'$ and  $\Phat_{t_\train}$ be the $r_0$ top eigenvectors.
%we can let $\lammin$ be the smallest value so that the sum of all eigenvalues of $\frac{1}{\alpha} \sum_{t=1}^{t_\train} \mt {\mt}'$ that are larger than or equal to $\lammin$ is at least 99\% of the sum of all eigenvalues of this matrix and we can let $\Phat_{t_\train}$ be the corresponding eigenvectors.

The following corollary is also proved when we prove the above result.

\begin{corollary} \label{thm1_mc_cor}
The following conclusions also hold under the assumptions of Theorem \ref{thm1} with probability at least $1 - n^{-10}$
\begin{enumerate}
\item The estimates of the subspace change times given by Algorithm \ref{reprocsdet} satisfy $t_j \leq \hat{t}_j \leq t_j + 2\alpha$, for $j = 1,\dots, J$;
\item The estimates of the number of new directions are correct, i.e. $\hat{r}_{j,\new,k} = r_{j,\new}$ for $j=1,\dots,J$ and $k = 1,\dots,K$;
\item The recovery error satisfies:
\begin{align*}
\|\lhat_t - \lt \|_2 &\leq
\begin{cases}
1.2 \left(\sqrt{\zeta} +  \sqrt{r_{\new}}\gamma_{\rmnew}  \right) & t \in \left[t_j,\hat{t}_j\right] \\
1.2 \left(1.84\sqrt{\zeta} +  (0.83)^{k-1}\sqrt{r_{\new}}\gamma_{\rmnew}  \right) &   t \in \left[\hat{t}_j + (k-1)\alpha, \hat{t}_j + k\alpha -1\right], \ k=1,2, \dots, K \\
2.4\sqrt{\zeta} & t \in \left[\hat{t}_j + K \alpha, t_{j+1} - 1\right];
\end{cases}
\end{align*}
\item The subspace error satisfies,
\begin{align*}
\SE_t \leq
\begin{cases}
1  & t \in \left[t_j,\hat{t}_j\right] \\
 10^{-2} \sqrt{\zeta} +  0.83^{k-1} & t \in \left[\hat{t}_j + (k-1)\alpha, \hat{t}_j + k\alpha -1\right], \ k=1,2, \dots, K  \\
10^{-2} \sqrt{\zeta}   &  t \in \left[\hat{t}_j + K \alpha, t_{j+1} - 1\right].
\end{cases}
\end{align*}
\end{enumerate}
\end{corollary}

\subsection{Main Result for Online Robust PCA}
Recall that in this case we assume that the observations $\mt$ satisfy
$
\mt = \lt + \xt
$
with the support of $\xt$, denoted $\T_t$, not known. We have the following result for Algorithm \ref{reprocsdet_orpca} given and explained in Section \ref{algosubsec}. This requires two extra assumptions beyond what the previous result needed. For the matrix completion problem, the set of missing entries is known, while in the robust PCA setting, the support set, $\T_t$, of the sparse outliers, $\xt$, must be determined.  We recover this using an ell-1 minimization step followed by thresholding. To do this correctly, we need a lower bound on the absolute values of the nonzero entries of $\xt$. Moreover, Algorithm \ref{reprocsdet_orpca} has two extra parameters - $\xi$, which is the bound on the two norm of the noise seen by the ell-1 minimization step, and $\omega$, which is the threshold used to recover the support of $\xt$. These need to be set appropriately.

\begin{theorem} \label{thm1}
Consider Algorithm \ref{reprocsdet_orpca}. Assume that $\mt$ satisfies \eqref{orpca_eq} and assume everything else in Theorem \ref{thm1_mc}. Also assume
\ben
\item The two extra algorithm parameters are set as: $\xi = \sqrt{r_{\new}}\gamma_{\new} + (\sqrt{r} + \sqrt{r_{\new}})\sqrt{\zeta}$ and $\omega = 7\xi$
\item We have $x_{\min}:= \min_t \min_{i: (\xt)_i \neq 0} |(\xt)_i| > 14 \xi$
\een
Then with probability at least $1-n^{-10}$,
\ben
\item all conclusions of Theorem \ref{thm1_mc} and Corollary \ref{thm1_mc_cor} hold;
\item the support set $\T_t$ is exactly recovered, i.e. $\hat{\T}_t = \T_t$ for all $t$;
\item $\|\xt - \xhatt\|_2 = \|\lt - \lhatt\|_2$ and $\|\lt - \lhatt\|_2$ satisfies the bounds given in Theorem \ref{thm1_mc} and Corollary \ref{thm1_mc_cor}.
\een
\end{theorem}
The second assumption above can be interpreted as either a lower bound on $x_{\min}$, or as an upper bound on $\sqrt{r_\new} \gamma_\new$ in terms of $x_{\min}$. This latter interpretation is another ``slow subspace change" condition. For the $\xt$'s, this result shows that their support is exactly recovered w.h.p. and its nonzero entries are accurately recovered.

\subsection{Simple Generalizations}

{\em Model on $\ell_t$. }
Consider the subspace change model, Model \ref{exp_model}. For simplicity we put a slow increase model on the eigenvalues along the new directions for the entire period $[t_j, t_{j+1})$. However, as explained below the model, the proof of our result does not actually use this slow increase model. It only uses \eqref{anew_small}, i.e. $\lammin \leq \lambda_{\new}^- \leq \lambda_{\new}^+  \leq 3 \lammin$. Recall that $\lambda_{\new}^-$ and $\lambda_{\new}^+$ are the minimum and maximum eigenvalues along the new directions for the first $d$ frames after a subspace change. Thus, in the interval $[t_j+d+1, t_{j+1})$ our proof actually does not need any constraint on $\bm{\Lambda}_{t,\new}$. %In particular, its entries could become as large as $\lambda^+$ after $t_j+d$ and surely after $t_{j+1}$.

With a minor modification to our proof, we can prove our result with an even weaker condition. We need \eqref{anew_small} to hold with $\lambda_{\new}^-$ being the minimum of the minimum eigenvalues of any $\alpha$-frame {\em average} covariance matrix along the new directions over the period $[t_j,t_j+d]$, i.e. with $\lambda_{\new}^- = \min_j \min_{\tau \in [t_j, t_j+d-\alpha]} \lambda_{\min}( \frac{1}{\alpha} \sum_{t=\tau}^{\tau+\alpha-1} \Lamtnew)$. For video analytics, this translates to requiring that, after a subspace change, {\em enough (but not necessarily all)} background frames have `detectable' energy along the new directions, so that the minimum eigenvalue of the average covariance is above a threshold. % (so that the minimum eigenvalue of the average covariance matrix.

Secondly, we should point out that there is a trade off between the bound on $q_i {v_i}^d$, and consequently on $\lambda_{\new}^+$, in Model \ref{exp_model} and the bound on ${\rrho}^2 \beta$ assumed in Model \ref{sbyrho}. Allowing a larger value of $q_i {v_i}^d$ (faster subspace change) will require a tighter bound on $\rho^2\beta$ which corresponds to requiring more changes to $\T_t$. We chose the bounds $q_i (v_i)^d\leq3$ and $\rho^2 \beta \leq .01 \alpha$ for simplicity of computations. There are many other pairs that would also work. The above trade-off can be seen from the proof of Lemma \ref{zetadecay}. The proof uses Model \ref{general_model} of which Model \ref{sbyrho} is a special case.
For video analytics, this means that if the background subspace changes are faster, then we also need the foreground objects to be moving more so we can `see' enough of the background behind them.

%NV
Thirdly, in Model \ref{exp_model} we let $\P_t \Lamt {\P_t}'$ be an EVD of $\bm{\Sigma}_t$. This automatically implies that $\Lamt$ is diagonal. But our proof only uses the fact that $\Lamt$ is block diagonal with blocks $\bm{\Lambda}_{t,*}$ and $\bm{\Lambda}_{t,\new}$. If we relax this and we let $\P_t \Lamt {\P_t}'$ be {\em a} decomposition of $\bm{\Sigma}_t$ where $\Lamt$ is block diagonal as above, then our model allows the variance along {\em any} direction from $\Span(\P_{t_j-1})$ to become zero for any period of time and/or become nonzero again later.
Thus, in the special case of \eqref{del_model} we can actually allow $\bm{P}_t  =[(\bm{P}_{t-1} \bm{R}_t \setminus  \bm{P}_{t,\del})  \  \ \bm{P}_{t,\new}]$, where $\bm{R}_t$ is an $r_{j-1}\times r_{j-1}$ rotation matrix and $\bm{P}_{t,\del}$ contains the columns of $\bm{P}_{t-1}\bm{R}_t$ for which the variance is zero. This will be a special case of this generalization if $[\bm{P}_{t_1-1} \ \bm{P}_{t_1,\new} \ \bm{P}_{t_2,\new} \ \dots \ \bm{P}_{t_J,\new}]$ is a basis matrix.
%A more explicit model for deletion of directions is to let $\bm{P}_t$ change as
%\[
%\bm{P}_t =
%\begin{cases}
% [(\bm{P}_{t-1}\bm{R}_t \setminus  \bm{P}_{t,\del})  \  \ \bm{P}_{t,\new}]  & \text{if} \ t = t_1 \text{ or } t_2 \text{ or } \dots t_J \\
%\bm{P}_{t-1}  &   \text{otherwise,}
%\end{cases}
%\]
%where $\bm{R}_t$ is an $r_{j-1}\times r_{j-1}$ rotation matrix and $\bm{P}_{t,\del}$ contains the columns of $\bm{P}_{t-1}\bm{R}_t$ for which the variance is zero.
%If we add the assumption that the matrix $[\bm{P}_{t_1-1} \ \bm{P}_{t_1,\new} \ \bm{P}_{t_2,\new} \ \dots \ \bm{P}_{t_J,\new}]$ is a basis matrix (i.e. $\P_{t_j,\new}$ is orthogonal to the initial subspace as well as to all previously added new directions), then this would be a special case of the above generalization of Model \ref{exp_model}. %We say special case because this only allows deletions at times $t_j$, whereas Model \ref{exp_model} allows deletion of old directions at any time.

{\em Initialization condition. }
The first condition of the theorem requires that we have accurate initial subspace knowledge. As explained below the theorem, this means that we can allow noisy training data. Moreover, notice that if we let $t_1 = t_\train + 1$, then new background directions can enter the subspace at the same time as the first foreground object. Said another way, all we need is an accurate enough estimate of all but $r_\new$ directions of the initial subspace, and an assumption of small eigenvalues for sometime ($d$ frames) along the directions for which we do not have an accurate enough estimate (or do not have an estimate).%
%??
%Another thing to point out is this. If $t_\train$ is large enough and the data truly is noise-free (if $\lt$'s are continuous random vectors, then $t_\train \ge r_0$ will suffice almost surely), then the estimate of $\Span(P_{t_\train}) = \Span(P_{t_0})$ obtained from training data will actually be exact, i.e. we will have $\mathrm{dif}( \Phat_{\ttrain}, \bm{P}_{\ttrain} ) = 0$. The theorem assumption that this error be below $r_0 \zeta$ allows for the initial training data to in fact be noisy.

%??? delete since give this in theorem itself
%Consider Model \ref{sbyrho} on $\T_t$. This is one example of the set of sufficient conditions on changes in $\T_t$. Our result is proved using Model \ref{general_model} on $\T_t$ which is more general. %As shown in Lemma \ref{spc_case}, Model \ref{sbyrho} is one special case of it. Some other special cases are described in Section \ref{extensions}.

{\em Denseness assumption. }
Consider the denseness assumption.
Define the (un)denseness coefficient as follows.
\begin{definition}
For a basis matrix $\bm{P}$, define $\ds\kappa_s(\bm{P}) := \max_{|\mathcal{T}| \leq s} \| {\I_{\mathcal{T}}}' \bm{P} \|_2$.
\end{definition}
Notice that left hand side in \eqref{candes_dense} is $\left[\kappa_1(\bm{P})\right]^2$.  Using the triangle inequality, it is easy to show that $\kappa_s(\bm{P}) \leq \sqrt{s}\kappa_1(\bm{P})$ \cite{ReProCS_IT}. Therefore, using the fact that for a basis matrix $[\P_1 \ \P_2]$, $(\kappa_s([\P_1 \ \P_2]))^2 \le (\kappa_s(\P_1))^2 + (\kappa_s(\P_2))^2$ (see proof of the first statement of Lemma \ref{RIC_bnd} in Appendix \ref{pf_cslem}), the denseness assumptions of Theorem \ref{thm1} imply that
\beq \label{kappa_dense}
\kappa_{s,*} :=\kappa_{2s}(\bm{P}_{t_J}) \leq 0.3 \quad \text{ and } \quad \kappa_{s,\rmnew}:= \max_j \kappa_{2s}(\bm{P}_{t_j,\new}) \leq 0.02.
\eeq
The proof of Theorem \ref{thm1} only uses \eqref{kappa_dense} for the denseness assumption.

The reason for defining $\kappa_{s}$ as above is the following lemma from \cite{ReProCS_IT}.
\begin{lem}[\cite{ReProCS_IT}]\label{kappadelta}
For a basis matrix $\bm{P}$, $\delta_s(\I - \bm{P}\bm{P}') = \left(\kappa_s(\bm{P})\right)^2$.
\end{lem}

{\color{blue}
{\em Lower bound on minimum nonzero entry of $x_t$ in the online RPCA result (Corollary \ref{thm1}). }
For online RPCA, notice that our result needs a lower bound on the minimum magnitude nonzero entry, $x_{\min}$, of the outlier vector $x_t$. This may seem counter-intuitive, since it means that outlier magnitudes need to be large enough for the proposed algorithm to work whereas one would expect that smaller corruptions are easier to deal with. This is actually true in our case as well, and the lower bound on minimum nonzero entry of $x_t$ is an artifact of trying to use a simpler model and a simpler proof approach. As we explain next, what we really need is that the corruptions either be small enough (to not affect subspace recovery too much) or be large enough (to be detectable).

\begin{corollary}[No lower bound on outlier magnitude]
Consider Algorithm \ref{reprocsdet_orpca}. Assume that $\mt$ satisfies \eqref{orpca_eq}. Assume that the following hold:
\ben
%\item  For a $q>0$, let $\T_{t,\mathrm{small}}(q)$ be the set of indices of the entries of $\xt$ whose magnitude is smaller than or equal to $q$ and let $\T_{t,\mathrm{large}}(q)$ be the set of entries of $\xt$ whose magnitude is strictly larger than $q$.
%Define $\epsilon_w(q): = \max_t \| (\xt)_{\T_{t,\mathrm{small}}(q)}\|_2$. Suppose that $\xt$ and $\lt$ are mutually independent; and there exists a $q>0$ so that
%\ben
%\item  $\min_{i \in \T_{t,\mathrm{large}}(q)}|(\xt)_i| > 14 \epsilon_w (q) + 14 ( \sqrt{r_\new}\gamma_{\new} + (\sqrt{r} + \sqrt{r_\new})\sqrt{\zeta} )$;  and
%\item $\epsilon_w(q)^2 \le 0.03\zeta  \lammin$.
%\een
%Let $q_0$ denote one such value of $q$ and let $\T_t \equiv \T_{t,\mathrm{large}}(q_0)$, i.e. let $\T_t$ be the set of indices of $\xt$ with entries whose magnitude is strictly larger than $q_0$.

\item Suppose that $\xt$ and $\lt$ are mutually independent; and there exists a partition of $\T_t$ into $\T_{t,\mathrm{large}},\T_{t,\mathrm{small}}$ so that
\bi
\item  $\min_{t} \min_{i \in \T_{t,\mathrm{large}}} |(\xt)_i| - 14 \max_t \|(\xt)_{\T_{t,\mathrm{small}}}\|_2    > 14(\sqrt{r_\new}\gamma_\new(d) + (\sqrt{r} + \sqrt{r_\new})\sqrt{\zeta}) $
\item and $\max_t \|(\xt)_{\T_{t,\mathrm{small}}}\|_2^2 \le 0.03 \zeta \lammin$
\ei
Let $\epsilon_w:=\max_t \|(\xt)_{\T_{t,\mathrm{small}}}\|_2$.

\item Algorithm parameters are set as %$\xi, \omega, \alpha$
$\xi = \epsilon_w + \sqrt{r_\new}\gamma_{\new} + (\sqrt{r} + \sqrt{r_\new})\sqrt{\zeta}$; $\omega = 7\xi$; $K = \left\lceil \frac{\log(0.16r_\new\zeta)}{\log(0.83)}\right\rceil$; and
$\alpha \ge {32\cdot100^2}\frac{\max\{16 , 1.2(\sqrt{\zeta} + \sqrt{r_\new}\gamma_\new + \epsilon_w)^4\}}{ \left(r_\new \zeta \lammin \right)^2} (\log(6(K+1)J) + 11\log n )$;

\item Everything else in Theorem \ref{thm1_mc} holds with $\T_t$ replaced by $\T_{t,\mathrm{large}}$.
\een
Then, with probability at least $1-n^{-10}$, the support set of the large entries of $x_t$, $\T_{t,\mathrm{large}}$, is exactly recovered at all times, $\|\xt - \xhatt\|_2 = \|\lt - \lhatt\|_2 \le  1.2 \left(\sqrt{\zeta} +  \sqrt{r_\new}\gamma_{\rmnew}  + \epsilon_w \right)$ and and all other conclusions of Theorem \ref{thm1_mc} hold.
%In the $\mt = \lt + \xt$ case, if the above replaces Model \ref{wt_model} and the condition ``$14 \xi \le \min_{t} \min_{i \in \T_t} |(x_t)_i|$" (in condition \ref{xmin_cond}) in Theorem \ref{thm1}; and if all its other conditions hold with $\epsilon_w$ replaced by $\epsilon_w (x_{\min})$ and with $\T_t$ replaced by the set of indices of $\xt$ with entries equal to or larger than $x_{\min}$; then all its conclusions hold.
\end{corollary}

\begin{proof}
The proof will follow in exactly the same fashion as the proof of the original theorem. We will just need to treat $(x_t)_{\T_{t,\mathrm{small}}}$ as an extra ``noise" term and use one of the following three facts at various places. Let $\E[.]$ denote expectation conditioned on accurate recovery so far and on $\T_t$ (this is formally defined in the proofs). We will use (a) $\|(\xt)_{\T_{t,\mathrm{small}}}\|_2 \le \epsilon_w \le \sqrt{0.03 \zeta \lammin}$; (b) $\E[\lt (\xt)_{\T_{t,\mathrm{small}}}{}']=0$ (this follows because $\lt$ is zero mean and $\lt$ and $\xt$ are independent (and hence $\lt$ and $\{\T_t,(\xt)_{\T_{t,\mathrm{small}}}\}$ are independent)); and (c) $\|\E[(\xt)_{\T_{t,\mathrm{small}}}(\xt)_{\T_{t,\mathrm{small}}}{}']\|_2 \le \max_t \|(x_t)_{\T_{t,\mathrm{small}}}(\xt)_{\T_{t,\mathrm{small}}}{}'\|_2 \le \epsilon_w^2 \le 0.03 \zeta \lammin$.
\end{proof}
}

\section{Discussion}\label{discussion}

\subsection{Discussion of the assumptions used}
In the previous section, we provide two related results, one for online matrix completion (MC) and the second for online robust PCA (RPCA). The result for online RPCA can also be interpreted as a result for online sparse matrix recovery in (potentially) large but structured noise $\lt$. Notice that our result does not require an upper bound on $\lambda^+$ (the maximum eigenvalue of $\cov(\lt)$ at any time) or on $\gamma$ (the bound on the maximum magnitude of any entry of $\P_t'\lt$ for any time $t$). Both these parameters are only used to select $\zeta$, which in turn governs the value of $K$ and $\alpha$ and hence governs the required delay between subspace change times.

Our results require accurate initial subspace knowledge. As explained earlier, for video analytics, this corresponds to requiring an initial short sequence of background-only video frames whose subspace can be estimated via SVD (followed by using a singular value threshold to retain a certain number of top left singular vectors). Alternatively if an initial short sequence of the video data satisfies the assumptions required by a batch method such as PCP (for RPCA) and NNM (for MC), that can be used to estimate the low-rank part, followed by SVD to get the column subspace. For online MC, another alternative is to use the initialization techniques of GROUSE \cite{grouse} or PETRELS \cite{petrels} or to use the adaptive MC idea of \cite{aarti_adaptive_mc}.%; this is needed since the matrix formed by the background images is only approximately low-rank

In Model \ref{exp_model}, we are placing a slow increase assumption on the eigenvalues along the new directions, $\P_{t_j,\new}$, for the interval $[t_j, t_{j+1})$. Thus after $t_{j+1}$, the eigenvalues along $\bm{P}_{t_j,\new}$ can increase gradually or suddenly to any large value up to $\lambda^+$. In fact as explained above, our proof needs the slow increase to hold only for the first $d$ time instants after $t_j$, so, in fact, at any time after $t_j+d$, the eigenvalues along $\P_{t_j,\new}$ could increase to a large value.

%In Model \ref{exp_model}, we are placing a slow increase assumption on the eigenvalues along the new directions only for the interval $[t_j, t_{j}+d)$. %In fact we only need slow increase for the first $d$ time instants after the change. % can modify our requirement to say: we need slow increase of covariance matrix only for d time instants after a principal subspace change. in particular the case of all
%Thus after $t_{j}+d$, and surely after $t_{j+1}$, the eigenvalues along $\bm{P}_{t_j,\new}$ can increase gradually or suddenly to any large value up to $\lambda^+$.

Model \ref{sbyrho} on $\T_t$ is a practical model for moving foreground objects in video. We should point out that this model is one special case of the general set of conditions we need (Model \ref{general_model}). Some other special cases of it are discussed in Section \ref{extensions}.

% (model on $\T_t$ and denseness) and the denseness constant   (total number of outliers) (for a given $\mu$)

The model on $\T_t$ (Model \ref{sbyrho}) and the denseness condition of the theorem constrain $s$ and $s,r_0,r_\new, J$ respectively. Model \ref{sbyrho} requires $s \le \rho_2 n / \alpha$ for a constant $\rho_2$.
Using the expression for $\alpha$, it is easy to see that as long as $J \in \bigo(n)$, we have $\alpha\in\bigo(\log n)$ and so Model \ref{sbyrho} needs $s \in \bigo(\frac{n}{\log n})$. With $s \in \bigo(\frac{n}{\log n})$, the denseness condition will hold if $r_0 \in \bigo(\log n)$, $J \in \bigo(\log n)$ and $r_\new$ is a constant. This is one set of sufficient conditions that we allow on the rank-sparsity product.

\subsection{Comparison with the results for PCP and NNM}
Let $\bm{L}:=[\bm{\ell}_1, \bm{\ell}_2, \dots, \bm{\ell}_{t_{\max}}]$ and $\bm{S}:=[\bm{x}_1, \bm{x}_2, \dots, \bm{x}_{t_{\max}}]$. Let $r_{\text{mat}}:=\rank(\bm{L})$. Clearly $r_{\text{mat}} \le r_0 + J r_\new$ and the bound is tight. Let $s_{\text{mat}}:= t_{\max} s$ be a bound on the total number of missing entries of $\bm{L}$ or on the support size of the outliers' matrix $\bm{S}$. In terms of $r_{\text{mat}}$ and $s_{\text{mat}}$, what we need is $r_{\text{mat}} \in \bigo(\log n)$ and $s_{\text{mat}} \in \bigo(\frac{n t_{\max}}{\log n})$.  This is stronger than what the PCP result from \cite{rpca} or the result for NNM from \cite{Candes_Recht} need (e.g., the PCP result from \cite{rpca} allows $r_{\text{mat}} \in \bigo \left(\frac{n}{(\log n)^2}\right)$ while allowing $s_{\text{mat}} \in \bigo(n t_{\max})$), but is similar to what the PCP results from \cite{rpca2,rpca_zhang} need.

Other disadvantages of our result are as follows. (1) Our result needs accurate initial subspace knowledge and slow subspace change of $\lt$. As explained earlier and in \cite[Fig. 6]{ReProCS_IT}, both of these are often practically valid for video analytics applications. Moreover, we also need the $\lt$'s to be zero mean and mutually independent over time. Zero mean is achieved by letting $\lt$ be the background image at time $t$ with an empirical `mean background image', computed using the training data, subtracted out. The independence assumption then models independent background variations around a common mean. As we explain in Section \ref{extensions}, this can be easily relaxed and we can get a result very similar to the current one under a first order autoregressive model on the $\lt$'s.
(2) Moreover, Algorithms \ref{reprocsdet} and \ref{reprocsdet_orpca} need multiple algorithm parameters to be appropriately set. The PCP or NNM results need this for none \cite{rpca,Candes_Recht} or at most one \cite{rpca2,rpca_zhang} algorithm parameter. (3) Thirdly, our result for online RPCA also needs a lower bound on $x_{\min}$ while the PCP results do not need this. (4) Moreover, even with this, we can only guarantee accurate recovery of $\lt$, while PCP or NNM guarantee exact recovery.

%compared to the PCP results \cite{rpca,rpca2,rpca_zhang} or the result from \cite{matcomp}  (or $\lt$ and $\xt$ in case or O-RPCA)

(1) The advantage of our work is that we analyze an online algorithm (ReProCS) that is faster and needs less storage compared with PCP or NNM.
It needs to store only a few $n \times \alpha$ or $n \times r_{\text{mat}}$ matrices, thus the storage complexity is $\bigo(n \log n)$ while that for PCP or NNM is $\bigo(n t_{\max})$. In general $t_{\max}$ can be much larger than $\log n$.
(2) Moreover, we do not need any assumption on the right singular vectors of $\bm{L}$ while all results for PCP or NNM do.
(3) Most importantly, our results allow highly correlated changes of the set of missing entries (or outliers). From the assumption on $\T_t$, it is easy to see that we allow the number of missing entries (or outliers) per row of $\bm{L}$ to be $\bigo(t_{\max})$ as long as the sets follow Model \ref{sbyrho}\footnote{In a period of length $\alpha$, the set $\T_t$ can occupy index $i$ for at most $\rrho\beta$ time instants, and this pattern is allowed to repeat every $\alpha$ time instants. So an index can be in the support for a total of $\rrho\beta\frac{t_{\max}}{\alpha}$ time instants and the model assumes $\rrho \beta \le \frac{0.01 \alpha}{\rrho}$ for a constant $\rrho$.}. %\footnote{if we split the total time into intervals of length $\alpha$, then the bound on $s$ ensures that an index is not re-visited; the bound on $\beta$ ensures that we need the support to change at least once every C alpha frames; also because object moves by s/rho and never comes back, it means a given index remains in the support for at most rho C alpha frames; this pattern can repeat again and so the maximum number of times for which an index could be in the support is C tmax.  }  as long as they follow the model
The PCP results from \cite{rpca2,rpca_zhang} need this number to be $\bigo(\frac{t_{\max}}{r_{\text{mat}}})$ which is stronger. %than what we need.
The PCP result from \cite{rpca} or the NNM result \cite{Candes_Recht} need an even stronger condition - they need the set $( \cup_{t=1}^{t_{\max}} \T_t )$ to be generated uniformly at random.%these sets to be independent

%need a denseness assumption on the right singular vectors or on $\max_{i,j}(UV')_{i,j}$ where $U$ and $V$ are the matrices of the left and right singular vectors

\subsection{Other results for online RPCA and online MC} \label{discussion_online}
%Now consider works that also use initial subspace knowledge.

Our online RPCA result improves upon the online RPCA results from our earlier work \cite{ReProCS_IT} for two reasons. First, the result of \cite{ReProCS_IT} was a {\em partial result} because it required a denseness assumption on $(\bm{I} - \bm{P}_{t_j,\rmnew}{\bm{P}_{t_j,\rmnew}}')\hat{\bm{P}}_{t}$ and $(\bm{I} - \hat{\bm{P}}_{t,*}\hat{\bm{P}}_{t,*}{}' -\hat{\bm{P}}_{t,\rmnew}\hat{\bm{P}}_{t,\rmnew}{}')\bm{P}_{t_j,\rmnew}$. Here $\Phat_{t,*}$ and $\Phat_{t,\new}$ are estimates computed by Algorithm \ref{reprocsdet_orpca}. Thus, the result depended on intermediate algorithm estimates satisfying certain properties. In this work, we remove this requirement and instead provide a {\em complete correctness result}. The extra assumption that we need is Model \ref{sbyrho} on $\T_t$ (or its generalization given in Model \ref{general_model} later). Secondly, we provide a correctness result for a ReProCS-based algorithm that detects subspace change automatically and also estimates the rank of the new subspace automatically. The algorithm studied in \cite{ReProCS_IT} required knowing $t_j$ and $r_{j,\new}$ exactly for each $j$.
Algorithms \ref{reprocsdet} and \ref{reprocsdet_orpca} in this work only require upper bounds on $r_\new$, $\gamma_\new$ and $J$ (these are needed to set the algorithm parameters - $\alpha$ and $K$ for Algorithm \ref{reprocsdet}, and also $\xi$ and $\omega$ for Algorithm \ref{reprocsdet_orpca}) and a small enough $\zeta$ (need bounds on $r$, $\lambda^+$ and $\gamma$ to set this). % and on $r$ and $\gamma$
A third minor advantage is that we also provide an algorithm and a result for online MC.

The proof of our results adapts the overall framework developed in \cite{ReProCS_IT}. The two important additions are: (a)  Model \ref{general_model} and Lemma \ref{blockdiag1} for it, and the way it is used in the proof of Lemma \ref{calHk}; and (b) the detection lemma (Lemma \ref{det}), the no false detection lemma (Lemma \ref{falsedet}) and the p-PCA lemma (Lemma \ref{pPCA}) and the lemmas used to prove these. (a) allows us to get a complete correctness result; (b) allows us to analyze an algorithm that does not use knowledge of $t_j$ or $r_{j,\new}$.% (Lemma \ref{Ak}, \ref{Akperp}, \ref{calHk}) the automatic estimation of $r_{j,\new}$ from

In \cite{rpca_stochatistic_optimization}, Feng et. al. propose a method for online RPCA and prove a partial result for their algorithm.  The approach is to reformulate the PCP program and use this reformulation to develop a recursive algorithm that converges asymptotically to the solution of PCP as long as the basis estimate $\hat{\bm{P}}_t$ is full rank at each time $t$. Since this result assumes something about the algorithm estimates, it is also only a {\em partial} result.
%Another somewhat related work is that of Feng et. al. \cite{OnlinePCA_ContaminatedData} on online PCA with contaminated data. This does not model the outlier as a sparse vector but defines anything that is far from the data subspace as an outlier. %The theorems in both papers only talk about asymptotically converging to the solution of the batch problem. %, whereas in the current work, we exploit slow subspace change to actually relax a key assumption needed by the batch methods (that of uniformly distributed random supports or of very frequent support change).  (so the above comparisons still apply)

Another recent work that uses knowledge of the initial subspace estimate and performs recovery in a piecewise batch fashion is modified-PCP \cite{zhan_pcp}. However, like PCP, the result for modified PCP also needs uniformly randomly generated support sets. Its advantage is that its assumption on the rank-sparsity product is weaker than that of PCP, and hence weaker than that needed by this work. A detailed simulation comparison between modified-PCP, ReProCS and PCP demonstrating both these things is available in \cite[Fig. 6]{zhan_pcp}.

Some other recent works that also study the online MC problem (defined differently from how we define it) include \cite{sequentialSVD}, Grassmanian Rank-One Update Subspace Estimation (GROUSE) \cite{grouse} and  Parallel Subspace Estimation and Tracking by Recursive Least Squares From Partial Observations (PETRELS) \cite{petrels}. We discuss the connection with \cite{sequentialSVD} in Section \ref{algosubsec}.
GROUSE is a first order stochastic gradient method. It uses rank-one updates to track the underlying subspace on the Grassmannian manifold. A result for its convergence to the local minimum of the cost function it optimizes is obtained in \cite{local_conv_grouse}. PETRELS is a second order stochastic gradient method. As explained in \cite{petrels}, in PETRELS, the low-dimensional subspace is tracked by minimizing a geometrically discounted sum of projection residuals on the observed entries at each time index. If missing entries are required then they can be reconstructed via least squares estimation. The subspace is updated recursively so that it is not necessary to retain historical data indefinitely. If the underlying subspace is fixed and the data stream is fully observed, then it is shown that the PETRELS estimate converges to the true subspace. In general, it always converges to the stationary point of  the cost function it optimizes \cite{petrels}.
The advantage of PETRELS and GROUSE is that they do not need initial subspace knowledge. For our algorithms, when the initial subspace knowledge is not available or initial complete and outlier-free data is not available, we can also use the PETRELS or GROUSE ideas for initialization.

\section{Automatic ReProCS Algorithms for Online MC and Online RPCA and Why They Work}\label{algosubsec}
In this section, we first introduce the automatic ReProCS based algorithm for online MC and explain why it works (this also provides the key idea why the proof of our main result would go through). Next, we do the same thing for the online RPCA algorithm. In the last two subsections (Sections \ref{insight} and \ref{outline}), we explain the key insight used by our proof and give the proof outline. %The actual proof is given in Sections \ref{pf_thm} and \ref{3_pfs}.

\subsection{Automatic ReProCS for Online MC (Algorithm \ref{reprocsdet})}
The model on $\mt$ from (\ref{omc_eq}) is a special case of that from (\ref{orpca_eq}) with $\xt = - \I_{{\T}_t}\I_{{\T}_t}{}' \lt$ and with the support of $\xt$, $\T_t$ known \cite{rpca}. Thus, we can use a simplification of the ReProCS idea for online RPCA \cite{ReProCS_IT} to also solve the online MC problem.

%The ReProCS algorithm that we study in this work is a significant improvement over the one studied in \cite{ReProCS_IT}. It does not assume knowledge of subspace change times, $t_j$, or of number of new directions, $r_{j,\new}$. Both are estimated automatically.
%
%The reprocs algorithm involves two parts (a) projected sparse recovery followed by recovering $\lt$ and (b) projection-PCA.
Algorithm \ref{reprocsdet} proceeds as follows.
Let $\Phat_{t-1}$ denote the basis matrix for the estimate of the subspace where $\l_{t-1}$ lies. If it is an accurate estimate,  because of ``slow subspace change", projecting the measurement $\mt = \xt + \lt$ onto its orthogonal complement will nullify most of $\lt$. Specifically, we compute $\bm{y}_t:= \bm{\Phi}_t \mt$ where $\bm{\Phi}_t:= \I - \Phat_{t-1}\Phat_{t-1}{}'$.  Thus, $\bm{y}_t$ can be rewritten as
\[
\bm{y}_t = \bm{\Phi}_t \xt + \bm{b}_t \ \text{where} \ \bm{b}_t:= \bm{\Phi}_t \lt
\]
and it can be argued that $\|\bm{b}_t\|_2$ is small.
Since the support of $\xt$, $\T_t$, is known, we can simply recover its nonzero entries by least squares (LS) estimation, i.e. we get $\xhatt = \I_{\T_t} (\bm{\Phi}_t)_{\T_t}{}^\dag \bm{y}_t$ and then get an estimate of $\lt$ as $\lhat_t = \mt - \xhatt$. The above approach of recovering $\lt$ is equivalent to that used by Brand in \cite{sequentialSVD}, there they recover $\lt$ as an LS estimate of $\Phat \Phat' \lt \approx \lt$.

Let $\et: = \lt - \lhat_t$. With the above, it is easy to see that% $\et$ satisfies
$$\et = \I_{\T_t} (\bm{\Phi}_t)_{\T_t}{}^\dag \bm{b}_t = \I_{\T_t} [(\bm{\Phi}_t)_{\T_t}{}'(\bm{\Phi}_t)_{\T_t}]^{-1} \I_{\T_t}{}' \bm{\Phi} \lt.$$
Using the denseness assumption, it can be argued that the RIC of $\bm{\Phi}_t$ will be small (see Lemma \ref{kappadelta}). Under the theorem's assumptions, and conditioned on accurate recovery so far, we can bound it by 0.14. Thus, $\|(\bm{\Phi}_t)_{\T_t}{}'(\bm{\Phi}_t)_{\T_t}^{-1}\|_2 \le 1/(1-0.14) < 1.2$ and so $\|\et\|_2 \le 1.2 \|\bm{b}_t\|_2$, i.e. it is small too (see Lemma \ref{cslem}).

{\em Projection-PCA (p-PCA). }
The next step is to use a modification of standard PCA called projection-PCA (p-PCA), to update the subspace estimate.  The reason we need p-PCA is this. Let $\sum_t$ denote a sum over an $\alpha$ length time interval. %$t \in [\that_j+(k-1)\alpha+1, \that_j+k\alpha]$.
In our problem, the error, $\et$, in the observation/estimate of $\lt$, $\lhat_t$, is correlated with $\lt$. Because of this, the dominant terms in the perturbation seen by standard PCA,  $\frac{1}{\alpha} \sum_t \lhatt \lhatt{}' - \frac{1}{\alpha} \sum_t \lt {\lt}'$,
are $\frac{1}{\alpha} \sum_t \lt {\et}'$ and its transpose\footnote{When $\lt$ and $\et$ are uncorrelated and one of them is zero mean, it can be argued by law of large numbers that, whp, these two terms will be close to zero and $\frac{1}{\alpha} \sum_t \et {\et}'$ will be the dominant term. %Thus, if $\cov(\et)$ is small, then $\frac{1}{\alpha} \sum_t \et \et'$, and hence the perturbation, will be small too.
}.
Thus, when the condition number of $\cov(\lt)$ is large, it becomes difficult to argue that the perturbation will be small compared to the smallest eigenvalue of $\cov(\lt)$. With a large perturbation, either the $\sin \theta$ theorem \cite{davis_kahan} (that bounds the subspace error between the eigenvectors of the true and estimated sample covariance matrices) cannot be applied or it gives a useless bound.

Our proposed approach, projection-PCA (p-PCA) addresses the above issue as follows. At $t=t_j$, let $\bm{P}_* := \bm{P}_{t_{j-1}}$, $\P_{\new}:=\P_{t_j,\new}$, and suppose that the subspace $\Span(\bm{P}_{*})$ has been accurately recovered, i.e. we have $\Phat_*$ so that $\mathrm{dif}(\Phat_*, \bm{P}_*) \ll 1$. Then at a time at or after $t_j+\alpha$ if we project the $\alpha$ previous $\lhat_t$'s perpendicular to $\Phat_*$, we will considerably reduce the perturbation seen by the PCA step. We detect subspace change by checking if the maximum singular value of the matrix formed by these projected $\lhat_t$'s is above a threshold. Denote the time at which change is detected by $\that_j$. After $\that_j$ we use SVD on $K$ different sets of $\alpha$ frames of the projected $\lhat_t$'s to get improved estimates of the new subspace $\Span(\bm{P}_{\new})$ in each iteration. To be precise, we get the $k$-th estimate, $\Phat_{\new,k}$, as the left singular vectors of $(\I - \Phat_* \Phat_*{}')[\lhat_{\that_j+(k-1)\alpha+1}, \dots , \lhat_{\that_j+ k\alpha}]$ with singular values above a threshold. After each p-PCA step, we update $\Phat_t$ as $\Phat_t = [\Phat_* \ \textit{} \Phat_{\new,k}]$. Finally at time $t = \that_j+K\alpha$, we update $\Phat_*$ as $[\Phat_{*} \ \Phat_{\new,K}]$.

In the subspace update step, Algorithm \ref{reprocsdet} toggles between the ``detect" phase and the ``ppca" phase. It starts in the ``detect" phase. When a subspace change is detected, i.e. at $t= \that_j$ it enters the ``ppca" phase. After $K$ iterations of p-PCA, i.e. at $t=\that_j+K\alpha+1$, the new subspace has been accurately estimated and this time it enters the ``detect" phase again.

%noise has a 2-norm that only depends on the projection of $\lt$ along the new directions (which is small for this interval by assumption)
{\em Why p-PCA works. }
The reason p-PCA works is as follows. Before the first p-PCA step, i.e. for $t \in [t_j, \that_j+\alpha)$, $\Phat_t = \Phat_*$ and thus the noise seen by the projected sparse recovery step, $\bm{b}_t = \bm{\Phi} \lt = (\I - \Phat_* \Phat_*{}') \lt$, will be the largest. Hence the error $\et$ will also be the largest for the $\lhatt$'s used for the first p-PCA step.
 However because of the projection perpendicular to $\Phat_*$ and slow subspace change, even this error is not too large. Because of this and because $\et$ is sparse and supported on $\T_t$ and $\T_t$ follows Model \ref{sbyrho}, we can argue that $\Phat_{\new,1}$ is a good estimate, i.e. $\mathrm{dif}([\Phat_* \ \Phat_{\new,1}], \P_\new) \le 0.2 < 1$.
After the first p-PCA step, $\Phat_t = [\Phat_* \ \Phat_{\new,1}]$ and this will reduce $\bm{b}_t$ and hence $\et$ for the $\lhatt$'s in the next $\alpha$ frames. This and the sparseness of $\et$, in turn, will mean that the perturbation seen by the second p-PCA step will be smaller and so $\Phat_{\new,2}$ will be a more accurate estimate of $\Span(\bm{P}_\new)$ than $\Phat_{\new,1}$.  This is done $K$ times with $K$ chosen so that $\mathrm{dif}([\Phat_* \ \Phat_{\new,K}],\bm{P}_{\new}) \leq r_{\new}\zeta$.
By the theorem assumptions, and because we can show $t_j \le \that_j < t_j+2\alpha$ (we explain this below), it is clear that $t_{j+1} > \that_j + K \alpha$. Thus, the new subspace added at $t_j$ is accurately estimated before the next change time $t_{j+1}$.

{\em Why $\that_j$ are correctly detected. }  As explained above,  we detect subspace changes by comparing the eigenvalues of $(\I-\hat{\bm{P}}_*\hat{\bm{P}}_*{}')\frac{1}{\alpha}\sum_{t}\lhat_t\lhat_t{}'(\I-\hat{\bm{P}}_*\hat{\bm{P}}_*{}')$ to a chosen threshold at every $t=u\alpha$ for $u=1,2,\dots,\llfloor\frac{t_{\max}}{\alpha}\rrfloor$ when the algorithm is in the ``detect" phase. In order to correctly detect $\that_j$, the algorithm first must not falsely detect new directions when none are present and it must detect subspace change within a short delay after it has occurred.
The former will occur because conditioned on accurate recovery of the current subspace, $(\I-\hat{\bm{P}}_*\hat{\bm{P}}_*{}')\frac{1}{\alpha}\sum_{t}\lhat_t\lhat_t{}'(\I-\hat{\bm{P}}_*\hat{\bm{P}}_*{}')$ will have very small eigenvalues when no new directions are present. If the recovery were exact and no new directions present, this matrix would be zero. In our case, the recovery is only accurate and so we show that all eigenvalues of this matrix will be below the chosen threshold (see Lemma \ref{falsedet}).
%
%When detecting subspace changes, the algorithm checks for new directions every $\alpha$ time instants. At $t=u\alpha$, it uses $(\I - \Phat_{*}\Phat_{*}{}')[\lhat_{(u-1)\alpha+1},\dots,\lhat_{u\alpha}]$ and computes its largest singular value.
%Consider the $j$-th subspace change time, i.e. $t=t_j$.
Next consider detection of the subspace change after it has occurred. When $u = u_j:= \llceil \frac{t_j}{\alpha} \rrceil$, i.e. when $t_j$ is in the interval $\big((u-1)\alpha+1,u\alpha\big]$, not all of the $\lt$'s in this interval will contain new directions.  Thus, depending on where in the interval $t_j$ lies, the algorithm may or may not detect the subspace change. However, in the {\em next} interval, $[u_j\alpha+1,(u_j+1)\alpha]$, all of the $\lt$'s will contain new directions, and we can prove that the subspace change will be detected w.h.p. (see Lemma \ref{det}). Thus, w.h.p., either $\hat{t}_j = u_j \alpha$, or $\hat{t}_j = (u_j+1)\alpha$.
Thus, we will be able to show that $t_j \le \that_j \le t_j+2\alpha$ w.h.p..

A visual description of Algorithm \ref{reprocsdet} is shown in Fig. \ref{algofig}. This figure uses Definition \ref{def_Phat_starnew}.

%
%{\em Subspace change detection (why $t_j \le \that_j < t_j+2\alpha$). }
%When detecting subspace changes, the algorithm checks for new directions every $\alpha$ time instants.
%At $t=u\alpha$, it uses $(\I - \Phat_{*}\Phat_{*}{}')[\lhat_{(u-1)\alpha+1},\dots,\lhat_{u\alpha}]$ and computes its largest singular value.  When $u = u_j:= \lceil \frac{t_j}{\alpha} \rceil$, i.e. when $t_j$ is in the interval $\big((u-1)\alpha+1,u\alpha\big]$, not all of the $\lt$'s in this interval will contain new directions.  Thus, depending on where in the interval $t_j$ lies, the algorithm may or may not detect the subspace change. However, in the {\em next} interval, $[u_j\alpha+1,(u_j+1)\alpha]$, all of the $\lt$'s will contain new directions, and we can prove that the subspace change will be detected whp. Thus, whp, either $\hat{t}_j = u_j \alpha$, or $\hat{t}_j = (u_j+1)\alpha$.
%
%%{\em Estimating number of new directions. } At each p-PCA step the algorithm detects one new direction for every singular value above the threshold $\sqrt\frac{\lambda^-}{2}$.  We can show that this number will be the same for all $K$ projection PCA steps and will in fact be equal to the true value: $r_{j,\new}$. %The reason is that the first $r_{j,\new}$ singular values of the projected data matrix

%$\lammin$ is defined Definition \ref{lammin_def}).
\begin{algorithm}[t]
\caption{ReProCS for Online MC}\label{reprocsdet}
{\em Parameters}: $\alpha$, $K$, {\em Inputs}: $\Phat_{t_\train}$, $\lammin$, $\bm{m}_t$ for each $t$,  {\em Output}: $\Lhat_t$, $\Phat_{t}$, $\hat{t}_{\jhat}$, $\hat{r}_{\jhat,\new,k}$

Let $\thresh=\frac{\lammin}{2}$ (this is the eigenvalue threshold that will be used to detect subspace change).

%Set $\Phat_{t,*} \leftarrow \Phat_{0}$,
Set $\Phat_{t,*} \leftarrow \Phat_{t_\train}$,
$\Phat_{t,\new} \leftarrow [.]$,
$\jhat \leftarrow 0$,  $\mathrm{phase} \leftarrow \mathrm{detect}$.

For every $t > t_\train$, do the following:
\bi
\item \label{othoproj} Compute $\bm{y}_t \leftarrow \bm{\Phi}_{t} \bm{m}_t$ where $\bm{\Phi}_{t} \leftarrow \bm{I} - \Phat_{t-1} \Phat_{t-1}{}'$
\item Estimate $\bm{\ell}_t$:
$\hat{\bm{\ell}}_t \leftarrow \bm{m}_t - \I_{\T_t} ((\bm{\Phi}_t)_{{\mathcal{T}}_t})^{\dag} \bm{y}_t $ %where $\bm{\Phi}_t:= I - \Phat_{t-1} \Phat_{t-1}{}'$

\item If $t \mod \alpha \neq 0$ then $\Phat_{t,*} \leftarrow \Phat_{t-1,*}$, $\Phat_{t,\new} \leftarrow \Phat_{t-1,\new}$, $\Phat_{t} \leftarrow [\Phat_{t,*} \ \Phat_{t,\new}]$

\item If $t \mod \alpha = 0$  then {\em detection or projection PCA}
\\
If  $\mathrm{phase} = \mathrm{detect}$ then
\begin{enumerate}

\item \label{detect} Set $ u = \frac{t}{\alpha}$ and compute
$\bm{\mathcal{D}}_u = (\I - \Phat_{u\alpha-1,*} \Phat_{u\alpha-1,*}{}') [\lhat_{(u-1)\alpha+1}, \dots \lhat_{u\alpha}]$
%$\bm{\mathcal{D}}_u = (\I - \Phat_{u\alpha-1,*} \Phat_{u\alpha-1,*}{}') \left(\frac{1}{\alpha}\sum_{\tau=(u-1)\alpha+1}^{u\alpha} \lhat_\tau \lhat_\tau{}'\right) (\I - \Phat_{u\alpha-1,*} \Phat_{u\alpha-1,*}{}')$

\item $\Phat_{t,*} \leftarrow \Phat_{t-1,*}$, $\Phat_{t,\new} \leftarrow \Phat_{t-1,\new}$, $\Phat_{t} \leftarrow [\Phat_{t,*} \ \Phat_{t,\new}]$ %$\Phat_{t} \leftarrow \Phat_{t-1}$

\item \label{change} If $\lambda_{\max} (\frac{1}{\alpha} \bm{\mathcal{D}}_u \bm{\mathcal{D}}_u{}'  ) \ge \thresh$  then \\
 $\mathrm{phase} \leftarrow \mathrm{ppca}$, $\jhat \leftarrow \jhat+1$, $k \leftarrow 0$, $\that_{\jhat} = t$

\end{enumerate}
Else  ($\mathrm{phase} = \mathrm{ppca}$) then
\begin{enumerate}
\item Set $ u = \frac{t}{\alpha}$ and compute
$\bm{\mathcal{D}}_u = (\I - \Phat_{u\alpha-1,*} \Phat_{u\alpha-1,*}{}') [\lhat_{(u-1)\alpha+1}, \dots \lhat_{u\alpha}]$
%$\bm{\mathcal{D}}_u = (\I - \Phat_{u\alpha-1,*} \Phat_{u\alpha-1,*}{}') \left(\frac{1}{\alpha}\sum_{\tau=(u-1)\alpha+1}^{u\alpha} \lhat_\tau \lhat_\tau{}'\right) (\I - \Phat_{u\alpha-1,*} \Phat_{u\alpha-1,*}{}')$

\item\label{PCA} $\Phat_{t,\new} \leftarrow \text{eigenvectors}\left(\frac{1}{\alpha} \bm{\mathcal{D}}_u \bm{\mathcal{D}}_u{}',\thresh\right)$,  $\Phat_{t,*} \leftarrow \Phat_{t-1,*}$, $\Phat_{t} \leftarrow [\Phat_{t,*} \ \Phat_{t,\new}]$

\item  $k \leftarrow k+1$, set $\hat{r}_{j,\new,k} = \rank(\Phat_{t,\new})$

\item If $k = K$, then \\
$\mathrm{phase} \leftarrow \mathrm{detect}$, $\Phat_{t,*} \leftarrow \Phat_t$, $\Phat_{t,\new} \leftarrow [.]$

\end{enumerate}
\ei
%The function $\text{left-singular-vectors}(\M_u,\thresh)$ returns a basis matrix for the span of all left singular vectors of $\M_u$ with singular values above $\thresh$.
$\mathrm{eigenvectors}(\bm{\mathcal{M}},\thresh)$ returns a basis matrix for the span of all eigenvectors whose eigenvalue is above $\thresh$.
\end{algorithm}

\subsection{Automatic ReProCS for online RPCA (Algorithm \ref{reprocsdet_orpca})}\label{Problem Definition and Assumptions_rPCA}
For online RPCA the only difference is that the support for $\xt$, $\T_t$, is not known. Hence we first recover $\xt$ by ell-1 minimization (or any other sparse recovery method) and then estimate its support by thresholding. %We use the estimated support to compute an LS estimate of $\xt$.
The rest of the steps remain the same as above.

%The algorithm for online robust PCA contains two extra steps that are not needed in the algorithm for online matrix completion.  This stems from the fact that in the matrix completion problem, the locations of the missing entries are {\em known}, while in robust PCA, the locations of the corrupted data points are {\em unknown}.
%
%To recover the location of the corrupted data points (the support of $\xt$) the algorithm performs a compressed sensing step followed by thresholding.  This starts by projecting the measurement $\mt$ perpendicular to the current estimate of the subspace where the $\lt$'s lie.  If the subspace estimate is accurate, this will nullify most of $\lt$.  Because the projection matrix $\bm{\Phi}_{t-1}$ is rank deficient and $\xt$ is sparse, the problem of recovering $\xt$ is one of sparse recovery in small noise.  Here the noise is any part of $\lt$ that was not annihilated by $\bm{\Phi}_{t-1}$.  With parameters properly set, thresholding on the recovered $\hat{\bm{x}}_{t,cs}$ will accurately recover the support of $\xt$ if 1) the norm of the noise is small enough and 2) the minimum absolute value of a non-zero entry of $\xt$ is large enough.

\subsection{Key Insight for the Proof} \label{insight}
% $\lambda_\new^-$ is the smallest eigenvalue along the new directions for the first $d$ frames and
The argument given while explaining why p-PCA works in Section \ref{algosubsec} can be formalized to show that, w.h.p., a subspace change is detected only after a change has occurred and within $2\alpha$ frames of the change; and that the subspace recovery error, $\SE_t$, will decay roughly exponentially with each p-PCA iteration and become small enough after $K$ iterations. To do this we will use the $\sin \theta$ theorem \cite{davis_kahan} (Lemma \ref{zetakbnd}) followed by the matrix Hoeffding inequality \cite{tail_bound} (Lemmas \ref{hoeffding_nonzero}, \ref{hoeffding_rec})) to get high probability bounds on each of the terms in the subspace error bound obtained by the $\sin \theta$ theorem.

While applying the matrix Hoeffding inequality, we need to use the following key insight about the structure of $\E[\frac{1}{\alpha}\sum_t (\I - \Phat_* \Phat_*{}') \lt {\et}']$. This matrix is the dominant term in the perturbation seen by the $k$-th p-PCA step.
Here $\E[.]$ denotes expectation conditioned on accurate subspace recovery so far and $\sum_t$ denotes the sum over $t \in [\that_j+(k-1)\alpha+1, \that_j+k\alpha]$.
The model on $\T_t$ and the fact that $\et$ is supported on $\T_t$ can be used to show that this matrix can be written as the product of a full matrix and a block-banded matrix: for example when $\rho=1$, the block-banded matrix will be block-diagonal, when $\rho=2$, it will be block-tridiagonal, and so on. Also, $\E[\frac{1}{\alpha} \sum_t \et {\et}']$ will be a block banded matrix. The 2-norm of a block banded matrix is bounded by the maximum norm of any block times the number of bands in it and hence is much smaller than that of a general full matrix.

The lemma that exploits the structure of a block-banded matrix generated due to the model on $\T_t$ is Lemma \ref{blockdiag1} given in Sec \ref{gen_supch_sec}. This lemma is used to bound $\E[\frac{1}{\alpha}\sum_t (\I - \Phat_* \Phat_*{}') \lt {\et}']$ and $\E[\frac{1}{\alpha} \sum_t \et {\et}']$ in the proof of Lemma \ref{calHk} in Section \ref{3_pfs}.

\begin{algorithm}[t]
\caption{ReProCS for Online RPCA}\label{reprocsdet_orpca}
{\em Parameters}: $\alpha$, $K$, $\xi$, $\omega$, {\em Inputs}: $\Phat_{t_\train}$, $\lammin$, $\bm{m}_t$ for each $t$,  {\em Output}: $\Lhat_t$, $\Phat_{t}$, $\hat{t}_{\jhat}$

Let $\thresh=\frac{\lammin}{2}$.
Set $\Phat_{t,*} \leftarrow \Phat_{t_\train}$, %  $\Phat_{t,*} \leftarrow \Phat_{0}$,
$\Phat_{t,\new} \leftarrow [.]$,
$\jhat \leftarrow 0$,  $\mathrm{phase} \leftarrow \mathrm{detect}$.

For every $t > t_\train$, do the following:
\bi
\item Estimate $\T_t$ (the support of the outlier vector $\xt$) and $\xt$.
\ben
\item \label{othoproj} compute $\bm{y}_t \leftarrow \bm{\Phi}_{t} \bm{m}_t$ where $\bm{\Phi}_{t} \leftarrow \bm{I} - \Phat_{t-1} \Phat_{t-1}{}'$

\item \label{Shatcs} solve $\min_{\bm{x}} \|\bm{x}\|_1 \ s.t. \ \|\bm{y}_t - \bm{\Phi}_{t} \bm{x}\|_2 \leq \xi$ and let $\hat{\bm{x}}_{t,\cs}$ denote its solution
\item \label{That} compute $\hat{\mathcal{T}}_t = \{i: \ |(\hat{\bm{x}}_{t,\cs})_i| > \omega\}$
\item \label{LS} LS estimate of $\bm{x}_t$: compute $\hat{\bm{x}}_t= \I_{\hat{\mathcal{T}}_t} ((\bm{\Phi}_t)_{\hat{\mathcal{T}}_t})^{\dag} \bm{y}_t$
\een

\item Use all steps of Algorithm \ref{reprocsdet} with $\T_t \leftarrow \hat{\T}_t$.
\ei
\end{algorithm}

\subsection{Proof Outline} \label{outline}
We will only prove Theorem \ref{thm1}. Theorem \ref{thm1_mc} follows as a corollary of Theorem \ref{thm1} because of the following reasons. (1) Algorithm \ref{reprocsdet} does not compute $\xhatt$ or its support $\That_t$. For the matrix completion problem, $\T_t$ is given. Thus it does not use the parameters $\xi$ (which is the noise bound in the ell-1 minimization step) and $\omega$ (which is the support estimation threshold). (2) The bound on $x_{\min}$ and the values of the parameters $\xi$ and $\omega$ are only used in the proof of Lemma \ref{cslem} to show exact support recovery, i.e $\That_t=\T_t$. Since for matrix completion $\T_t$ is given, Theorem \ref{thm1_mc} does not need need the lower bound on $x_{\min}$.

The proof of Theorem \ref{thm1} is given in Sections \ref{pf_thm} and \ref{3_pfs}. Before this, in the next section (Section \ref{gen_supch_sec}) we give the most general model on changes in the missing/outlier entries' set $\T_t$, Model \ref{general_model}, and we show that Model \ref{sbyrho} is a special case of this model. Next, we give a key lemma for sums of sparse matrices supported on rows and columns indexed by $\T_t$ satisfying this model (Lemma \ref{blockdiag1}). %There are various other practically relevant special cases of Model \ref{general_model}; some of these are described in Section \ref{extensions}.

Section \ref{pf_thm} begins with defining various quantities needed for the proof. Next, we state the main lemmas used to prove the theorem, followed by the theorem's proof. There is a main lemma associated with each of the three main tasks of the algorithm: 1) accurately recovering $\xt$ and hence $\lt$ at each time $t$ (Lemma \ref{cslem}), 2) detecting (subspace change) when and only when the subspace has changed, i.e. new directions have been added to the subspace (Lemmas \ref{det} and \ref{falsedet}), and 3) successfully estimating the dimension of the new subspace and updating its estimate by p-PCA (Lemma \ref{pPCA}). To maintain the flow of the argument, we defer the proofs of these lemmas to the end of the section or to the appendix.

The proofs of Lemmas \ref{Ak}, \ref{Akperp}, and \ref{calHk} that are used together to prove Lemmas \ref{det}, \ref{falsedet} and \ref{pPCA} are rather long and are given in section \ref{3_pfs}. The proof of Lemma \ref{calHk} uses Lemma \ref{blockdiag1} from Section \ref{gen_supch_sec}. %This lemma is where we use the model on how $\T_t$ changes. This lemma uses Model \ref{general_model}. As shown in Lemma \ref{spc_case}, Model \ref{sbyrho}) is a special case of this more general model.

\section{Most General Model on Changes in $\T_t$ and a Key Lemma}\label{gen_supch_sec}

\subsection{Most General Model on Changes in $\T_t$}
Here we give our most general model on how $\T_t$ (the set of missing entries or the support set of $\xt$) can change. What we need to prevent is $\T_t$ occupying the same indices for too many time instants in a given interval. If $\T_t$ does not change `enough' in a time interval of length $\alpha$, we will be unable to see enough entries of a given index of $\lt$ in order to be able to accurately fill in the missing ones.
 %the matrix $[\bm{x}_t \ \dots \ \bm{x}_{t+\alpha}]$  can become significantly rank deficient, and this causes an identifiability problem \cite{rpca}. Of course since we start with initial subspace knowledge and use slow subspace change, we need a weaker condition on the motion.
The following model quantifies `enough' for our purposes.  The number of time instants for which an index is part of $\T_t$ is determined both by how often this set changes, and by how much it moves when it changes. The latter is parameterized by $\rho$ which controls how much the set moves when it changes.  For example $\rho=1$ would require that distinct sets be disjoint, and $\rho=2$ would mean that at least half of the set is displaced whenever it changes.  The parameter $h^+ \in (0,1)$ represents the maximum fraction of time for which the set remains in a given area in a time interval of length $\alpha$.  The smaller $h^+$, the more frequently the set will need to change in order to satisfy the model. Our result requires a bound on the product $\rho^2 h^+$.  %The square in $\rho$ is a result of how we apply the model to bound the norm of certain matrices in the proof (see the proof of Lemma \ref{blockdiag1}).

\begin{sigmodel}\label{general_model}
Let $\rho$ be a positive integer. Split $[1, t_{\max}]$ into intervals of length $\alpha$. Use $\J_u:=[(u-1)\alpha+1,u\alpha]$ to denote the $u$-th interval.
For a given interval, $\J_u$, let $\mathcal{T}_{(i),u}$ for $i = 1,\dots, l_u$ be mutually disjoint subsets of $\{ 1, \dots, n\}$ and let
 $\J_{(i),u}, i=1,2, \dots, l_u$ be {\em a} partition\footnote{i.e. the $\J_{(i),u}$'s are mutually disjoint intervals and their union equals $\J_u$} of the interval $\J_u$ so that
\begin{equation}\label{union}
\text{  for all $t \in \J_{(i),u}$,   }
\T_t \subseteq \mathcal{T}_{(i),u} \cup \mathcal{T}_{(i+1),u} \cup \dots \cup \mathcal{T}_{(i+\rho-1),u}
\end{equation}
Define
\begin{align}\label{alphabyp}  %u=1,\dots,\lceil\frac{t_{\max}}{\alpha}\rceil
h_u\left(\alpha;\{\mathcal{T}_{(i),u}\}_{\substack{\\i=1,\dots,l_u}}, \{\J_{(i),u} \}_{\substack{\\i=1,\dots,l_u}} \right)
&:= \max_{i=1,2,\dots l_u}\big| \J_{(i),u} \big|
\end{align}
and define $h_u^*(\alpha)$ which takes the minimum over all choices of $\mathcal{T}_{(i),u}$ and over all choices of the partition $\J_{(i),u}$.
%
%\max_{u = 1,\dots, \lceil\frac{t_{\max}}{\alpha}\rceil}
% which are subsets of $\{1,2,\dots n\}$
\begin{align}\label{hstar}
h_u^*(\alpha) &:=
\min_{\substack{\text{\em all choices of mutually disjoint $\mathcal{T}_{(i),u}, i=1,2, \dots l_u$}  \\  \text{\em and all choices of mutually disjoint $\J_{(i),u},i=1,2,\dots l_u$} \\ \text{\em so that $\cup_{i=1}^{l_u} \J_{(i),u} = \J_u$ and \eqref{union} holds}  } }
h_u\left(\alpha;\{\mathcal{T}_{(i),u}\}_{\substack{\\i=1,\dots,l_u}}, \{\J_{(i),u} \}_{\substack{\\i=1,\dots,l_u}} \right)
\end{align}
Assume that $|\T_t|\le s$ and that for all $u = 1,\dots,\llceil\frac{t_{\max}}{\alpha}\rrceil$,
\[
h_u^*(\alpha)\leq h^+\alpha \ \text{ with }  h^+ \le \frac{0.01}{\rho^2}.
\]
\end{sigmodel}
In the above model, $h_u^*(\alpha)$ provides a bound on how long $\T_t$ remains in a given ``area", $\mathcal{T}_{(i),u} \cup \mathcal{T}_{(i+1),u} \cup \dots \cup \mathcal{T}_{(i+\rho-1),u}$ during the interval $\J_u$, for the best allocation of $\T_t$'s to a given ``area" and the best choice of the ``areas."

Notice that \eqref{union} can always be trivially satisfied by choosing $l_u =1$, $\mathcal{T}_{(1),u} = \{  1, \dots, n \}$ and $\J_{(1),u}=\J_u$, but this will give $h_u(\alpha;.) = \alpha$ and hence is not a good choice. This is why we take a minimum over all choices.
%This assumption ensures that the support of $\xt$ changes at least a certain amount during every interval $\J_u$.  The smaller $h^+$, the more the support is required to change.

\begin{lem}\label{spc_case}
Model \ref{sbyrho} is a special case of Model \ref{general_model} above with $h^+= \frac{\beta}{\alpha}$.% $\rho=\rrho$ and
\end{lem}
The proof is in Appendix \ref{pf_supch}.

Some other special cases of the above model are discussed in Section \ref{extensions}.

\subsection{A Key Lemma that uses Model \ref{general_model}}
\begin{lem}\label{blockdiag1}
Let $s_t = |\mathcal{T}_t|$. Consider a sequence of $s_t \times s_t$ symmetric positive-semidefinite matrices $\bm{A}_t$ such that
$\| \bm{A}_t\|_2 \leq \sigma^+$ for all $t$.  Assume that the $\mathcal{T}_t$ obey Model \ref{general_model}. %and the assumptions of Theorem \ref{general_cor}. ?? do you need any Theorem assumptions for this lemma ??
Let $\ds\bm{M} = \sum_{t\in\J_u} \bm{I}_{\mathcal{T}_t} \bm{A}_t {\bm{I}_{\mathcal{T}_t}}'$ be an $n \times n$ matrix ($\I$ is an $n\times n$ identity matrix).
Then
\begin{align*}
\|\bm{M}\|_2 &\leq    \rho^2h^+ \alpha\sigma^+  \leq 0.01\sigma^+  \alpha
\end{align*}
%?? do you need any Theorem assumptions for this lemma, I could not find any used - commented out that part ??
\end{lem}

\begin{proof}
We will first prove the lemma for the special case when $\rho=2$.  After this, we will show how to generalize the proof when $\rho>2$. For a given $u$, let $\T_{(i),u}$, $i=1,2,\dots l_u$, and correspondingly $\J_{(i),u}$ denote the best choices, i.e. the choices that attain the minimum values in the definition of $h_u^*(\alpha)$.

In the rest of the proof, we remove the subscript $u$ from $l_u$ and from $\T_{(i),u}$'s for ease of notation.  For simplicity of notation, we will let $\T_{(l+1),u}=\emptyset$.

For times $t \in \J_{(i),u}$, define $\bm{A}_{t,\full}$ to be $\bm{A}_t$ with rows and columns of zeros appropriately inserted so that
\begin{equation}\label{full}
\bm{I}_{\mathcal{T}_t} \bm{A}_t {\bm{I}_{\mathcal{T}_t}}' = \bm{I}_{\mathcal{T}_{(i)} \cup \mathcal{T}_{(i+1)}} \bm{A}_{t,\full} {\bm{I}_{\mathcal{T}_{(i)}\cup \mathcal{T}_{(i+1)}}}'.
\end{equation}
Such an $\bm{A}_{t,\full}$ exists because $\mathcal{T}_t \subseteq \mathcal{T}_{(i)} \cup \mathcal{T}_{(i+1)} $ for any $t \in \J_{(i),u}$.
Notice that
\begin{equation}\label{normequal}
\|\bm{A}_{t,\full}\|_2 = \|\bm{A}_{t}\|_2,
\end{equation}
because $\bm{A}_{t,\full}$ is permutation similar to
\[
\left[
\begin{array}{cc}
\bm{A}_t & 0 \\
0    & 0
\end{array}
\right].
\]

Since $\mathcal{T}_{(i)}$ and $\mathcal{T}_{(i+1)}$ are disjoint, we can, after permutation similarity, correspondingly partition $\bm{A}_{t,\full}$ as
\[
%\bm{A}_{t,\full} =
\left[
\begin{array}{ccc}
\bm{A}_{t,\full}^{(0,0)} & \bm{A}_{t,\full}^{(0,1)} \\
\bm{A}_{t,\full}^{(1,0)} & \bm{A}_{t,\full}^{(1,1)}
\end{array}
\right].
\]
for all $t \in \J_{(i),u}$.
Notice that because $\bm{A}_t$ is symmetric, $\bm{A}_{t,\full}^{(1,0)} = \big( \bm{A}_{t,\full}^{(0,1)} \big)'$.
Then,
\begin{align*}
\bm{M} &= \sum_{t\in\J_u} \bm{I}_{\mathcal{T}_t} \bm{A}_t {\bm{I}_{\mathcal{T}_t}}' \\
&= \sum_{i=1}^{l} \sum_{ t \in \J_{(i),u} }  \bm{I}_{\mathcal{T}_{(i)} \cup \mathcal{T}_{(i+1)}} \bm{A}_{t,\full} {\bm{I}_{\mathcal{T}_{(i)}\cup \mathcal{T}_{(i+1)}}}' \qquad \text{by \eqref{full} } \\
&= \sum_{i=1}^{l}\sum_{ t \in \J_{(i),u} }  [ \bm{I}_{\mathcal{T}_{(i)}} \ \bm{I}_{\mathcal{T}_{(i+1)}}] \left[
\begin{array}{ccc}
\bm{A}_{t,\full}^{(0,0)} & \bm{A}_{t,\full}^{(0,1)} \\
\bm{A}_{t,\full}^{(1,0)} & \bm{A}_{t,\full}^{(1,1)}
\end{array}
\right] \left[\begin{array}{c} {\bm{I}_{\mathcal{T}_{(i)}}}' \\ {\bm{I}_{\mathcal{T}_{(i+1)}}}' \end{array}\right] \\
&= \sum_{i=1}^{l} \sum_{ t \in \J_{(i),u} } \left[ \bm{I}_{\mathcal{T}_{(i)}}\bm{A}_{t,\full}^{(0,0)}{\bm{I}_{\mathcal{T}_{(i)}}}' + \bm{I}_{\mathcal{T}_{(i)}}\bm{A}_{t,\full}^{(0,1)}{\bm{I}_{\mathcal{T}_{(i+1)}}}' + \bm{I}_{\mathcal{T}_{(i+1)}}\bm{A}_{t,\full}^{(1,0)}{\bm{I}_{\mathcal{T}_{(i)}}}' + \bm{I}_{\mathcal{T}_{(i+1)}} \bm{A}_{t,\full}^{(1,1)} {\bm{I}_{\mathcal{T}_{(i+1)}}}' \right]  \\
&= \bm{I}_{\T_{(1)}} \left(\sum_{t \in \J_{(1),u}} \bm{A}_{t,\full}^{(0,0)} \right) {\bm{I}_{\T_{(1)}}}'  +
\sum_{i=2}^{l} \left[ \bm{I}_{\mathcal{T}_{(i)}}\left( \sum_{t \in \J_{(i-1),u}} \bm{A}_{t,\full}^{(1,1)} + \sum_{t \in \J_{(i),u}}  \bm{A}_{t,\full}^{(0,0)} \right) {\bm{I}_{\mathcal{T}_{(i)}}}' \right] +
\bm{I}_{\T_{(l)}} \left(\sum_{t \in \J_{(l),u}} \bm{A}_{t,\full}^{(1,1)} \right) {\bm{I}_{\T_{(l)}}}'\\
&\hspace{.5in}
+ \sum_{i=1}^{l-1} \left[ \bm{I}_{\mathcal{T}_{(i)}} \left(\sum_{t \in \J_{(i),u}} \bm{A}_{t,\full}^{(0,1)}\right) {\bm{I}_{\mathcal{T}_{(i+1)}}}'  +  \bm{I}_{\mathcal{T}_{(i+1)}} \left(\sum_{t \in \J_{(i),u}} \bm{A}_{t,\full}^{(1,0)}\right) {\bm{I}_{\mathcal{T}_{(i)}}}' \right]
\end{align*}

Because $\mathcal{T}_{(i)}$ and $\T_{(k)}$ are disjoint for $i\neq k$, $\bm{M}$ has a block tridiagonal structure (by a permutation similarity if necessary):
\begin{equation}\label{blockstruct}
%\bm{M} =
\left[
\begin{array}{cccc}
\bm{B}_{(1)} & \bm{C}_{(1)} & 0 & 0\\
{\bm{C}_{(1)}}' & \bm{B}_{(2)} & \ddots & 0 \\
0          & \ddots &     \ddots & \bm{C}_{(l-1)} \\
0 & 0 & {\bm{C}_{(l-1)}}' & \bm{B}_{(l)}
\end{array}
\right]
\end{equation}
where $\bm{B}_{(1)} = \sum_{t \in \J_{(1),u}} \bm{A}_{t,\full}^{(0,0)}$, $\bm{B}_{(l)} = \sum_{t \in \J_{(l),u}} \bm{A}_{t,\full}^{(1,1)}$,
\begin{equation}\label{doubledose}
\bm{B}_{(i)} = \sum_{t \in \J_{(i-1),u}} \bm{A}_{t,\full}^{(1,1)} + \sum_{t \in \J_{(i),u}} \bm{A}_{t,\full}^{(0,0)} \ \ \text{ for } i=2,3, \dots, l
\end{equation}
and
\begin{equation}\label{B(i)}
\bm{C}_{(i)} = \sum_{t \in \J_{(i),u}} \bm{A}_{t,\full}^{(0,1)}  \ \ \text{ for } i=1,2, \dots, l-1.
\end{equation}

Now we proceed to bound $\|\bm{M}\|_2$.
\begin{align*}
\|\bm{M}\|_2 & =
\left\| \begin{array}{cccc}
\bm{B}_{(1)} & \bm{C}_{(1)}  & 0 &0 \\
{\bm{C}_{(1)}}' & \ddots & \ddots & 0 \\
0     &  \ddots & \ddots & \bm{C}_{(l-1)} \\
0 & 0 & {\bm{C}_{(l-1)}}' & \bm{B}_{(l)}
\end{array}\right\|_2 \\
&\leq
\left\| \begin{array}{cccc}
\bm{B}_{(1)} & 0  & 0 &0 \\
0 & \ddots & 0 & 0 \\
0     &  0 & \ddots &0 \\
0 & 0 & 0 & \bm{B}_{(l)}
\end{array}\right\|_2
+
\left\| \begin{array}{cccc}
0 & \bm{C}_{(1)}  & 0 &0 \\
0 & 0 & \ddots & 0 \\
0     &  0 & 0 & \bm{C}_{(l-1)} \\
0 & 0 & 0 & 0
\end{array}\right\|_2
+
\left\| \begin{array}{cccc}
0 &0  & 0 &0 \\
{\bm{C}_{(1)}}' & 0 & 0 & 0 \\
0     &  \ddots & 0 &0 \\
0 & 0 & {\bm{C}_{(l-1)}}' & 0
\end{array}\right\|_2. \nn
\end{align*}

Call the middle matrix $\bm{C}$, and observe that $\bm{CC}'$ is block diagonal with blocks $\bm{C}_{(i)}{\bm{C}_{(i)}}'$.  So $\|\bm{C}\|_2 = \max_{i}\|\bm{C}_{(i)}\|_2$.
Therefore,
\begin{align*}
\|\bm{M}\|_2 &\leq \max_{i}\|\bm{B}_{(i)}\|_2 + 2\max_{i}\|\bm{C}_{(i)}\|_2 \\
& \le \max_{i} \bigg\| \sum_{t \in \J_{(i-1),u}} \bm{A}_{t,\full}^{(1,1)} + \sum_{t \in \J_{(i),u}} \bm{A}_{t,\full}^{(0,0)}\bigg\|_2 + 2 \max_{i} \bigg\| \sum_{t \in \J_{(i),u}} \bm{A}_{t,\full}^{(0,1)} \bigg\|_2  \qquad \text{ by \eqref{doubledose} and \eqref{B(i)}}\nn\\
%
%&\leq \max_{i} \left( \sum_{t \in \J_{(i-1),u}} \big\|  \bm{A}_{t,\full}^{(1,1)}\big\|_2 + \sum_{t \in \J_{(i),u}}\big\|  \bm{A}_{t,\full}^{(0,0)}\big\|_2 \right) + 2 \max_{i} \sum_{t \in \J_{(i),u}}\big\|  \bm{A}_{t,\full}^{(0,1)} \big\|_2 \nn \\
%
&\leq \max_{i} \left( \sum_{t \in \J_{(i-1),u}}\big\|  \bm{A}_{t}\big\|_2 + \sum_{t \in \J_{(i),u}}\big\|  \bm{A}_{t}\big\|_2 \right) + 2 \max_{i} \sum_{t \in \J_{(i),u}}\big\|  \bm{A}_{t} \big\|_2 \nn \qquad \text{ by \eqref{normequal}}\\
%
%&\leq \max_{i} \left( \sum_{t \in \J_{(i-1),u}}\sigma^{+} + \sum_{t \in \J_{(i),u}}\sigma^{+} \right) + 2 \max_{i} \sum_{t \in \J_{(i),u}} \sigma^{+} \nn \\
%
&\leq  (\sigma^{+}h_u^*(\alpha) + \sigma^{+}h_u^*(\alpha))  + 2 \sigma^{+}h_u^*(\alpha)  \leq {4}\sigma^+h^+\alpha
%\qquad \text{ by the definition of $h_u^*(\alpha)$ in \eqref{hstar} }\nn \\
%
\end{align*}
The third row used the fact that $\|\bm{A}_{t,\full}^{(\cdot,\cdot)}\|_2 \le  \|\bm{A}_{t,\full}\|_2 = \|\bm{A}_{t}\|_2$ for any sub-matrix of $\bm{A}_{t,\full}$.

This finishes the proof for the $\rho=2$ case. For this case, notice that there are 3 bands in \eqref{blockstruct} - the diagonal band and one band on each side of the diagonal one. When $\rho=3$, everything will follow analogously to the above; instead of 3 bands, there will be 5 bands in the definition of $\bm{M}$ and we will be able to bound its norm by
\begin{align*}
\|\bm{M}\|_2  \le  &
\max_{i} \left(\sum_{t \in \J_{(i-2),u}}\big\|  \bm{A}_{t}\big\|_2  + \sum_{t \in \J_{(i-1),u}}\big\|  \bm{A}_{t}\big\|_2 + \sum_{t \in \J_{(i),u}}\big\|  \bm{A}_{t}\big\|_2 \right) +
2 \max_{i} \left( \sum_{t \in \J_{(i-1),u}}\big\|  \bm{A}_{t} \big\|_2 +  \sum_{t \in \J_{(i),u}}\big\|  \bm{A}_{t} \big\|_2 \right)
\\
& + 2 \max_{i} \sum_{t \in \J_{(i),u}} \big\|  \bm{A}_{t} \big\|_2 \\
\le &  3 \sigma^{+} h_u^*(\alpha) + 2( 2 \sigma^{+}h_u^*(\alpha) + \sigma^{+}h_u^*(\alpha)) \le 9 \sigma^+ h^+ \alpha
\end{align*}
Proceeding this way, for a general $\rho$, there will be $1+ 2(\rho-1) = 2\rho-1$ bands. Any term in the central band will contain a summation of $\big\|  \bm{A}_{t}\big\|_2$ over $\rho$ sub-intervals $\J_{(i),u}$; any term in the first band away from the diagonal will contain this summation over $(\rho-1)$ sub-intervals; any term in the second band away from the diagonal will contain this summation over $(\rho-2)$ sub-intervals; and so on. %Since each sub-interval $\J_{(i),u}$ is of length at most $h_u(\alpha) \le h^+ \alpha$, we will be summing $\|A_t\|_2$ $(\rho + 2\sum_{i=1}^{\rho-1} i) = \rho^2$ times.
Thus, we will be summing the quantity $\sigma^+h^+\alpha$  a total of $(\rho + 2\sum_{i=1}^{\rho-1} i) = \rho^2$ times and so we will get $\|\bm{M}\|_2 \leq \rho^2\sigma^+h^+\alpha$.
%%For the $\rho=2$ case, each term in the diagonal band contains summations over $\rho=2$ sub-matrices of $\bm{A}_{t,\full}$. The first and only band away from the diagonal contains a summation over $\rho-1=1$ sub-matrix of $\bm{A}_{t,\full}$. Thus each of the $2\rho-1 =3$ bands contains summations over at most $\rho=2$ sub
%A tighter bound can be obtained by observing that, the first band away from the diagonal will contain part of matrices $\bm{A}_{t,\full}$ from $\rho-1$ different sets, and so on until the band $\rho-1$ away from the diagonal which will only contain part of $\bm{A}_{t,\full}$ from one set.  Stated another way, the matrices $\bm{B}_{(i)}$ above can be decomposed into $\rho$ sums, and the matrices $\bm{C}_{(i)}$ can be decomposed into $\rho-1$ sums (in the above $\rho=2$).  For each of the off diagonal bands, there are two bands (one above the diagonal and one below).  Thus there are $\rho + 2\sum_{i=1}^{\rho-1} i = \rho^2$ terms in the last line of the proof.  Thus $\|\bm{M}\|_2 \leq \rho^2\sigma^+h^+\alpha$.
\end{proof}

\section{Proof of Theorem \ref{thm1} and Theorem \ref{thm1_mc}} \label{pf_thm}
As explained in Section \ref{outline}, we will only prove  Theorem \ref{thm1}. Theorem \ref{thm1_mc} follows as an easy corollary.

\subsection{Definitions}

\begin{definition}
Define $\bm{e}_t$ to be the error made in estimating $\xt$ and $\lt$.  That is
\[
\et := \xhatt - \xt = \lt - \lhatt
\]
\end{definition}

\begin{definition}
Define the interval
\[
\mathcal{J}_u := [(u-1)\alpha+1,u\alpha].
\]
Also define $u_j$ to be the $u$ such that $t_j \in \mathcal{J}_u$.  That is
$$u_j := \llceil \frac{t_j}{\alpha}\rrceil.$$
For the purposes of describing events before the first subspace change, let $u_0 := 0$.
Also define $$\hat{u}_j := \frac{\hat{t}_j}{\alpha}.$$ Notice from the algorithm that this will be an integer, because detection only occurs when $t \mod{\alpha} = 0 $.
\end{definition}

We will show that, under appropriate conditioning, w.h.p., $\hat{u}_j = u_j$ or $\hat{u}_j = u_j + 1$.

\begin{definition} \label{def_P_starnew}
Define
\begin{align*}
\bm{P}_{(j)} &:= \bm{P}_{t_{j}} \text{ for } j = 0,1,\dots,J \\ %\ \text{and} \ \bm{P}_{(j)}:=\bm{P}_{t_1-1}  \\
\bm{P}_{(j),*} &:= \bm{P}_{(j-1)} = \bm{P}_{t_{j-1}}  \text{ and }  \bm{P}_{(j),\new}:= \bm{P}_{t_j,\new} \text{ for } j = 1,\dots,J \\
\bm{a}_{t,*} &:= {\Pjs}'\lt   \text{ and }  \bm{a}_{t,\rmnew} := {\Pjnew}'\lt  \text{ for }  t \in [t_j, t_{j+1})
\end{align*}
\end{definition}
Thus, for $t\in[t_{j}, t_{j+1})$, $\lt$ can be written as
\[
\lt = [ \Pjs \ \Pjnew ]\vect{\bm{a}_{t,*}}{\atnew} = \Pjs\bm{a}_{t,*} + \Pjnew\atnew
\]
 and $\cov(\lt)= \bm{\Sigma}_t$ can be rewritten as
\[
\bm{\Sigma}_t =
\left[\bm{P}_{(j),*} \ \bm{P}_{(j),\new}\right]
\left[\begin{array}{cc}\bm{\Lambda}_{t,*} & \bm{0} \\ \bm{0} & \bm{\Lambda}_{t,\rmnew}\end{array}\right]
\left[\begin{array}{c} {\bm{P}_{(j),*}}' \\ {\bm{P}_{(j),\new}}' \end{array} \right]
\]

\begin{definition}\label{def_Phat_starnew}
For $j=1,2,\dots,J$ and $k=1,2,\dots,K$ define
\begin{enumerate}
\item $\ds \Phat_{(1),*}:=\Phat_{t_\train}$ (the initial estimate) and
$\ds\Phat_{(j),*} := \Phat_{\that_{j-1}+K\alpha}$.  If all subspace changes are correctly detected, this is the final estimate of $\bm{P}_{(j),*} = \bm{P}_{(j-1)}$.
\item $\ds\Phat_{(j),\new,0} := [.]$ and $\ds\Phat_{(j),\new,k} := \Phat_{\that_j + k \alpha,\new}$. This is the $k^{\text{th}}$ estimate of $\bm{P}_{(j),\new}$ (again, conditioned on correct change time detection).
\end{enumerate}
\end{definition}

Notice from the algorithm that %if all subspace changes are accurately detected,
\begin{enumerate}
\item  $\Phat_{t,*}  = \Phat_{(j),*}$ for all $t \in [\that_{j-1}+K \alpha, \that_{j}+K \alpha-1]$
\item $\Phat_{t,\new} =  \Phat_{(j),\new,k}$ for all $t \in \J_{\uhat_j+(k+1)}$
\item At all times $\Phat_t = [\Phat_{t,*} \  \Phat_{t,\new}]$. Thus $\Phat_t$ and $\Phat_{t,\new}$ update at every $t = \that_j+ k \alpha$, $k=1,2,\dots,K$, $j=1,2,\dots,J$ while $\Phat_{t,*}$ updates at every $t=\that_{j-1}+K \alpha$, $j=2,\dots,J$.
\item $\Phat_{t-1,*}\perp\Phat_{t,\new}$ at $t=\that_j+k \alpha$ and so $\Phat_{(j),*} \perp \Phat_{(j),\new,k}$
\item $\bm{\Phi}_t = (\I - \Phat_{(j),*}\Phat_{(j),*}{}' - \Phat_{(j),\new,k}\Phat_{(j),\new,k}{}' ) $ when $t\in \J_{\uhat_j+(k+1)}$, for $k=1,2,\dots K-1$.
%[\uhat_j+k\alpha+1,\uhat_j+(k+1)\alpha]$
\item $\bm{\Phi}_t = (\I - \Phat_{(j),*}\Phat_{(j),*}{}')$ when $t \in [t_j, \that_j+\alpha]$ (recall that $\that_j = \uhat_j \alpha$).
\item $\bm{\Phi}_t = (\I - \Phat_{(j+1),*}\Phat_{(j+1),*}{}')$ when $t \in [\that_j + K \alpha +1, t_{j+1}-1]$.
\end{enumerate}

Using the notation from the above definition, Figure \ref{algofig} summarizes Algorithm \ref{reprocsdet}.

\begin{figure}
\includegraphics[width=\textwidth]{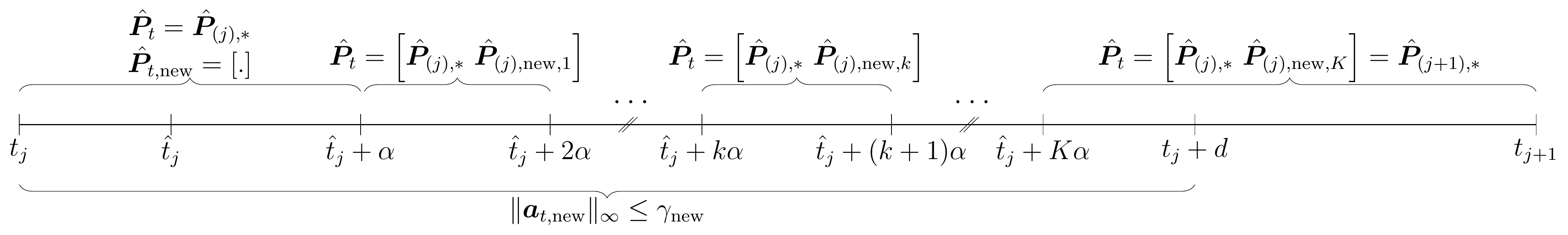}
\caption{A diagram to visualize Algorithm \ref{reprocsdet} and Definition \ref{def_Phat_starnew}. The $k$-th projection-PCA step (at $t = \that_j + k \alpha$) computes the top left singular vectors of $(\I - \Phat_{(j),*}\Phat_{(j),*}{}')[\lhat_{\that_j+(k-1)\alpha+1},\lhat_{\that_j+(k-1)\alpha+2}, \dots \lhat_{\that_j+k\alpha}]$. \label{algofig} }
\end{figure}

\begin{definition}\label{difzeta}
Recall that for basis matrices $\bm{P}$ and $\bm{Q}$, $\mathrm{dif}(\bm{P},\bm{Q}) := \|(\I - \bm{P}\bm{P}')\bm{Q}\|_2$.
Define
\begin{enumerate}
\item $\ds \zeta_{j,*} := \mathrm{dif}(\hat{\bm{P}}_{(j),*},\bm{P}_{(j),*}) $
\item $\ds \zeta_{j,\new,k} := \mathrm{dif}([\Phat_{(j),*} \ \Phat_{(j),\new,k}],\bm{P}_{(j),\new}) $
\end{enumerate}
Recall $\mathrm{SE}_t = \mathrm{dif}(\Phat_t,\bm{P}_t)$.
Notice that if subspace change times are correctly detected, for $t\in\J_{\uhat_j+k}$, $\mathrm{SE}_t \leq \zeta_{j,*}+\zeta_{j,\new,k-1}$ for $k=1,2, \dots K$; for $t \in [t_j, \that_j+\alpha]$, $\SE_t \le 1$; and for $t \in [\that_j + K \alpha +1, t_{j+1}-1]$, $\SE_t = \zeta_{j+1,*}$.
\end{definition}

\begin{definition}
Define
\begin{enumerate}
\item $\ds \zeta_{j,*}^+ := \big( r_0 + (j-1)r_{\new} \big)\zeta $
\item $\ds \zeta_{j,\new,0}^+ :=1$,  $\ds \zeta_{j,\new,k}^+ := \frac{b_{\bm{\mathcal{H}},k}}{b_{\bm{A}} - b_{\bm{A},\perp} - b_{\bm{\mathcal{H}},k}} $ for $k=1,2, \dots, K$ where $b_{\bm{A}}$, $b_{\bm{A},\perp}$, and $b_{\bm{\mathcal{H}},k}$ are defined in Lemmas \ref{Ak}, \ref{Akperp}, and \ref{calHk}. Their expressions use $\epsilon$ given by (\ref{def_eps}).
\end{enumerate}
We will show that these are high probability upper bounds on $\zeta_{j,*}$ and $\zeta_{j,\new,k}$ under appropriate conditioning.
\end{definition}
As we will see later, $b_{\bm{A}} \approx \lambda_\new^-$, $b_{\bm{A},\perp} \approx \zeta_{j,*}^+{}^2 \lambda^+$ and
   $b_{\bm{\mathcal{H}},k} \approx 2\sqrt{\rho^2 h^+} \phi^+ (\zeta_{j,*}^+{}^2 \lambda^+ + \zeta_{j,\new,k-1}^+{}  \lambda_\new^+)$. Here $\approx$ means we are giving only the most dominant term for each expression. Thus,
   \[
   \zeta_{j,\new,k}^+ \approx \frac{2 \sqrt{\rho^2 h^+} \phi^+ (\zeta_{j,\new,k-1}^+{}  \lambda_\new^+ + \zeta_{j,*}^+{}^2 \lambda^+ )}{\lambda_\new^- - \zeta_{j,*}^+{}^2 \lambda^+ -  2 \sqrt{\rho^2 h^+} \phi^+  (\zeta_{j,\new,k-1}^+{}  \lambda_\new^+ + \zeta_{j,*}^+{}^2 \lambda^+ ) }.
   \]
By using (\ref{anew_small}), the bounds on $\zeta$ from the theorem, and the bound on $\rho^2 h^+$, one can show that this decays roughly exponentially with $k$ (see Lemma \ref{zetadecay}).
  % \[
%   \zeta_{j,\new,k}^+ \approx \frac{\sqrt{\rho^2 h^+} \phi^+ ( 0.00015 r_\new \zeta +  3 \zeta_{j,\new,k-1}^+{})}{1 - 0.00015 r_\new \zeta -  \sqrt{\rho^2 h^+} \phi^+ ( 0.00015 r_\new \zeta  + 3 \zeta_{j,\new,k-1}^+{}) }.
%   \]

\begin{definition}
Define the random variable
\[
X_u := \{\bm{a}_1,\dots,\bm{a}_{u\alpha}\}
\]
%If assuming Model \ref{randsup} (random support change), then both $\{\theta_t\}$ and $\{\mu_t\}$ for $t=1,\dots,t_{\max}$ are also included in the definition of $X_u$ for all $u$.  Thus whenever conditioning on an $X_u$, they can be treated as deterministic.
\end{definition}

\begin{definition}
Recall the definition of ${\bm{\mathcal{D}}_u}$ from Algorithm \ref{reprocsdet}.
For $j = 1,\dots,J$, $k = 1,\dots,K$, and for $a = u_j$ or $a = u_j + 1$, define the following events
\begin{itemize}
\item $\ds \mathrm{DET}_j^a := \left\{ \hat{u}_j = a \right\} $
\item $\ds \mathrm{PPCA}_{j,k}^a := \left\{ \hat{u}_j = a \ \text{\em and}\ \rank(\Phat_{(j),\new,k}) = r_{j,\new} \ \text{\em and}\ \zeta_{j,\new,k}\leq \zeta_{j,\new,k}^+  \right\}$
\item $ \mathrm{NODETS}_j^a := \left\{ \hat{u}_j = a \ \text{\em and}\ \lambda_{\max}\left(\frac{1}{\alpha}\bm{\mathcal{D}}_u{\bm{\mathcal{D}}_u}'\right)<\thresh \ \text{\em for all}\ u \in [\hat{u}_j +K+1, u_{j+1}-1] \right\} $

\item $\Gamma_{0,\rmend} := \left\{ \zeta_{1,*}\leq r_0\zeta  \right\} \cap \left\{ \lambda_{\max}\left(\frac{1}{\alpha}\bm{\mathcal{D}}_u{\bm{\mathcal{D}}_u}'\right)<\thresh \ \text{\em for all}\ u \in [1, u_1-1] \right\} $
\item $\ds \Gamma_{j,0}^a :=  \Gamma_{j-1,\rmend}\cap \mathrm{DET}_j^a$
\item $\ds \Gamma_{j,k}^a := \Gamma_{j,k-1}^a \cap \mathrm{PPCA}_{j,k}^a$
\item $\ds \Gamma_{j,\rmend} := \Big( \Gamma_{j,K}^{u_j} \cap \mathrm{NODETS}_j^{u_j} \Big) \cup \left(  \Gamma_{j,K}^{u_j +1} \cap \mathrm{NODETS}_j^{u_j +1} \right)$
\end{itemize}

We misuse notation as follows.
Suppose that a set $\Gamma$ is a subset of all possible values that a r.v. $X$ can take. For two r.v.s' $\{X,Y\}$, when we need to say ``$X \in \Gamma$ and $Y$ can be anything" we will sometimes misuse notation and just say ``$\{X,Y\} \in \Gamma$."  For example, we sometimes say $X_{u_j}\in\Gamma_{j,\rmend}$.  This means $X_{u_j-1}\in\Gamma_{j,\rmend}$ and $\bm{a}_{t}$ for $t\in\J_{u_j}$ are unconstrained.
%For example $X_{\hat{u}_{j+1}}\in\Gamma_{j,\rmend}$ is technically incorrect, but we will use it to mean $\left\{\bm{\nu}_1,\dots,\bm{\nu}_{\hat{u}_j+K\alpha}\right\}\in\Gamma_{j,\rmend}$ and $\left\{ \bm{\nu}_{\hat{u}_j + K\alpha+1}, \dots ,\bm{\nu}_{\hat{u}_{j+1}} \right\}$ are unconstrained.
\end{definition}

%\begin{definition}
%Define $\M_{u} := \frac{1}{\alpha}\bm{\mathcal{D}}_u {\bm{\mathcal{D}}_u}'$, where $\bm{\mathcal{D}}_u$ is defined in Algorithm \ref{reprocsdet}.
%Conditioned on $\Gamma_{j-1,\rmend} \cap \bar{\mathrm{DET}_j^{u_j}}$,  for $u = u_j +1$, $\Phat_{t-1,*} = \Phat_{(j),*}$ and thus, for this value of $u$, $\M_u$ satisfies (\ref{Mudef}) below.
%\begin{equation} \label{Mudef}
%\M_{u} = (\I - \hat{\bm{P}}_{(j),*}\hat{\bm{P}}_{(j),*}{}')  \left(\frac{1}{\alpha} \sum_{t\in\J_u}\lhat_t\lhat_t{}'\right)(\I - \hat{\bm{P}}_{(j),*}\hat{\bm{P}}_{(j),*}{}').
%\end{equation}
%Conditioned on $\Gamma_{j,0}^{\uhat_j}$, for  $u = \uhat_j + k$, $k=1,2, \dots K$, and $\uhat_j=u_j$ or $\uhat_j=u_j+1$, $\Phat_{t-1,*} = \Phat_{(j),*}$ and thus, for these values of $u$, $\M_u$ satisfies (\ref{Mudef}).
%In this case, $\M_{u}$ is also the matrix whose eigenvectors with eigenvalue above $\mathrm{thresh}$ form $\Phat_{(j),\new,k}$ (see step \ref{PCA} of Algorithm \ref{reprocsdet}).  In this case, $\M_{u}$ has eigendecomposition
%\begin{align*}
%\M_u \overset{\mathrm{EVD}}{=} \left[ \begin{array}{cc} \Phat_{(j),\new,k} & \Phat_{(j),k,\new,\perp} \\ \end{array} \right]
%\left[\begin{array}{cc}\hat{\bm{\Lambda}}_u \ & \bm{0} \ \\ \bm{0} \ & \ \hat{\bm{\Lambda}}_{u,\perp} \\ \end{array} \right]
%\left[ \begin{array}{c}\Phat_{(j),\new,k}{}' \\ \Phat_{(j),k,\new,\perp}{}' \\ \end{array} \right].
%\end{align*}
%\end{definition}

\begin{definition}\label{defHk}
Define
\begin{enumerate}
\item Let $\bm{D}_{j,\new}:= (\bm{I} - \Phat_{(j),*} \Phat_{(j),*}{}')\bm{P}_{(j),\new} \overset{QR}{=} \bm{E}_{j,\rmnew} \bm{R}_{j,\rmnew}$ denote its reduced QR decomposition, i.e. let $\bm{E}_{j,\rmnew}$ be a basis matrix for $\Span\left(\bm{D}_{j,\new}\right)$ and let $\bm{R}_{j,\rmnew} = {\bm{E}_{j,\rmnew}}'\bm{D}_{j,\new}$.

\item Let $\bm{E}_{j,\rmnew,\perp}$ be a basis matrix for the orthogonal complement of $\Span(\bm{E}_{j,\rmnew})$. To be precise, $\bm{E}_{j,\rmnew,\perp}$ is an $n\times(n-r_j)$ basis matrix so that $[\bm{E}_{j,\rmnew} \ {\bm{E}_{j,\rmnew,\perp}}]$ is unitary.

\item
For $u = u_{j}+1$ and $u = \hat{u}_{j}+k$ for $k = 1,\dots,K$, define $\bm{A}_{u}$, $\bm{A}_{u,\perp}$, $\bm{\mathcal{A}}_{u}$ as
\begin{align*}
\bm{A}_{u} &:= \frac{1}{\alpha} \sum_{t \in \mathcal{J}_{u}} {\bm{E}_{j,\rmnew}}' (\I - \hat{\bm{P}}_{(j),*}\hat{\bm{P}}_{(j),*}{}') \bm{\ell}_t {\bm{\ell}_t}' (\I - \hat{\bm{P}}_{(j),*}\hat{\bm{P}}_{(j),*}{}')\bm{E}_{j,\rmnew} \\
\bm{A}_{u,\perp} &:= \frac{1}{\alpha} \sum_{t \in \mathcal{J}_{u}} {\bm{E}_{j,\rmnew,\perp}}' (\I - \hat{\bm{P}}_{(j),*}\hat{\bm{P}}_{(j),*}{}') \bm{\ell}_t {\bm{\ell}_t}' (\I - \hat{\bm{P}}_{(j),*}\hat{\bm{P}}_{(j),*}{}') \bm{E}_{j,\rmnew,\perp}
\end{align*}
and let
\[
\bm{\mathcal{A}}_{u} := \left[ \begin{array}{cc} \bm{E}_{j,\rmnew} & \bm{E}_{j,\rmnew,\perp} \\ \end{array} \right]
\left[\begin{array}{cc} \bm{A}_{u} \ & \bm{0} \ \\ \bm{0} \ & \bm{A}_{u,\perp}  \\ \end{array} \right]
\left[ \begin{array}{c} {\bm{E}_{j,\rmnew}}' \\ {\bm{E}_{j,\rmnew,\perp}}' \\ \end{array} \right]
\]

\item For $u = u_{j}+1$ and $u = \hat{u}_{j}+k$ for $k = 1,\dots,K$, define $\bm{\mathcal{M}}_{u}$ and $\bm{\mathcal{H}}_{u}$ as
\[
\M_{u} = (\I - \hat{\bm{P}}_{(j),*}\hat{\bm{P}}_{(j),*}{}')  \left(\frac{1}{\alpha} \sum_{t\in\J_u}\lhat_t\lhat_t{}'\right)(\I - \hat{\bm{P}}_{(j),*}\hat{\bm{P}}_{(j),*}{}')
\]
and
\[
\bm{\mathcal{H}}_{u} := \bm{\mathcal{M}}_u - \bm{\mathcal{A}}_u
\]
\end{enumerate}
\end{definition}

% \cap \overline{\mathrm{DET}_j^{u_j}}
\begin{remark}
Recall the definition of ${\bm{\mathcal{D}}_u}$ from Algorithm \ref{reprocsdet}.

Conditioned on $\Gamma_{j-1,\rmend}$, for $u=u_j+1$, $\Phat_{u\alpha-1,*} = \Phat_{(j),*}$ (in other words all $j-1$ previous subspace changes were detected) and thus, for this value of $u$,
$$\frac{1}{\alpha}\bm{\mathcal{D}}_u {\bm{\mathcal{D}}_u}' = \M_u.$$ In this case, $\M_u$ is the matrix whose maximum eigenvalue is checked to detect subspace change.

%and $\uhat_j=u_j$ or $\uhat_j=u_j+1$,
Conditioned on $\Gamma_{j,0}^{\uhat_j}$, for  $u = \uhat_j + k$, $k=1,2, \dots, K$, $\Phat_{u\alpha-1,*} = \Phat_{(j),*}$ and thus, for these values of $u$ also,
\[
\frac{1}{\alpha}\bm{\mathcal{D}}_u {\bm{\mathcal{D}}_u}' = \M_u.
\]
In this case, $\M_{u}$ is the matrix whose eigenvectors with eigenvalue above $\mathrm{thresh}$ form $\Phat_{(j),\new,k}$ (see step \ref{PCA} of Algorithm \ref{reprocsdet}). In other words, $\M_{u}$ has eigendecomposition
\begin{align*}
\M_u \overset{\mathrm{EVD}}{=} \left[ \begin{array}{cc} \Phat_{(j),\new,k} & \Phat_{(j),\new,k,\perp} \\ \end{array} \right]
\left[\begin{array}{cc}\hat{\bm{\Lambda}}_u \ & \bm{0} \ \\ \bm{0} \ & \ \hat{\bm{\Lambda}}_{u,\perp} \\ \end{array} \right]
\left[ \begin{array}{c}\Phat_{(j),\new,k}{}' \\ \Phat_{(j),\new,k,\perp}{}' \\ \end{array} \right].
\end{align*}
\end{remark}

\begin{definition}
Define
\ben
\item $\kappa_{s,*} := \kappa_s(\bm{P}_{(J)})$ and $\kappa_{s,\rmnew} := \max_{j} \kappa_s(\bm{P}_{(j),\rmnew})$.
\item $\kappa_s^+:=0.0215$. As we will show later in Lemma \ref{Dnew0_lem}, this upper bounds $\| {\I_{\T_t}}'\bm{D}_{j,\new}\|_2$ under appropriate conditioning.
\item $\phi^+:=1.2$. As we will show later in Lemma \ref{cslem}, this upper bounds $\phi_t := \| [ ({\bm{\Phi}_{t})_{\mathcal{T}_t}}'(\bm{\Phi}_{t})_{\mathcal{T}_t}]^{-1} \|_2$ under appropriate conditioning.% when conditioned on appropriate events.
\een
\end{definition}

\begin{remark}
The entire proof uses Model \ref{general_model} on $\T_t$. By Lemma \ref{spc_case}, Model \ref{sbyrho} is a special case of it. In particular, this means that (a) Model \ref{sbyrho} also implies $\rho^2 h^+ \le 0.01$ and (b) Model \ref{sbyrho} also allows us to use Lemma \ref{blockdiag1}. This lemma is used in the proof of Lemma \ref{calHk} in Section \ref{3_pfs}.
\end{remark}

\subsection{Five Main Lemmas for Proving Theorem \ref{thm1}}

\begin{fact} \label{d_large}
Observe that $\Gamma_{j,0}^{a}$ both for $a = u_j$ and $a=u_j+1$ implies that $u_j\leq\hat{u}_j \leq u_j + 1$. Thus, in both cases, $t_j\leq \hat{t}_j \leq t_j + 2\alpha$.  So with the model assumption that $d\geq (K+2)\alpha$, we have that $\J_{\hat{u}_j+k} \subseteq [t_j, t_j+d]$ for $k=1,2,\dots, K$.  This fact is needed so that we can use the ``slow subspace change" inequality, (\ref{anew_small}), to bound the eigenvalues along the new directions, and so that we can bound $\|\atnew\|_\infty$ by $\gamma_\new$.
%the tighter bounds on $\bm{a}_{t,\new}$ hold. %$\|\bm{a}_{t,\new}\|_{\infty} \leq \gamma_{\new}$ and $\lambda_{\max}(\bm{\Lambda}_{a,t,\new})\leq\lambda_{\new}^+$ hold.
\end{fact}

\begin{lem}\label{zetadecay}[{Exponential decay of the bound on $\zeta_{j,\new,k}$ (similar to \cite[Lemma 6.1]{ReProCS_IT})}]
Under the conditions of Theorem \ref{thm1},
\[
\zeta_{j,\new,k}^+ \leq 0.83^k + 0.84 r_{\new}\zeta
\]
\end{lem}

This lemma follows by applying simple algebra on the definition and using the bounds assumed on $\zeta$, $\lambda_{\new}^+$ and $\rho^2 h^+$ in Theorem \ref{thm1}. A detailed proof of this lemma is given in Appendix \ref{zeta_k_section}.

\begin{lem}[{Sparse Recovery Lemma (similar to \cite[Lemma 6.4]{ReProCS_IT})}]\label{cslem}
Assume that all of the conditions of Theorem \ref{thm1} hold.
Recall that $\mathrm{SE}_t = \operatorname{dif}(\Phat_t,\bm{P}_t)$.
\begin{enumerate}
\item  Conditioned on $\Gamma_{j-1,\rmend}$, for $t \in [t_j, (\uhat_j+1)\alpha]$
\begin{enumerate}
\item $\phi_t := \| [ ({\bm{\Phi}_{t})_{\mathcal{T}_t}}'(\bm{\Phi}_{t})_{\mathcal{T}_t}]^{-1} \|_2 \leq \phi^+ := 1.2$.
\item  the support of $\xt$ is recovered exactly i.e. $\hat{\mathcal{T}}_t = \mathcal{T}_t$ and  $\et$ satisfies:
\begin{align}\label{etdef0}
\et := \hat{\bm{x}}_t - \xt  = \lt - \lhatt = \bm{I}_{\mathcal{T}_t} [ ({\bm{\Phi}_{t})_{\mathcal{T}_t}}'(\bm{\Phi}_{t})_{\mathcal{T}_t}]^{-1}  {\bm{I}_{\mathcal{T}_t}}' \bm{\Phi}_{t} \bm{\ell}_t.
\end{align} %= \bm{I}_{\mathcal{T}_t} {(\bm{\Phi}_{t})_{\mathcal{T}_t}}^{\dag} \bm{{b}}_t
\item Furthermore,
\begin{align*}
\mathrm{SE}_t &\leq 1 \ \text{, and} \\
\|\bm{e}_t\|_2 &\leq \phi^+ (\zeta_{j,*}^+ \sqrt{r}\gamma +  \sqrt{r_{\new}}\gamma_{\rmnew}) \leq
1.2 \left(\sqrt{\zeta} +  \sqrt{r_{\new}}\gamma_{\rmnew}  \right)
\end{align*}
\end{enumerate}

\item For $k=2,3, \dots, K$ and $\hat{u}_j = u_j$ or $\hat{u}_j = u_j +1$, conditioned on $\Gamma_{j,k-1}^{\hat{u}_j}$, for $t \in \J_{\hat{u}_j+k} = \left[(\uhat_j+ k-1)\alpha+1, (\uhat_j+ k)\alpha \right]$, the first two conclusions above hold. That is, $\phi_t \leq\phi^+$ and $\et$ satisfies (\ref{etdef0}).  Furthermore,
\begin{align*}
\mathrm{SE}_t &\leq \zeta_{j,*}^+ + \zeta_{j,\new,k-1}^+  \ \text{, and} \\
\|\bm{e}_t\|_2 &\leq \phi^+ (\zeta_{j,*}^+ \sqrt{r}\gamma +  \zeta_{j,\new,k-1}^+ \sqrt{r_{\new}}\gamma_{\rmnew})
\leq 1.2 \left(1.84\sqrt{\zeta} +  (0.83)^{k-1}\sqrt{r_{\new}}\gamma_{\rmnew}  \right)
\end{align*}

\item For $\hat{u}_j = u_j$ or $\hat{u}_j = u_j +1$, conditioned on $\Gamma_{j,K}^{\uhat_j}$, for $t \in \left[(\uhat_j+ K)\alpha+1, t_{j+1}-1\right]$, the first two conclusions above hold ($\phi_t\leq\phi^+$ and $\et$ satisfies \eqref{etdef0}).
Furthermore,
\begin{align*}
\SE_t &\leq \zeta_{j+1,*}^+  \ \text{, and} \\
\|\bm{e}_t\|_2 &\leq  \phi^+ \zeta_{j+1,*}^+ \sqrt{r}\gamma \leq 1.2 \sqrt{\zeta}
\end{align*}
\end{enumerate}

\end{lem}

Notice that cases 1) and 3) of the above lemma occur when the algorithm is in the detection phase, while during the intervals for case 2) the algorithm is performing projection-PCA.  In case 1) new directions have been added but not estimated, so the error is larger.  In case 2), the error is decaying exponentially with each estimation step.  Finally, case 3) occurs after the new directions have been successfully estimated and contains the tightest error bounds.

The proof is given in Appendix \ref{pf_cslem}.

\begin{lem}[No false detection of subspace changes] \label{falsedet}
\
\begin{enumerate}
\item The event $\Gamma_{j,K}^{\hat{u}_j}$ and so also the event $\Gamma_{j,\rmend}$ imply that $\zeta_{j+1,*}\leq\zeta_{j+1,*}^+$.
\item $\ds  \Pr\left(\mathrm{NODETS}_j^a \ | \ \Gamma_{j,K}^a \right) = 1$ for $a = u_j$ or $a = u_j + 1$.

\end{enumerate}
\end{lem}

\begin{lem}[Subspace change detected within $2\alpha$ frames] \label{det}
For $j = 1,\dots,J$,
\[
\Pr\left(\mathrm{DET}^{u_j+1} \ | \ \Gamma_{j-1,\rmend}, \overline{\mathrm{DET}^{u_j}} \right) \geq p_{\det,1} := 1 - p_{\bm{A}} - p_{\bm{\mathcal{H}}}.
\]
The definitions of $p_{\bm{A}}$ and $p_{\bm{\mathcal{H}}}$ can be found in Lemmas \ref{Ak} and \ref{calHk} respectively.
\end{lem}

\begin{lem}[$k$-th iteration of pPCA works well] \label{pPCA}
\[
\Pr\left(\Gamma_{j,k}^a \ | \ \Gamma_{j,k-1}^a \right) = \Pr\left(\mathrm{PPCA}_{j,k}^a \ | \ \Gamma_{j,k-1}^{a} \right)  \geq p_{\mathrm{ppca}} := 1 - p_{\bm{A}} - p_{\bm{A},{\perp}} - p_{\bm{\mathcal{H}}}
\]
for $a = u_j$ or $a = u_{j}+1$.
The definitions of $p_{\bm{A}}$, $p_{\bm{A},{\perp}}$, and $p_{\bm{\mathcal{H}}}$ can be found in Lemmas \ref{Ak}, \ref{Akperp}, and \ref{calHk} respectively.
\end{lem}
The above lemma says that, conditioned on $k-1$ previous successful p-PCA steps and on accurate recovery of $P_{(j-1),*}$, the probability of correctly estimating $r_{j,\new}$ and of a successful $k^{\text{th}}$ projection PCA step is lower bounded by $p_{\mathrm{ppca}}$. This is true whether the new directions are detected at $u_j$ or at $u_j+1$.

\subsection{Proof of Theorem \ref{thm1}}
\begin{corollary}\label{GammaCor}
Let
\[
p_{\det,0} := \Pr\left( \mathrm{DET}^{u_j} \ | \  \Gamma_{j-1,\rmend} \right).
\]
%and therefore, $1 - p_{\det,0} = \Pr\left( \overline{\mathrm{DET}^{u_j}} \ | \ \Gamma_{j-1,\rmend} \right)$.
From the above lemmas, we get that
\begin{align*}
\Pr\left( \Gamma_{j,\rmend} \ | \ \Gamma_{j-1,\rmend} \right)
&= \Pr\Big( \big(\mathrm{DET}^{u_j} \cap\mathrm{PPCA}_{j,1}^{u_j} \cap\dots\cap \mathrm{PPCA}_{j,K}^{u_j}\cap\mathrm{NODETS}_j^{u_j}\big) \cup \\
&\hspace{-.36in} \big(\overline{\mathrm{DET}^{u_j}}\cap\mathrm{DET}^{u_j + 1} \cap \mathrm{PPCA}_{j,k}^{u_j+1}\cap\dots\cap \mathrm{PPCA}_{j,K}^{u_j+1}\cap\mathrm{NODETS}_j^{u_j+1}\big)   \ | \ \Gamma_{j-1,\rmend}\Big) \\
&= \Pr\left( \mathrm{DET}^{u_j} \cap\mathrm{PPCA}_{j,1}^{u_j} \cap\dots\cap \mathrm{PPCA}_{j,K}^{u_j} \ | \ \Gamma_{j-1,\rmend} \right) \\
&\quad + \Pr\left(\overline{\mathrm{DET}^{u_j}}\cap\mathrm{DET}^{u_j + 1} \cap \mathrm{PPCA}_{j,k}^{u_j+1}\cap\dots\cap \mathrm{PPCA}_{j,K}^{u_j+1} \ | \ \Gamma_{j-1,\rmend}\right)\\
&\geq p_{\det,0}\cdot (p_{\mathrm{ppca}})^K + (1-p_{\det,0})\cdot p_{\det,1}\cdot (p_{\mathrm{ppca}})^K \\
&\geq p_{\det,0}\cdot p_{\det,1}\cdot (p_{\mathrm{ppca}})^K + (1-p_{\det,0})\cdot p_{\det,1}\cdot (p_{\mathrm{ppca}})^K \\
&= p_{\det,1}\cdot (p_{\mathrm{ppca}})^K.
\end{align*}
\end{corollary}

\begin{proof}[{\bf \em Proof of Theorem \ref{thm1}}]
Theorem \ref{thm1} follows from Corollary \ref{GammaCor} and the assumed lower bound on $\alpha$.
Notice that by Lemma \ref{zetadecay}, the choice of $K$, and Lemma \ref{cslem}, the event $\Gamma_{J,\rmend}$ will imply all conclusions of the theorem.

By the first assumption (accurate initial subspace knowledge) and the argument used to prove Lemma \ref{falsedet}, we get that $\Pr(\Gamma_{0,\rmend})=1$.
By the chain rule,
$
\Pr(\Gamma_{J,\rmend}) = \prod_{j=1}^{J} \Pr(\Gamma_{j,\rmend} \ | \ \Gamma_{j-1,\rmend},\Gamma_{j-2,\rmend},\dots,\Gamma_{1,\rmend},\Gamma_{0,\rmend}).
$
Because $\Gamma_{j-1,\rmend}\subseteq\Gamma_{j-2,\rmend}\subseteq\dots\subseteq \Gamma_{1,\rmend}\subseteq\Gamma_{0,\rmend}$, we get
\begin{align*}
\Pr(\Gamma_{J,\rmend}) &= \prod_{j=1}^{J} \Pr(\Gamma_{j,\rmend} \ | \ \Gamma_{j-1,\rmend})\\
&\geq \prod_{j=1}^{J} p_{\det,1}\cdot (p_{\mathrm{ppca}})^K = (p_{\det,1}\cdot (p_{\mathrm{ppca}})^K)^{J}\\
&\geq 1 - n^{-10}
\end{align*}
The last line is by the lower bound on $\alpha$ assumed in the theorem and the fact that $p_{\det,1}\geq p_{\mathrm{ppca}}$.
\end{proof}

\subsection{Key Lemmas for Proving of Lemmas \ref{falsedet}, \ref{det}, and \ref{pPCA}}

Before proving the lemmas from the preceding subsection, we introduce several lemmas which will be used in the proofs.% Their statement requires the following definition.

The following lemma follows from the $\sin\theta$ theorem \cite{davis_kahan} and Weyl's theorem.  It is taken from \cite{ReProCS_IT}.
\begin{lem}[\cite{ReProCS_IT}, Lemma 6.9]\label{zetakbnd}
At $u = \hat{u}_j + k$,
if $\rank(\Phat_{(j),\new,k}) = r_{j,\new}$, and if $\lambda_{\min}(\bm{A}_u) - \|\bm{A}_{u,\perp}\|_2 - \|\bm{\mathcal{H}}_u\|_2 >0$, then
\begin{align} \label{zetakbound}
\zeta_{j,\new,k} \leq  \frac{\|\bm{\mathcal{H}}_u\|_2}{\lambda_{\min} (\bm{A}_u) - \|\bm{A}_{u,\perp}\|_2 - \|\bm{\mathcal{H}}_u\|_2}
\end{align}
\end{lem}

The next three lemmas each assert a high probability bound for one of the terms in \eqref{zetakbound}.
In the following lemmas, let
\beq \label{def_eps}
\epsilon = \frac{r_\new \zeta\lammin}{100}.
\eeq

\begin{lem}\label{Ak} %[similar to \cite[claim 1 of Lemma 6.11]{ReProCS_IT}]
Let $p_{\bm{A}} :=  r_\new \exp \left(\frac{-\alpha \zeta^2 (\lammin)^2}{8 \cdot 100^2 \cdot {\gamma_{\rmnew}}^4  } \right) + r_\new \exp \left( \frac{-\alpha {r_\new}^2 \zeta^2(\lammin)^2} {8 \cdot 100^2 \cdot 4^2}\right)$
%c \digamma(\alpha, \epsilon, r_{\new}{\gamma_{\rmnew}}^2)+ 3(r+r_{\new})\digamma(\alpha,\epsilon,2\sqrt{r_{\new}r}\gamma\gamma_{\new}) + (r+r_{\new})\digamma(\alpha, \epsilon,4\sqrt{r_{\new}r}\gamma\gamma_{\new})$
and
\[
b_{\bm{A}} := (1-(\zeta_{j,*}^+)^2) \lambda_{\rmnew}^-  - 2\epsilon .
\]
For $k = 1,\dots,K$,
\begin{align*}
\Pr\left( \lambda_{\min} \left(\bm{A}_{\hat{u}_j + k} \right)\geq b_{\bm{A}} \ \big| \ X_{\hat{u}_j+k-1} \right) \geq 1 -  p_{\bm{A}}
\end{align*}
for all $X_{\hat{u}_j+k-1}\in\Gamma_{j,k-1}^{\hat{u}_j}$ with $\hat{u}_j=u_j$ or $\hat{u}_j=u_j+1$.

The same bound holds for $\lambda_{\min} ( \bm{{A}}_{u_j + 1} )$ when we condition on $X_{u_j} \in \Gamma_{j-1,\rmend}$.
%Also,
%\begin{align*}
%\Pr\left( \lambda_{\min} \left(\bm{A}_{u_{j}+1} \right)\geq b_{\bm{A}} \ \big| \ X_{u_j} \right) \geq 1 -  p_{\bm{A}}
%\end{align*}
%for all $X_{u_j} \in \Gamma_{j-1,\rmend}$.
\end{lem}

\begin{lem}\label{Akperp} %[similar to \cite[claim 2 of Lemma 6.11]{ReProCS_IT}]
Let $p_{\bm{A},{\perp}} := (n-r_\new) \exp \left(\frac{-\alpha {r_\new}^2 \zeta^2 (\lammin)^2}{8 \cdot 100^2}\right)$
%r\digamma\left(\alpha,\epsilon,(\zeta_{j,*}^+)^2r\gamma^2\right)$
and
\[
b_{\bm{A},\perp} := (\zeta_{j,*}^+)^2 \lambda^+ + \epsilon.
\]

For $k = 1,\dots,K$,
\begin{align*}
\Pr\left( \lambda_{\max}\left(\bm{A}_{\hat{u}_j + k,\perp} \right) \leq b_{\bm{A},\perp}  \ \big| \ X_{\hat{u}_j+k-1} \right)\geq 1 - p_{\bm{A},{\perp}}
\end{align*}
for all $X_{\hat{u}_j+k-1}\in\Gamma_{j,k-1}^{\hat{u}_j}$ with $\hat{u}_j=u_j$ or $\hat{u}_j=u_j+1$.

The same bound holds for $\lambda_{\max} ( \bm{{A}}_{u_j + 1,\perp} )$ when we condition on $X_{u_j} \in \Gamma_{j-1,\rmend}$.
\end{lem}

\begin{lem}\label{calHk}
Let %$p_{\bm{\mathcal{H}}}:= ??$
\begin{align*}
p_{\bm{\mathcal{H}}}:=  & \ n \exp\left(\frac{-\alpha {r_\new}^2\zeta^2 (\lammin)^2}{32\cdot 100^2 (\phi^+)^2 (\sqrt{\zeta}+\sqrt{r_\new}\gamma_{\rmnew})^4}\right) + n \exp\left(\frac{-\alpha {r_\new}^2\zeta^2 (\lammin)^2}{8 \cdot 100^2 \Big( \phi^+ (\sqrt{\zeta}+\sqrt{r_\new}\gamma_{\rmnew}) \Big)^4}\right)
+    \\
& n\exp\left(\frac{-\alpha {r_\new}^2 \zeta^2 (\lammin)^2 }{32 \cdot 100^2 ( \zeta + \sqrt{\zeta} \sqrt{r_\new} \gamma_\new )^2}\right).
%&\leq 3 n \exp\left(\frac{-\alpha r_\new^2\zeta^2 (\lammin)^2}{8 \cdot 100^2 \Big( \phi^+ (\sqrt{\zeta}+\sqrt{r_\new}\gamma_{\rmnew}) \Big)^4}\right).
\end{align*}
%:=  n\digamma\left(\alpha,\epsilon,\left[\phi^+ \left(\zeta_{j,*}^+\sqrt{r}\gamma + \sqrt{r_{\new}}\gamma_{\new}\right)\right]^2 \right)  + r\digamma\left(\alpha,\epsilon, r\gamma^2 \right) +  r_{\new}\digamma\left(\alpha,\epsilon, r_{\new}{\gamma_{\new}}^2 \right) + 3(r+r_{\new})\digamma(\alpha,\epsilon,2\sqrt{r_{\new}r}\gamma\gamma_{\new}) + (r+r_{\new})\digamma(\alpha, \epsilon,4\sqrt{r_{\new}r}\gamma\gamma_{\new}) + r_{\new}\digamma\left(\alpha,\epsilon, r{\gamma}^2 \right) +  3(r+r_{\new})\digamma(\alpha,\epsilon,2\sqrt{r_{\new}r}\gamma\gamma_{\new}) + (r+r_{\new})\digamma(\alpha, \epsilon,4\sqrt{r_{\new}r}\gamma\gamma_{\new})$.
and
\[
b_{\bm{\mathcal{H}},k} := 2b_{\l \e,k} + b_{\e \e,k} + 2b_{\bm{F}}  % b_{\e \e,k} +2b_{\l \e,k}+ 2b_{\bm{F}},
\]
where
\begin{align*}
b_{\l \e,k} := \begin{cases}
\begin{array}{l}
\phi^+ \big( \sqrt{\rho^2h^+} (\zeta_{j,*}^+)^2 \lambda^+  +
  \kappa_s^+ \lambda_{\rmnew}^+ \big)  +
 \epsilon
\end{array} & k=1
\\
\begin{array}{l}
\Big[ (\zeta_{j,*}^+)^2 \lambda^+  + \zeta_{j,\new,k-1}^+ \lambda_{\rmnew}^+
\Big] \Big(\sqrt{\rho^2h^+}\phi^+\Big) +  \epsilon
\end{array} & k\geq2
\end{cases}
\end{align*}
\begin{align*}
b_{\e \e,k} :=
\begin{cases}
\begin{array}{l}
\rho^2h^+(\phi^+)^2 \big( (\zeta_{j,*}^+)^2 \lambda^+ + (\kappa_s^+)^2 \lambda_{\new}^+  \big) +\epsilon
\end{array} & k = 1
\\
\begin{array}{l}
\rho^2h^+(\phi^+)^2 \big( (\zeta_{j,*}^+)^2(\lambda^+ ) + (\zeta_{j,\new,k-1}^+)^2(\lambda_{\new}^+ ) \big)
+\epsilon
\end{array} & k \geq 2
\end{cases}
\end{align*}
and
\[
b_{\bm{F}} := (\zeta_{j,*}^+)^2 \lambda^+  + \epsilon.
\]
For $k = 1,\dots,K$,
\begin{align}
&\Pr \left(\|\bm{\mathcal{H}}_{\hat{u}_j + k}\|_2 \leq b_{\bm{\mathcal{H}},k}  \ \big| \ X_{\hat{u}_j+k-1}\right)\geq 1 - p_{\bm{\mathcal{H}}}
\label{normHk}
\end{align}
for all $X_{\hat{u}_j+k-1}\in\Gamma_{j,k-1}^{\hat{u}_j}$ with $\hat{u}_j=u_j$ or $\hat{u}_j=u_j+1$

The same bound ($k=1$ case), i.e. $\|\bm{\mathcal{H}}_{u_j + 1}\|_2 \le b_{\bm{\mathcal{H}},1}$, also holds with the same probability when we condition on $X_{u_j} \in \Gamma_{j-1,\rmend}$.
\end{lem}

The above lemmas are proved in the next section (Section \ref{3_pfs}). The proofs use Fact \ref{d_large}.

\subsection{Proofs of Lemmas \ref{falsedet}, \ref{det}, and \ref{pPCA}}

\begin{proof}[Proof of Lemma \ref{falsedet}]
Recall that $\Gamma_{j,\rmend}:= \Big(\Gamma_{j,K}^{u_j} \cap \mathrm{NODETS}_j^{u_j}\Big) \cup \left(\Gamma_{j,K}^{u_j+1} \cap \mathrm{NODETS}_j^{u_j+1}\right)$.
\begin{enumerate}
\item By the definition of $\Gamma_{j,K}^{\hat{u}_j}$, both for $\uhat_j = u_j$ and $\uhat_j = u_j + 1$, $\zeta_{j,*}\leq \zeta_{j,*}^+ = (r_0 +(j-1)r_{\new})\zeta$ and $\zeta_{j,K}\leq\zeta_{j,\new,K}^+$. Lemma \ref{zetadecay} and the choice of $K$ imply that $\zeta_{j,\new,K}^+\leq r_{\new}\zeta$.  Thus, $\zeta_{j+1,*}\leq \zeta_{j,*} + \zeta_{j,\new,k} \le \zeta_{j+1,*}^+ = (r_0 + jr_{\new})\zeta$.
\item
$\Pr(\mathrm{NODETS}_j^{\uhat_j} \ | \ \Gamma_{j,K}^{\uhat_j} )
= \Pr\Big(\lambda_{\max}( \frac{1}{\alpha}\bm{\mathcal{D}}_u \bm{\mathcal{D}}_u ) < \thresh  \text{ for all } u \in [\uhat_j + K +1, u_{j+1}-1] \ | \ \Gamma_{j,K}^{\uhat_j}\Big)$
for $\uhat_j = u_j$ or $\uhat_j = u_j + 1$.
\end{enumerate}

%\bm{\mathcal{M}}_u %\frac{1}{\alpha} \sum_{t \in \J_u} \bm{\mathcal{D}}_u \bm{\mathcal{D}}_u

As shown in 1), $\Gamma_{j,K}^{\uhat_j}$ implies that $\mathrm{dif}(\hat{\bm{P}}_{(j+1),*},\bm{P}_{(j+1),*})\leq \zeta_{j+1,*}^+ = (r_0 +jr_{\new})\zeta$. Recall that $\bm{P}_{(j+1),*} = \bm{P}_{(j)}$. Also, for $u\in[\hat{u}_j + K + 1,u_{j+1} -1]$, $\Phat_{u\alpha-1,*} = \Phat_{(j+1),*}$. Also, for all  $t \in \J_u$ for these $u$'s, $\lt=\bm{P}_{(j)}\at = \bm{P}_{(j+1),*} \at$.  Therefore,
\begin{align*}
\lambda_{\max} \left( \frac{1}{\alpha}  \bm{\mathcal{D}}_u \bm{\mathcal{D}}_u \right) &= \lambda_{\max}\left(\frac{1}{\alpha}\sum_{t\in\mathcal{J}_u} (\I - \hat{\bm{P}}_{u\alpha-1,*}\hat{\bm{P}}_{u\alpha-1,*}{}') \lhat_t\lhat_t{}'(\I - \hat{\bm{P}}_{u\alpha-1,*}\hat{\bm{P}}_{u\alpha-1,*}{}')\right)\\
&= \lambda_{\max}\left( \frac{1}{\alpha}\sum_{t\in\mathcal{J}_u} (\I - \hat{\bm{P}}_{(j+1),*}\hat{\bm{P}}_{(j+1),*}{}') (\bm{P}_{(j)}\at-\et)(\bm{P}_{(j)}\at-\et){}'(\I - \hat{\bm{P}}_{(j+1),*}\hat{\bm{P}}_{(j+1),*}{}') \right) \\
%\\
%&\leq \frac{1}{\alpha}\sum_{t\in\J_u} \Bigg[ \lambda_{\max}\left((\I - \hat{\bm{P}}_{(j+1),*}\hat{\bm{P}}_{(j+1),*}{}')\bm{P}_{(j)}\at\at'{\bm{P}_{(j)}}'(\I - \hat{\bm{P}}_{(j+1),*}\hat{\bm{P}}_{(j+1),*}{}')\right) + \\
%&\hspace{.8in} 2 \| (\I - \hat{\bm{P}}_{(j+1),*}\hat{\bm{P}}_{(j+1),*}{}')\bm{P}_{(j)}\at{\et}'(\I - \hat{\bm{P}}_{(j+1),*}\hat{\bm{P}}_{(j+1),*}{}')\|_2 + \\
%&\hspace{.8in} \lambda_{\max}\left((\I - \hat{\bm{P}}_{(j+1),*}\hat{\bm{P}}_{(j+1),*}{}')\et{\et}'(\I - \hat{\bm{P}}_{(j+1),*}\hat{\bm{P}}_{(j+1),*}{}')\right) \Bigg]\\
%&\leq (\zeta_{j+1,*}^+)^2r\gamma^2 + 2(\phi^+)(\zeta_{j+1,*}^+)^2r\gamma^2 + (\phi^+)^2(\zeta_{j+1,*}^+)^2r\gamma^2\\
&\leq (\zeta_{j+1,*}^+)^2r\gamma^2 + 2\phi^+(\zeta_{j+1,*}^+)^2r\gamma^2 + (\phi^+)^2(\zeta_{j+1,*}^+)^2r\gamma^2 \\
&\leq 4(\phi^+)^2\zeta\lammin  \leq \frac{\lammin}{2}.
\end{align*}
The bound on $\et$ comes from Lemma \ref{cslem}.
The penultimate inequality uses the bound $\zeta\leq\frac{\lambda_{\train}^-}{r^3\gamma^2}$ assumed in Theorem \ref{thm1}.
\end{proof}

The next two proofs follow using the following two facts and the four lemmas from the previous subsection.
\begin{fact}\label{event}
For an event $\mathcal{E}$ and random variable $X$, $\mathbb{P}(\mathcal{E}|X) \geq p$ for all $X \in \mathcal{C}$ implies that
$\mathbb{P}(\mathcal{E}|X\in \mathcal{C}) \geq p$.
\end{fact}

\begin{fact} \label{bounds_b_A}
Using the bounds on $\zeta$ and on $\rho^2 h^+$ and using \eqref{anew_small}, we get%, and the assumption that $\lambda_{\new}^- \geq \lammin$ we get:
\begin{align*}
b_{\bm{A}} &\geq 0.94 \lambda_{\new}^- \geq 0.94\lammin  \\
b_{\bm{A},\perp} &\leq 0.011 \lammin \\
b_{\bm{\mathcal{H}},k} &\leq 0.24\lammin.
\end{align*}
Thus, $b_{\bm{A}} - b_{\bm{\mathcal{H}},k} \geq 0.5\lammin = \thresh $ and $b_{\bm{A},\perp} + b_{\bm{\mathcal{H}},k} < 0.25 \lammin < \thresh$.
\end{fact}

\begin{proof}[Proof of Lemma \ref{det}]
We will prove that $\Pr\left(\mathrm{DET}^{u_j+1} \ |\ X_{u_j} \right) > p_{\det,1}$ for all $X_{u_j} \in \Gamma_{j-1,\rmend}$. In particular, this will imply that $\Pr(\mathrm{DET}^{u_j+1} \ |\ X_{u_j}) > p_{\det,1}$ for all $X_{u_j} \in \Gamma_{j-1,\rmend} \cap \overline{\mathrm{DET}^{u_j}}$ and so we can conclude that  $\Pr(\mathrm{DET}^{u_j+1} \ |\ \Gamma_{j-1,\rmend}, \overline{\mathrm{DET}^{u_j}}) > p_{\det,1}$.

Recall that $\bm{\mathcal{M}}_u = \frac{1}{\alpha}\bm{\mathcal{D}}_u{\bm{\mathcal{D}}_u}'$, and observe that
\begin{align*}
\Pr\left(\mathrm{DET}^{u_j+1} \ |\ X_{u_j} \right) &= \Pr\left( \lambda_{\max}(\M_{u_{j}+1}) > \thresh \ | \ X_{u_j} \right)
\end{align*}
By Weyl's Theorem
\begin{align*}
\lambda_{\max}(\bm{\mathcal{M}}_{u_{j}+1}) &\geq \lambda_{\max}(\bm{\mathcal{A}}_{u_j+1}) + \lambda_{\min}(\bm{\mathcal{H}}_{u_{j}+1})\\
%&\geq \lambda_{\max}(\bm{\mathcal{A}}_{u_{j}+1}) - \|\bm{\mathcal{H}}_{u_{j}+1}\|_2 \\
&\geq \lambda_{\max}(\bm{A}_{u_{j}+1}) - \|\bm{\mathcal{H}}_{u_{j}+1}\|_2 \\
&\geq \lambda_{\min}(\bm{A}_{u_{j}+1}) - \|\bm{\mathcal{H}}_{u_{j}+1}\|_2
%&\geq b_{\bm{A}} - b_{\bm{\mathcal{H}},1}  \geq \frac{\lammin}{2}
\end{align*}
When $X_{u_{j}} \in \Gamma_{j-1,\rmend}$, Lemmas \ref{Ak} and \ref{calHk}  applied with $\epsilon$ given by \eqref{def_eps} show that $\lambda_{\min}(\bm{A}_{u_j+1}) \ge b_{\bm{A}}$ and $\|\bm{\mathcal{H}}_{u_j+1}\|_2 \le b_{\bm{\mathcal{H}},1}$ with probability at least $1-p_{\bm{A}}-p_{\bm{\mathcal{H}}} = p_{\mathrm{det},1}$.
Using Fact \ref{bounds_b_A}, $b_{\bm{A}} - b_{\bm{\mathcal{H}},1}  \geq \thresh$ and so the lemma follows.
%So the above inequality holds with probability greater than or equal to $1-p_{\bm{A}}-p_{\bm{\mathcal{H}}} = p_{\mathrm{det},1}$.
%Using Fact \ref{bounds_b_A}, we get that this is greater than $\frac{\lammin}{2}$.
%the bounds on $\zeta$ and $b$ assumed in Theorem \ref{thm1},
\end{proof}

\begin{proof}[Proof of Lemma \ref{pPCA}]
To prove this Lemma we need to show two things.  First, conditioned on $\Gamma_{j,k-1}^{\hat{u}_j}$, the $k^{\text{th}}$ estimate of the number of new directions is correct.  That is: $\hat{r}_{j,\new,k} = r_{j,\new}$.  Second, we must show $\zeta_{j,\new,k}\leq\zeta_{j,\new,k}^+$, again conditioned on $\Gamma_{j,k-1}^{\uhat_j}$.

Notice that $\hat{r}_{j,\new,k} = \rank(\Phat_{(j),\new,k})$.  To show that $\rank(\Phat_{(j),\new,k}) = r_{j,\new}$, we need to show that for $u = \hat{u}_j+k$, $k=1,\dots,K$, $\lambda_{r_{j,\new}}(\M_u)>\thresh$ and $\lambda_{r_{j,\new}+1}(\M_u)<\thresh$.  To do this we proceed similarly to above.

Observe that, $\bm{\mathcal{M}}_u = \bm{\mathcal{A}}_u + \bm{\mathcal{H}}_u$.
By Fact \ref{bounds_b_A}, $b_{\bm{A}} > b_{\bm{A},\perp}$.  Combining this with Lemmas \ref{Ak} and \ref{Akperp} gives, $\lambda_{\min}(\bm{A}_u) > \lambda_{\max}(\bm{A}_{u,\perp})$ with probability at least $1-p_{\bm{A}}-p_{\bm{A},\perp}$ under the appropriate conditioning (conditioned on $\Gamma_{j,k-1}^{\uhat_j}$).  Since $\bm{A}_u$ is of size $r_{j,\new}\times r_{j,\new}$, this means that $\lambda_{r_{j,\new}}(\bm{\mathcal{A}}_u) = \lambda_{\min}(\bm{A}_u)$ and $\lambda_{r_{j,\new}+1}(\bm{\mathcal{A}}_u) = \lambda_{\max}(\bm{A}_{u,\perp})$. Using this and Weyl's Theorem,
\begin{align*}
\lambda_{r_{j,\new}}(\bm{\mathcal{M}}_u) &\geq \lambda_{r_{j,\new}}(\bm{\mathcal{A}}_u) + \lambda_{\min}(\bm{\mathcal{H}}_u)\\
&\geq \lambda_{r_{j,\new}}(\bm{\mathcal{A}}_u) - \|\bm{\mathcal{H}}_u\|_2 \\
&= \lambda_{\min}(\bm{A}_u) - \|\bm{\mathcal{H}}_u\|_2
\end{align*}
and
% with probability at least $1-p_{\bm{A}}-p_{\bm{A},\perp}$.
%By Fact \ref{bounds_b_A}, we have that $b_{\bm{A}} > b_{\bm{A},\perp}$.  Combining this with Lemmas \ref{Ak} and \ref{Akperp} gives, $\lambda_{\min}(\bm{A}_u) > \lambda_{\max}(\bm{A}_{u,\perp})$ with probability at least $1-p_{\bm{A}}-p_{\bm{A},\perp}$.  Since $\bm{A}_u$ is of size $r_{j,\new}\times r_{j,\new}$, we have equality in the last line.
\begin{align*}
\lambda_{r_{j,\new}+1}(\bm{\mathcal{M}}_u) &\leq \lambda_{r_{j,\new}+1}(\bm{\mathcal{A}}_u) + \lambda_{\max}(\bm{\mathcal{H}}_u)\\
&\leq \lambda_{r_{j,\new}+1}(\bm{\mathcal{A}}_u) + \|\bm{\mathcal{H}}_u\|_2 \\
&= \lambda_{\max}(\bm{A}_{u,\perp}) + \|\bm{\mathcal{H}}_u\|_2
\end{align*}
 with probability at least $1-p_{\bm{A}}-p_{\bm{A},\perp}$ under the appropriate conditioning.
Using Lemmas \ref{Ak}, \ref{Akperp}, and \ref{calHk} applied with $\epsilon$ given by \eqref{def_eps} and Fact \ref{bounds_b_A}, we can conclude that with probability greater than $p_{\mathrm{ppca}}$, $\lambda_{r_{j,\new}}(\bm{\mathcal{M}}_u) > b_{\bm{A}} - b_{\bm{\mathcal{H}},k} \ge \thresh $ and $\lambda_{r_{j,\new}+1}(\bm{\mathcal{M}}_u)< b_{\bm{A},\perp} + b_{\bm{\mathcal{H}},k} <  \thresh$. Therefore $\rank(\Phat_{(j),\new,k})=r_{j,\new}$ with probability greater than $p_{\mathrm{ppca}}$ under the appropriate conditioning.

To show that $\zeta_{j,\new,k} \leq \zeta_{j,\new,k}^+$, we also use Lemmas \ref{Ak}, \ref{Akperp}, and \ref{calHk} applied with $\epsilon$ given by \eqref{def_eps}. Using $\rank(\Phat_{(j),\new,k})=r_{j,\new}$ and applying Lemma \ref{zetakbnd} with these bounds gives the desired result.
\end{proof}

\section{Proofs of Lemmas \ref{Ak}, \ref{Akperp}, and \ref{calHk}} \label{3_pfs}

\subsection{Some definitions, remarks and facts} %Proofs of Lemmas \ref{Ak}, \ref{Akperp}, \ref{calHk}}
\begin{definition} \label{def_Dnewk}
Define the following for $k=0,1,\dots, K$. Recall that $\Phat_{(j),\rmnew,0}=[.]$.
\ben
%\item $\Phat_{(j),\rmnew,0} = [.]$ (empty matrix)
%\item $\bm{\Phi}_{(k)} := \bm{I} - \Phat_{*} {\Phat_{*}}{}' - \Phat_{\rmnew,k} \Phat_{\rmnew,k}{}' $.
%Because $\Phat_{*}\perp\Phat_{\rmnew,k}$, $\bm{\Phi}_{(k)} = ( \bm{I} - \Phat_{*} {\Phat_{*}}{}' )( \bm{I} - \Phat_{\rmnew,k} \Phat_{\rmnew,k}{}' )$ Notice (from Algorithm \ref{reprocsdet}) that for $t\in \mathcal{I}_{j,k}$, $\bm{\Phi}_t = \bm{\Phi}_{(k-1)}$.

\item $\bm{D}_{j,\rmnew,k} := (\bm{I} - \Phat_{(j),*} {\Phat_{(j),*}}{}' - \Phat_{(j),\rmnew,k} \Phat_{(j),\rmnew,k}{}')   \bm{P}_{(j),\rmnew}$. Thus $\bm{D}_{j,\rmnew} = \bm{D}_{j,\rmnew,0}$.

\item $\bm{D}_{j,*,k} := (\bm{I} - \Phat_{(j),*} {\Phat_{(j),*}}{}' - \Phat_{(j),\rmnew,k} \Phat_{(j),\rmnew,k}{}') \bm{P}_{(j),*}$ and $\bm{D}_{j,*} := \bm{D}_{j,*,0}$.
\item Recall that $\zeta_{j,\new,0} = \|\bm{D}_{j,\rmnew}\|_2$, $\zeta_{j,\new,k} = \|\bm{D}_{j,\rmnew,k}\|_2$, $\zeta_{j,*} = \|\bm{D}_{j,*}\|_2$.  Also, clearly, $\|\bm{D}_{j,*,k}\|_2 \le \|\bm{D}_{j,*}\|_2 \le \zeta_{j,*}$.
\een
\end{definition}

\begin{definition}
For ease of notation, define
\[
%\Ltil_t := \Phi_{(0)} \bm{\ell}_t \ \text{with} \ \Phi_{(0)} = (\I - \Phat_{(j),*}\Phat_{(j),*}{}')
\Ltil_t :=  (\I - \Phat_{(j),*}\Phat_{(j),*}{}') \bm{\ell}_t
\]
\end{definition}

\begin{remark}\label{notes1}
In the rest of this section, for ease of notation, we do the following.
\begin{itemize}
\item We remove the subscript $j$ from $\bm{D}_{j,\rmnew,k}$, $\bm{E}_{j,\new}$, and $\zeta_{j,\new,k}$ etc. and from everything in Definitions \ref{def_P_starnew}, \ref{def_Phat_starnew},  \ref{difzeta}, \ref{defHk} and \ref{def_Dnewk}.

\item Similarly we also let $X_k:= X_{\hat{u}_j +k}$ and $\Gamma_k:= \Gamma_{j,k}^{\uhat_j}$ for both $\uhat_j=u_j$ and $\uhat_j=u_{j+1}$. More precisely, whenever we say $\Pr\left(\mathrm{event}|X_{k-1}\in\Gamma_{k-1}\right) \ge p_0$ we mean $\Pr\left(\mathrm{event}|X_{u_j +k-1}\in\Gamma_{j,k-1}^{u_j}\right) \ge p_0$ and $\Pr\left(\mathrm{event}|X_{u_j+1 +k-1}\in\Gamma_{j,k-1}^{u_j+1}\right) \ge p_0$.

\item Finally, $\sum_t$ refers to $\sum_{t \in \J_u}$ for $u = \uhat_j+k$ %$k=1,2,\dots, K$, $j=1,2,\dots,J$.
\end{itemize}
Also, note the following.
\begin{itemize}
%\item We implicitly use Fact \ref{d_large} in the proof.
%\item Recall that $\mathcal{T}_t$ is included in the definition of $X_{k-1}$, so conditioned on $X_{k-1}$, the $\mathcal{T}_t$ are deterministic.

\item The proof for the bound on $\bm{A}_u$ for $u=u_j+1$ is the same as that for $u = \uhat_j+1$ since in both cases $\Phat_{t,*} = \Phat_{(j),*}$ and $\Phat_{t,\new}=[.]$ for all $t \in \J_u$. The same is true for the bounds on $\bm{A}_{u_j+1,\perp}$ and $\bm{\mathcal{H}}_{u_j+1}$.
\end{itemize}% = \Phat_{(j-1)}
\end{remark}

\begin{fact}\label{matbnds}
When $X_{k-1}\in\Gamma_{k-1}$,
\begin{enumerate}
\item $\|\bm{D}_{*,k-1}\|_2 \leq \zeta_{j,*}^+$ for $k = 1,\dots,K$.
\item $\|\bm{D}_{\new,k-1}\|_2 \leq \zeta_{\new,k-1}^+$ for $k = 1,\dots,K+1$ (by definition of $\Gamma_{k-1}$).
\item Recall that $\zeta_{\new,0}^+ = 1$.
\item $\|[ ({\bm{\Phi}_{t})_{\mathcal{T}_t}}'(\bm{\Phi}_{t})_{\mathcal{T}_t}]^{-1}\|_2 \leq \phi^+$ (from Lemma \ref{cslem})
\item $\lambda_{\min}(\bm{R}_{\new}{\bm{R}_{\new}}') \geq 1 -(\zeta_{*}^+)^2$ (this follows because $\|\Phat_*{}'\P_\new\|_2 = \|\Phat_*(\I - \P_*{\P_*}')'\P_\new\|_2  \le \zeta_*$)
\item $\bm{E}_\new{}' \bm{D}_\new = \bm{E}_\new{}' \bm{E}_\new \bm{R}_\new = \bm{R}_\new$ and $\bm{E}_{\new,\perp}{}' \bm{D}_\new = \bm{0}$.
\item $\Ltil_t = \bm{D}_* \ats +  \bm{D}_\new \atnew$.
\item $\et$ satisfies (\ref{etdef0}) with probability one, i.e. $\et = \bm{I}_{\mathcal{T}_t}[{(\bm{\Phi}_{t})_{\mathcal{T}_t}}'(\bm{\Phi}_{t})_{\mathcal{T}_t}]^{-1} {\bm{I}_{\mathcal{T}_t}}'(\bm{D}_{*,k-1} \bm{a}_{t,*} + \bm{D}_{\rmnew,k-1} \bm{a}_{t,\rmnew})$.
\end{enumerate}
%The same is true for $X_{{u}_j +k-1}\in\Gamma_{j,k-1}^{u_j+1}$.
%?? do these bounds hold for $X_{u_j} \in \Gamma_{j,\rmend}$: yes put $k=1$ and then they hold. ?? I think best to delete this lemma
\end{fact}

\subsection{Preliminaries}
First observe that the matrices $\bm{D}_{\new}$, $\bm{R}_{\new}$, $\bm{E}_{\new}$, $\bm{D}_{*}, \bm{D}_{\rmnew,k-1}$ are all functions of the random variable $X_{k-1}$.  Since $X_{k-1}$ is independent of any $\bm{a}_{t}$ for $t \in  \J_{\uhat_j+k}$, the same is true for these matrices.
All terms that we bound for Lemmas \ref{Ak} and \ref{Akperp} are of the form $\frac{1}{\alpha} \sum_{t \in  \J_{\uhat_j+k}} \bm{Z}_t$ where $\bm{Z}_t= f_1(X_{k-1}) \bm{Y}_t f_2(X_{k-1})$, $\bm{Y}_t$ is a sub-matrix of $\bm{a}_t {\bm{a}_t}'$, and $f_1(.)$ and $f_2(.)$ are functions of $X_{k-1}$. Thus, conditioned on  $X_{k-1}$, the $\bm{Z}_t$'s are mutually independent.

All the terms that we bound for Lemma \ref{calHk} contain $\bm{e}_t$. Using Lemma \ref{cslem}, conditioned on $X_{k-1}$, $\et$ satisfies (\ref{etdef0}) with probability one whenever $X_{k-1} \in \Gamma_{k-1}$. Using (\ref{etdef0}), it is easy to see that all the terms needed for this lemma are also of the above form whenever $X_{k-1} \in \Gamma_{k-1}$.  Thus, conditioned on $X_{k-1}$, the $\bm{Z}_t$'s for all the above terms are mutually independent, whenever $X_{k-1} \in \Gamma_{k-1}$.

We will use the following corollaries of the matrix Hoeffding inequality from \cite{tail_bound}. These are proved in \cite{ReProCS_IT}.
\begin{corollary}[Matrix Hoeffding conditioned on another random variable for a nonzero mean Hermitian matrix \cite{tail_bound,ReProCS_IT}]\label{hoeffding_nonzero}
Given an $\alpha$-length sequence $\{\bm{Z}_t\}$ of random Hermitian matrices of size $n\times n$, a r.v. $X$, and a set ${\cal C}$ of values that $X$ can take. Assume that, for all $X \in \calc$, (i) $\bm{Z}_t$'s are conditionally independent given $X$; (ii) $\mathbb{P}(b_1 \bm{I} \preceq \bm{Z}_t \preceq b_2 \bm{I} | X) = 1$ and (iii) $b_3 \bm{I} \preceq \frac{1}{\alpha}\sum_t \E(\bm{Z}_t | X) \preceq b_4 \bm{I} $. Then for all $\epsilon > 0$,
\begin{align*}
\mathbb{P} \left( \lambda_{\max}\left(\frac{1}{\alpha}\sum_t \bm{Z}_t \right) \leq b_4 + \epsilon \Big | X \right)
\geq 1- n \exp\left(\frac{-\alpha \epsilon^2}{8(b_2-b_1)^2}\right) \ \text{for all} \ X \in \calc
\end{align*}
\begin{align*}
\mathbb{P} \left(\lambda_{\min}\left(\frac{1}{\alpha}\sum_t \bm{Z}_t \right) \geq b_3 -\epsilon \Big| X \right)
\geq  1- n \exp\left(\frac{-\alpha \epsilon^2}{8(b_2-b_1)^2} \right)  \text{for all} \ X \in \calc
\end{align*}
\end{corollary}

\begin{corollary}[Matrix Hoeffding conditioned on another random variable for an arbitrary nonzero mean matrix]\label{hoeffding_rec}
Given an $\alpha$-length sequence $\{\bm{Z}_t\}$ of random matrices of size $n_1 \times n_2$, a r.v. $X$, and a set ${\mathcal{C}}$ of values that $X$ can take. Assume that, for all $X \in \calc$, (i) $\bm{Z}_t$'s are conditionally independent given $X$; (ii) $\mathbb{P}(\|\bm{Z}_t\|_2 \le b_1|X) = 1$ and (iii) $\|\frac{1}{\alpha}\sum_t \E( \bm{Z}_t|X)\|_2 \le b_2$. Then, for all $\epsilon >0$,
\begin{align*}
\mathbb{P} \left(\bigg\|\frac{1}{\alpha}\sum_t \bm{Z}_t \bigg\|_2 \leq b_2 + \epsilon \Big| X \right) %\\%
\geq 1-(n_1+n_2) \exp\left(\frac{-\alpha \epsilon^2}{32 {b_1}^2}\right)  \ \text{for all} \ X \in \calc
\end{align*}
\end{corollary}

\subsection{Simple Lemmas Needed for the Proofs}

\begin{lem}\label{Eatat}    %\label{Ematched}
For $j=1,\dots,J$ and $k = 1,\dots,K$, for all $X_{\uhat_j+k-1}\in\Gamma_{j,k-1}^{\uhat_j}$
\begin{enumerate}

\item \label{starstar}
$\ds \bm{0}\preceq  \E\left[\bm{a}_{t,*}{\bm{a}_{t,*}}' \ \big| \ X_{\hat{u}_{j}+k-1} \right] = \bm{\Lambda}_{t,*}
\preceq  \lambda^+ \I$  %\frac{1}{\alpha}\sum_{t\in\J_{\hat{u}_j+k}}

\item \label{newnew}
$\ds  \lambda_{\new}^- \I \preceq \E\left[\bm{a}_{t,\new}{\bm{a}_{t,\new}}' \ \big| \ X_{\hat{u}_{j}+k-1} \right] = \bm{\Lambda}_{t,\new}
\preceq \lambda_{\new}^+ \I$ and $ \lammin \le \lambda_{\new}^- \le \lambda_{\new}^+ \le 3 \lammin$

\item \label{starnew}
$\ds \E \left[\bm{a}_{t,*}{\bm{a}_{t,\rmnew}}' \ \big| \ X_{\hat{u}_{j}+k-1} \right] = \bm{0}$

%and  $\ds \| \frac{1}{\alpha}\sum_{t\in\J_{\hat{u}_j+k}}\E_{t-1}\left[\bm{a}_{t,*}{\bm{a}_{t,\rmnew}}' \ \big| \ X_{\hat{u}_{j}+k-1} \right] \|_2 \le ?? $ (need this in the $H_u$ bound)
\end{enumerate}
%?? delete for all $X_{\hat{u}_j+k-1}\in\Gamma_{j-1,\rmend}$. ??
with $\hat{u}_j=u_j$ or $\hat{u}_j=u_j+1$.

The same bounds also hold for summation over $t \in \J_{u_j+1}$ when we condition on $X_{u_j}\in\Gamma_{j-1,\rmend}$.
\end{lem}

\begin{proof}
The proof follows from Model \ref{exp_model} and Fact \ref{d_large}. The only reason we need $X_{\hat{u}_j+k-1}\in\Gamma_{j,k-1}^{\uhat_j}$ is to apply Fact \ref{d_large} which allows us to lower and upper bound in the eigenvalues of $\bm{\Lambda}_{t,\new}$ by $\lambda_\new^-$ and $\lambda_\new^+$ and then use \eqref{anew_small}.
\end{proof}

\begin{lem} \label{Dnew0_lem}
Assume that the assumptions of Theorem \ref{thm1} hold. Recall that $\bm{D}_\new = \bm{D}_{\new,0}$.
Conditioned on $X_{k-1} \in \Gamma_{k-1}$, %for $\hat{u}_j = u_j$ or $\hat{u}_j = u_j+1$,
\begin{equation}\label{kappasplus}
\| {\I_{\mathcal{T}}}' \bm{D}_{\rmnew} \|_2 \le \kappa_s^+ := .0215
\end{equation}
for all $\mathcal{T}$ such that $|\mathcal{T}|\leq s$.
\end{lem}
The proof is in Appendix \ref{zeta_k_section}.

\subsection{Proofs of Lemma \ref{Ak} and \ref{Akperp}}

\begin{proof}[Proof of Lemma \ref{Ak}]
We obtain the bounds on $\bm{A}_u$ for $u = \uhat_j+k$ for $k=1,2,\dots, K$ and $\uhat_j = u_j$ or $u_j+1$.
For $u = \uhat_j+k$, recall that $\bm{A}_u :=  \frac{1}{\alpha} \sum_{t} {\bm{E}_{\rmnew}}' \Ltil_t \Ltil_t' \bm{E}_{\rmnew}$.

Notice that ${\bm{E}_{\rmnew}}' \Ltil_t = \bm{R}_{\rmnew} \bm{a}_{t,\rmnew} + {\bm{E}_{\rmnew}}' \bm{D}_* \bm{a}_{t,*}$. Let $\bm{Z}_t = \bm{R}_{\rmnew} \bm{a}_{t,\rmnew} {\bm{a}_{t,\rmnew}}' {\bm{R}_{\rmnew}}'$, and let $\bm{Y}_t = \bm{R}_{\rmnew} \bm{a}_{t,\rmnew}{\bm{a}_{t,*}}' {\bm{D}_*}' {\bm{E}_{\rmnew}} +  {\bm{E}_{\rmnew}}' \bm{D}_* \bm{a}_{t,*}{\bm{a}_{t,\rmnew}}' {\bm{R}_{\rmnew}}'$, then
\beq
\bm{A}_u  \succeq \frac{1}{\alpha} \sum_t \bm{Z}_t + \frac{1}{\alpha} \sum_t \bm{Y}_t \label{lemmabound_1}
\eeq

Consider $\frac{1}{\alpha} \sum_t \bm{Z}_t$. %= \sum_t \bm{R}_{\rmnew} \bm{a}_{t,\rmnew} {\bm{a}_{t,\rmnew}}' {\bm{R}_{\rmnew}}'$.
(1) The $\bm{Z}_t$'s are conditionally independent given $X_{k-1}$.
(2) With probability 1, $\|\bm{Z}_t\|_2\leq r_{\new}{\gamma_{\new}}^2$.
(3) Using a theorem of Ostrowoski \cite[Theorem 4.5.9]{hornjohnson}, conditioned on $X_{k-1} \in \Gamma_{k-1}$,
$\lambda_{\min}\left( \E[\frac{1}{\alpha} \sum_t \bm{Z}_t | X_{k-1}] \right) = \lambda_{\min}\left(  \bm{R}_{\rmnew} ( \frac{1}{\alpha} \sum_t \Lamtnew ) {\bm{R}_{\rmnew}}'\right) \geq \lambda_{\min} \left(\bm{R}_{\rmnew} {\bm{R}_{\rmnew}}'\right)\lambda_{\min} \left(\frac{1}{\alpha} \sum_t  \Lamtnew \right) \ge (1- (\zeta_{*}^+)^2) \lambda_\new^-$. The last inequality uses Lemma \ref{Eatat} and Fact \ref{matbnds}.

Thus, applying Corollary \ref{hoeffding_nonzero} with $\epsilon$ given by (\ref{def_eps}), we get that, for all $X_{k-1}\in\Gamma_{k-1}$,
\begin{equation}
\mathbb{P}\left(\lambda_{\min} \left(\frac{1}{\alpha} \sum_t \bm{Z}_t\right)
\geq   (1-(\zeta_*^+)^2)\lambda_{\rmnew}^-  - \epsilon \bigg| X_{k-1} \right)
\geq  1- r_\new \exp \left(\frac{-\alpha \zeta^2 (\lammin)^2}{8 \cdot 100^2 \cdot {\gamma_{\rmnew}}^4}\right).
 \label{lemma_add_A1}
\end{equation}

Consider $\bm{Y}_t = \bm{R}_{\rmnew} \bm{a}_{t,\rmnew}{\bm{a}_{t,*}}' {\bm{D}_*}' {\bm{E}_{\rmnew}} +  {\bm{E}_{\rmnew}}' \bm{D}_* \bm{a}_{t,*}{\bm{a}_{t,\rmnew}}' {\bm{R}_{\rmnew}}'$.
(1)  The $\bm{Y}_t$'s are conditionally independent given $X_{k-1}$.
(2) Using  the bound on $\zeta$ from the theorem, $\|\bm{Y}_t\| \le 2\sqrt{r_\new r} \zeta_*^+ \gamma \gamma_{\rmnew}  \leq 2\sqrt{r_\new r} \zeta_*^+ {\gamma}^2 \le  2$ holds with probability one for all $X_{k-1} \in \Gamma_{k-1}$. Thus, under the same conditioning, $-2 \bm{I} \preceq \bm{Y}_t  \preceq 2 \bm{I}$ with with probability one.
(3) By Lemma \ref{Eatat},  $\E\left(\frac{1}{\alpha}\sum_t \bm{Y}_t|X_{k-1}\right) = \bm{0}$ for all $X_{k-1} \in \Gamma_{k-1}$.

Thus, applying Corollary \ref{hoeffding_nonzero} with $\epsilon$ given by (\ref{def_eps}), we get that, for all $X_{k-1}\in\Gamma_{k-1}$
\beq
\mathbb{P}\left(\lambda_{\min} \left(\frac{1}{\alpha} \sum_t \bm{Y}_t \right) \geq - \epsilon \Big| X_{k-1} \right) \geq 1- c \exp \left( \frac{-\alpha {r_\new}^2 \zeta^2(\lammin)^2} {8 \cdot 100^2 \cdot (4)^2}\right)
\label{lemma_add_A2}
\eeq

Combining (\ref{lemmabound_1}), (\ref{lemma_add_A1}) and (\ref{lemma_add_A2}) and using the union bound, we get the lemma.
\end{proof}

\begin{proof}[Proof of Lemma \ref{Akperp}]
Remark \ref{notes1} applies.

We obtain the bounds on $\bm{A}_{u,\perp}$ for $u = \uhat_j+k$ for $k=1,2,... K$ with $\uhat_j = u_j$ or $u_j+1$. For all these $u$'s, recall that $\bm{A}_{u,\perp} := \frac{1}{\alpha} \sum_t {\bm{E}_{\rmnew,\perp}}' \Ltil_t \Ltil_t{}' \bm{E}_{\rmnew,\perp}$. Using $\bm{E}_{\rmnew,\perp}{}' \bm{D}_\new = 0$, we get that ${\bm{E}_{\rmnew,\perp}}' \Ltil_t = {\bm{E}_{\rmnew,\perp}}' \bm{D}_* \bm{a}_{t,*}$. Thus, $\bm{A}_{u,\perp} = \frac{1}{\alpha} \sum_t \bm{Z}_t$ with  $\bm{Z}_t={\bm{E}_{\rmnew,\perp}}' \bm{D}_* \bm{a}_{t,*} {\bm{a}_{t,*}}' {\bm{D}_*}' \bm{E}_{\rmnew,\perp}$.

Using the same ideas as for the previous proof we can show that $\bm{0} \preceq \bm{Z}_t \preceq r (\zeta_*^+)^2 {\gamma}^2 \bm{I} \preceq \zeta \bm{I}$ and $\E\left(\frac{1}{\alpha}\sum_t \bm{Z}_t|X_{k-1}\right) \preceq (\zeta_*^+)^2 \lambda^+ \bm{I}$.  Thus by Corollary \ref{hoeffding_nonzero} the lemma follows.
\end{proof}

\subsection{Proof of Lemma \ref{calHk}}

\begin{proof}[Proof of Lemma \ref{calHk}]\label{calHk_proof}
Remark \ref{notes1} applies.
Using the expression for $\bm{\mathcal{H}}_u$ given in Definition \ref{defHk}, and noting that for a basis matrix $\bm{E}$, $\bm{EE}' + \bm{E}_{\perp}{\bm{E}_{\perp}}' = \I$ we get that
\[
\bm{\mathcal{H}}_u = \frac{1}{\alpha} \sum_{t\in\J_u} \Big( (\I - \Phat_*\Phat_*{}') \bm{e}_t {\bm{e}_t}' (\I - \Phat_*\Phat_*{}') - (\Ltil_t {\bm{e}_t}' (\I - \Phat_*\Phat_*{}') + (\I - \Phat_*\Phat_*{}') \bm{e}_t \Ltil_t{}')  + (\bm{F}_t + {\bm{F}_t}') \Big)
\]
%\frac{1}{\alpha} \sum_{t\in\J_u} \Big( (\I - \hat{\bm{P}}_{(j),*}\hat{\bm{P}}_{(j),*}{}') \bm{e}_t {\bm{e}_t}' (\I - \hat{\bm{P}}_{(j),*}\hat{\bm{P}}_{(j),*}{}') - \big(\Ltil_t {\bm{e}_t}' (\I - \hat{\bm{P}}_{(j),*}\hat{\bm{P}}_{(j),*}{}') + (\I - \hat{\bm{P}}_{(j),*}\hat{\bm{P}}_{(j),*}{}') \bm{e}_t \Ltil_t{}'\big)  + \big(\bm{F}_t + {\bm{F}_t}'\big) \Big) \]
where
\[
\bm{F}_t = \bm{E}_{\rmnew,\perp}{\bm{E}_{\rmnew,\perp}}' \Ltil_t \Ltil_t{}' \bm{E}_{\rmnew}{\bm{E}_{\rmnew}}'.
\]
Thus,
\begin{equation}\label{add_calH1}
\| \bm{\mathcal{H}}_u \|_2 \leq  2\bigg\| \frac{1}{\alpha} \sum_t  \Ltil_t {\bm{e}_t}'  \bigg\|_2 + \bigg\| \frac{1}{\alpha} \sum_t \bm{e}_t {\bm{e}_t}' \bigg\|_2 + 2\bigg\| \frac{1}{\alpha} \sum_t \bm{F}_t \bigg\|_2
\end{equation}
Next we obtain high probability bounds on each of the three terms on the right hand side of (\ref{add_calH1}). %using the Azuma corollaries.

Consider $\big\| \frac{1}{\alpha} \sum_t  \Ltil_t {\bm{e}_t}'  \big\|_2$. Using Lemma \ref{cslem}, $\et$ satisfies \eqref{etdef0} with probability one for all $X_{k-1} \in  \Gamma_{k-1}$.
%Thus under this conditioning,
%\begin{align*}
% \Ltil_t {\bm{e}_t}'
%&=   (\bm{D}_* \bm{a}_{t,*} + \bm{D}_{\rmnew} \bm{a}_{t,\rmnew})(\bm{D}_{*,k-1} \bm{a}_{t,*} + \bm{D}_{\rmnew,k-1} \bm{a}_{t,\rmnew})' \bm{I}_{\mathcal{T}_t}[{(\bm{\Phi}_{t})_{\mathcal{T}_t}}'(\bm{\Phi}_{t})_{\mathcal{T}_t}]^{-1} {\bm{I}_{\mathcal{T}_t}}'
%\end{align*}

Let $\bm{Z}_t: = \Ltil_t {\bm{e}_t}' $.
(1) Conditioned on $X_{k-1}$, the various $\bm{Z}_t$'s used in the summation are mutually independent, for all $X_{k-1} \in \Gamma_{k-1}$.
(2) For  $X_{k-1} \in \Gamma_{k-1}$,
\[
\|\bm{Z}_t\|_2 = \|\Ltil_t {\bm{e}_t}' \|_2 \leq \Big(\zeta_*^+\sqrt{r}\gamma + \sqrt{r_\new}\gamma_{\rmnew} \Big) \Big( \phi^+ (\zeta_*^+ \sqrt{r} \gamma + \zeta_{\new,k-1}^+ \sqrt{r_\new}\gamma_{\rmnew}) \Big) := b_3
\]
holds with probability one.
(3) First consider the $k\geq2$ case. When $X_{k-1} \in \Gamma_{k-1}$,
\begin{align*}
&\bigg\| \E \bigg[   \frac{1}{\alpha}\sum_t \Ltil_t {\bm{e}_t}'  \ \big| \ X_{k-1} \bigg] \bigg\|_2  \\
 = & \bigg\| \frac{1}{\alpha}\sum_t \left[ \Big(   \bm{D}_*  \bm{\Lambda}_{t,*}  {\bm{D}_{*,k-1}}'  + \bm{D}_{\rmnew} \bm{\Lambda}_{t,\new}  {\bm{D}_{\rmnew,k-1}}' \Big) \bm{I}_{\mathcal{T}_t}[{(\bm{\Phi}_{t})_{\mathcal{T}_t}}'(\bm{\Phi}_{t})_{\mathcal{T}_t}]^{-1} {\bm{I}_{\mathcal{T}_t}}' \right] \bigg\|_2\\
\leq &  \sqrt{\lambda_{\max}\left(\frac{1}{\alpha}\sum_t \Big(   \bm{D}_*  \bm{\Lambda}_{t,*}  {\bm{D}_{*,k-1}}'  + \bm{D}_{\rmnew} \bm{\Lambda}_{t,\new}  {\bm{D}_{\rmnew,k-1}}' \Big)\Big(   \bm{D}_*  \bm{\Lambda}_{t,*}  {\bm{D}_{*,k-1}}'  + \bm{D}_{\rmnew} \bm{\Lambda}_{t,\new}  {\bm{D}_{\rmnew,k-1}}' \Big)' \right)} \\
& \sqrt{\lambda_{\max}\left( \frac{1}{\alpha}\sum_t \Big(\bm{I}_{\mathcal{T}_t}[{(\bm{\Phi}_{t})_{\mathcal{T}_t}}'(\bm{\Phi}_{t})_{\mathcal{T}_t}]^{-1} {\bm{I}_{\mathcal{T}_t}}'\Big)\Big(\bm{I}_{\mathcal{T}_t}[{(\bm{\Phi}_{t})_{\mathcal{T}_t}}'(\bm{\Phi}_{t})_{\mathcal{T}_t}]^{-1} {\bm{I}_{\mathcal{T}_t}}'\Big)' \right)  } \\
\leq& \left( (\zeta_*^+)^2 \lambda^+ + \zeta_{\new,k-1}^+\lambda_{\rmnew}^+  \right)\left(\sqrt{\rho^2 h^+}\phi^+\right).
\end{align*}
The first inequality is by Cauchy-Schwarz for a sum of matrices.  This can be found as Lemma \ref{CSmat} in Appendix \ref{prelim}.
The second inequality uses Fact \ref{matbnds} (for the first term of the product) and Lemma $\ref{blockdiag1}$ with $\sigma^+ = (\phi^+)^2$ (for the second term of the product).

Now consider the $k=1$ case. To bound $\bigg\| \frac{1}{\alpha}\sum_t \bm{D}_*  \bm{\Lambda}_{t,*}  {\bm{D}_{*,0}}' \bm{I}_{\mathcal{T}_t}[{(\bm{\Phi}_{t})_{\mathcal{T}_t}}'(\bm{\Phi}_{t})_{\mathcal{T}_t}]^{-1} {\bm{I}_{\mathcal{T}_t}}' \bigg\|_2$ we proceed exactly as we did for the $k\ge 2$ case. We can bound this by $(\zeta_*^+)^2 \lambda^+ \sqrt{\rho^2 h^+} \phi^+$.
To bound $\bigg\| \frac{1}{\alpha}\sum_t \bm{D}_{\rmnew} \bm{\Lambda}_{t,\new}  {\bm{D}_{\new,0}}' \I_{\T_t}[{(\bm{\Phi}_{t})_{\mathcal{T}_t}}'(\bm{\Phi}_{t})_{\mathcal{T}_t}]^{-1} {\bm{I}_{\mathcal{T}_t}}' \bigg\|_2$, we apply Lemma \ref{Dnew0_lem} to get\footnote{Notice that if we want to use the bound of Lemma \ref{Dnew0_lem}, we cannot also apply Lemma \ref{blockdiag1} for this term. We can get a simpler proof by not using Lemma \ref{Dnew0_lem} at all and proceeding exactly as we did for the $k\ge2$ case; but doing this will require a much tighter bound on $\rho^2 h^+$ than what we currently need.}  $\|{\bm{D}_{\new,0}}'\I_{\T_t}\|_2 \leq \kappa_s^+$ . Using this and Fact \ref{matbnds}, we can bound this by $\kappa_{s}^+\lambda_{\new}^+ \phi^+$.
Thus,  when $X_{0} \in \Gamma_{0}$,
\begin{align*}
&\bigg\| \E \bigg[   \frac{1}{\alpha}\sum_t \Ltil_t {\bm{e}_t}'  \ \big| \ X_{0} \bigg] \bigg\|_2  \\
%=& \bigg\|  \E\bigg[ \frac{1}{\alpha}\sum_t\bm{\Phi}_{(0)}(\bm{P}_{*}\bm{a}_{t,*} + \bm{P}_{\rmnew} \bm{a}_{t,\rmnew}) (\bm{P}_{*} \bm{a}_{t,*} + \bm{P}_{\rmnew} \bm{a}_{t,\rmnew})' {\bm{\Phi}_{t}}'\bm{I}_{\mathcal{T}_t}[{(\bm{\Phi}_{t})_{\mathcal{T}_t}}'(\bm{\Phi}_{t})_{\mathcal{T}_t}]^{-1} {\bm{I}_{\mathcal{T}_t}}'\bm{\Phi}_{(0)}   \ \big| \ X_{0} \bigg]\bigg\|_2 \\
%=& \bigg\| \E \bigg[   \frac{1}{\alpha} \sum_t (\bm{D}_* \bm{a}_{t,*} + \bm{D}_{\new} \bm{a}_{t,\new})(\bm{D}_{*} \bm{a}_{t,*} + \bm{D}_{\new} \bm{a}_{t,\rmnew})' \bm{I}_{\mathcal{T}_t}[{(\bm{\Phi}_{t})_{\mathcal{T}_t}}'(\bm{\Phi}_{t})_{\mathcal{T}_t}]^{-1} {\bm{I}_{\mathcal{T}_t}}'  \ \big| \ X_{0} \bigg] \bigg\|_2 \\
=& \bigg\| \frac{1}{\alpha}\sum_t \left[ \Big(   \bm{D}_*  \bm{\Lambda}_{t,*}  {\bm{D}_{*}}'  + \bm{D}_{\rmnew} \bm{\Lambda}_{t,\new}  {\bm{D}_{\new}}' \Big) \bm{I}_{\mathcal{T}_t}[{(\bm{\Phi}_{t})_{\mathcal{T}_t}}'(\bm{\Phi}_{t})_{\mathcal{T}_t}]^{-1} {\bm{I}_{\mathcal{T}_t}}' \right] \bigg\|_2\\
\leq & \left( \sqrt{\rho^2 h^+} (\zeta_{*}^+)^2\lambda^+ + \kappa_{s}^+\lambda_{\new}^+ \right)\phi^+.
\end{align*}

Thus, by Corollary \ref{hoeffding_rec} with $\epsilon$ given by (\ref{def_eps}), we get that,  for all $X_{k-1}\in\Gamma_{k-1}$,
\begin{equation}
\label{add_ltet}
\mathbb{P} \left(  \Big\| \frac{1}{\alpha} \sum_{t}  \Ltil_t {\bm{e}_t}'  \Big\|_2 \leq b_{\l \e,k}  \Bigg| X_{k-1} \right) \geq 1 - n \exp\left(  \frac{-\alpha {r_\new}^2\zeta^2(\lammin)^2}{32 \cdot 100^2 {b_3}^2} \right).
\end{equation}

Consider $\|\frac{1}{\alpha} \sum_t \bm{e}_t {\bm{e}_t}'\|_2$. Let $\bm{Z}_t = \bm{e}_t {\bm{e}_t}'$.
(1) Conditioned on $X_{k-1}$, the various $\bm{Z}_t$'s in the summation are independent, for all $X_{k-1} \in \Gamma_{k-1}$.
(2) Using Lemma \ref{cslem}, conditioned on $X_{k-1}\in\Gamma_{k-1}$,
\[
\bm{0} \preceq \bm{Z}_t \preceq \Big( \phi^+ (\zeta_*^+ \sqrt{r} \gamma + \zeta_{\new,k-1}^+ \sqrt{r_\new}\gamma_{\rmnew}) \Big)^2 \bm{I}:=  b_1 \I
\]
with probability one.
(3)  By Fact \ref{matbnds}, when $X_{k-1} \in \Gamma_{k-1}$,
\begin{align*}
 &  \frac{1}{\alpha}\sum_t \E \left[ \et{\et}' | X_{k - 1} \right] \\
%&=   \frac{1}{\alpha}\sum_t \E \left[\Big( \bm{I}_{\mathcal{T}_t} [ ({\bm{\Phi}_{t})_{\mathcal{T}_t}}'(\bm{\Phi}_{t})_{\mathcal{T}_t}]^{-1}  {\bm{I}_{\mathcal{T}_t}}' \bm{\Phi}_{t} \bm{\ell}_t \Big)\Big(  \bm{I}_{\mathcal{T}_t} [ ({\bm{\Phi}_{t})_{\mathcal{T}_t}}'(\bm{\Phi}_{t})_{\mathcal{T}_t}]^{-1}  {\bm{I}_{\mathcal{T}_t}}' \bm{\Phi}_{t} \bm{\ell}_t \Big)' | X_{k - 1} \right]\\
%&=    \frac{1}{\alpha}\sum_t \E\left[\Big(  \bm{I}_{\mathcal{T}_t} [ ({\bm{\Phi}_{t})_{\mathcal{T}_t}}'(\bm{\Phi}_{t})_{\mathcal{T}_t}]^{-1}  {\bm{I}_{\mathcal{T}_t}}' \bm{\Phi}_{t} \bm{\ell}_t {\lt}' \bm{\Phi}_{t}\bm{I}_{\mathcal{T}_t}[ ({\bm{\Phi}_{t})_{\mathcal{T}_t}}'(\bm{\Phi}_{t})_{\mathcal{T}_t}]^{-1}{\bm{I}_{\mathcal{T}_t}}' | X_{k - 1} \right] \\
%&= \frac{1}{\alpha}\sum_t   \bm{I}_{\mathcal{T}_t} [ ({\bm{\Phi}_{t})_{\mathcal{T}_t}}'(\bm{\Phi}_{t})_{\mathcal{T}_t}]^{-1}  {\bm{I}_{\mathcal{T}_t}}' \bm{\Phi}_{t} (\bm{P}_{*}(\bm{\Lambda}_{t})_{*}{\bm{P}_{*}}' + \bm{P}_{\rmnew}\bm{\Lambda}_{t,\new}{\bm{P}_{\rmnew}}' ) \bm{\Phi}_{t}\bm{I}_{\mathcal{T}_t}[ ({\bm{\Phi}_{t})_{\mathcal{T}_t}}'(\bm{\Phi}_{t})_{\mathcal{T}_t}]^{-1}{\bm{I}_{\mathcal{T}_t}}'\\
&= \frac{1}{\alpha}\sum_t   \bm{I}_{\mathcal{T}_t} [ ({\bm{\Phi}_{t})_{\mathcal{T}_t}}'(\bm{\Phi}_{t})_{\mathcal{T}_t}]^{-1}  {\bm{I}_{\mathcal{T}_t}}' \Big(\bm{D}_{*,k-1}\bm{\Lambda}_{t,*} {\bm{D}_{*,k-1}}'+ \bm{D}_{\rmnew,k-1}\bm{\Lambda}_{t,\new}{\bm{D}_{\rmnew,k-1}}' \Big) \bm{I}_{\mathcal{T}_t}[ ({\bm{\Phi}_{t})_{\mathcal{T}_t}}'(\bm{\Phi}_{t})_{\mathcal{T}_t}]^{-1}{\bm{I}_{\mathcal{T}_t}}'
\end{align*}

When $k=1$ we can apply Lemma \ref{Dnew0_lem} to get that $\|{\bm{D}_{\new,0}}'\I_{\T_t}\|_2 \leq \kappa_s^+$.  Then we apply Lemma \ref{blockdiag1} with $\sigma^+ = (\phi^+)^2 \left( (\zeta_{*}^+)^2\lambda^+ +   (\kappa_{s}^+)^2 \lambda_{\rmnew}^+ \right)$.
This gives
\[
\bm{0} \preceq \E \left[  \sum_{t} \et{\et}' \Big| X_{0} \right] \preceq  \rho^2 h^+(\phi^+)^2 \Big((\zeta_*^+)^2\lambda^+ + (\kappa_{s}^+)^2 \lambda_{\rmnew}^+ \Big)\I
\quad\text{ for all }X_{0} \in \Gamma_{0}.
\]
When $k\geq2$ we can apply Lemma \ref{blockdiag1} with $\sigma^+ = (\phi^+)^2 \left( (\zeta_{*}^+)^2\lambda^+ +  (\zeta_{\new,k-1}^+)^2 \lambda_{\rmnew}^+ \right)$ to get that,
\[
\bm{0} \preceq \E \left[  \sum_{t} \et{\et}' \Big| X_{k - 1} \right] \preceq  \rho^2 h^+(\phi^+)^2 \Big((\zeta_*^+)^2\lambda^+ + (\zeta_{\new,k-1}^+)^2 \lambda_{\rmnew}^+ \Big)\I
\quad\text{ for all }X_{k - 1} \in \Gamma_{k-1}.
\]
Thus, applying Corollary \ref{hoeffding_nonzero} with $\epsilon$ given by \eqref{def_eps}, we get that,  for all $X_{k-1}\in\Gamma_{k-1}$,
\beq
\label{add_etet}
\mathbb{P} \left( \Big\|\frac{1}{\alpha} \sum_{t}  \bm{e}_t {\bm{e}_t}' \Big\|_2 \leq b_{\e \e,k}  \Big| X_{k-1} \right) \geq 1- n \exp\left(\frac{-\alpha r_\new^2 \zeta^2 (\lammin)^2}{ 8 \cdot 100^2 {b_1}^2}\right).
\eeq

Finally, consider  $\big\| \frac{1}{\alpha} \sum_t \bm{F}_t \big\|_2$.
Since $\bm{E}_{\rmnew,\perp}{}'\bm{D}_\rmnew = 0$,
\begin{align*}
 \bm{F}_t   &  =   \bm{E}_{\rmnew,\perp} {\bm{E}_{\rmnew,\perp}}' \Ltil_t \Ltil_t{}' \bm{E}_{\rmnew}{\bm{E}_{\rmnew}}'   \\
%&=  \bm{E}_{\rmnew,\perp}{\bm{E}_{\rmnew,\perp}}'(\bm{D}_* \bm{a}_{t,*} + \bm{D}_{\rmnew} \bm{a}_{t,\rmnew})(\bm{D}_{*} \bm{a}_{t,*} + \bm{D}_{\rmnew} \bm{a}_{t,\rmnew})' \bm{E}_{\rmnew}{\bm{E}_{\rmnew}}'  \\
&=  \bm{E}_{\rmnew,\perp}{\bm{E}_{\rmnew,\perp}}'(\bm{D}_* \bm{a}_{t,*} )(\bm{D}_{*} \bm{a}_{t,*} + \bm{D}_{\rmnew} \bm{a}_{t,\rmnew})' \bm{E}_{\rmnew}{\bm{E}_{\rmnew}}'
%\\ &= \bm{E}_{\rmnew,\perp}{\bm{E}_{\rmnew,\perp}}'\big(\bm{D}_* \bm{a}_{t,*}{\bm{a}_{t,*}}'{\bm{D}_*}' +  \bm{D}_* \bm{a}_{t,*}{\bm{a}_{t,\rmnew}}'{\bm{D}_{\rmnew}}' \big)' \bm{E}_{\rmnew}{\bm{E}_{\rmnew}}'
\end{align*}
(1) Conditioned on $X_{k-1}$, the $\bm{F}_t$'s  are mutually independent, for all $X_{k-1} \in  \Gamma_{k-1}$.
(2) For $X_{k-1} \in  \Gamma_{k-1}$,
\[
\|\bm{F}_t\|_2 \leq (\zeta_*^+)^2 r\gamma^2 + \zeta_*^+ \sqrt{r r_\new} \gamma \gamma_\new := b_5
\]
 holds with probability 1.
(3) For $X_{k-1} \in \Gamma_{k-1}$,
\begin{align*}
\bigg\| \E \Big[ \frac{1}{\alpha} \sum_t \bm{F}_t  \ \big| \ X_{k-1} \Big] \bigg\|_2 %&  = \bigg\| \E \Big[ \frac{1}{\alpha} \sum_t \bm{E}_{\rmnew,\perp} {\bm{E}_{\rmnew,\perp}}' \Ltil_t \Ltil_t{}' \bm{E}_{\rmnew}{\bm{E}_{\rmnew}}' \ \big| \ X_{k-1} \Big] \bigg\|_2 \\
%&\leq \bigg\| \E \Big[ \frac{1}{\alpha} \sum_t \bm{E}_{\rmnew,\perp}{\bm{E}_{\rmnew,\perp}}'\Ltil_t \Ltil_t{}' \ \big| \ X_{k-1} \Big] \bigg\|_2  \\
%&= \bigg\|\frac{1}{\alpha}\sum_t  \bm{E}_{\rmnew,\perp} {\bm{E}_{\rmnew,\perp}}' (\bm{D}_*\bm{\Lambda}_t {\bm{D}_*}' + \bm{D}_{\rmnew}\bm{\Lambda}_{(j),\rmnew}\bm{D}_{\rmnew}') \bigg\|_2 \\
%&= \bigg\|\frac{1}{\alpha}\sum_t  \bm{E}_{\rmnew,\perp} {\bm{E}_{\rmnew,\perp}}' (\bm{D}_*\bm{\Lambda}_t {\bm{D}_*}' ) \bigg\|_2 \\
&\leq \bigg\|\frac{1}{\alpha}\sum_t  (\bm{D}_*\bm{\Lambda}_{t,*} {\bm{D}_*}' ) \bigg\|_2 \leq (\zeta_*^+)^2\lambda^+ = b_{\bm{F}}
\end{align*}
Applying Corollary \ref{hoeffding_rec} with $\epsilon$ given by (\ref{def_eps}), we get that,  for all $X_{k-1}\in\Gamma_{k-1}$,
\begin{equation}\label{add_Ft}
\mathbb{P} \left( \Big\|  \frac{1}{\alpha}\sum_t \bm{F}_t  \Big\|_2 \leq b_{\bm{F}} \Bigg| X_{k-1} \right) \geq 1 - n \exp\left(  \frac{-\alpha {r_\new}^2\zeta^2(\lammin)^2}{32 \cdot 100^2 {b_5}^2} \right)
\end{equation}
Combining \eqref{add_calH1} with \eqref{add_ltet}, \eqref{add_etet} and \eqref{add_Ft} and using the union bound, we get the lemma. The expression for $p_{\mathcal{H}}$ given in the lemma uses the bounds on $\zeta$ from the theorem and uses the loose bound $\zeta_{j,\new,k-1}^+ \le 1$ (to get a simpler expression for the probabilities).
\end{proof}

\section{Simulation Experiments} \label{sims}

In this section we provide some simulations that demonstrate the robust PCA result we have proven above.
More detailed simulations using real data can be found in \cite{han_tsp}.

The data for Figure \ref{fig:cor} was generated as follows.
We chose $n = 256$ and $t_{\max} = 15,000$.  Each measurement had $s=20$ missing or corrupted entries, i.e. $|\T_t|=20$. Each non-zero entry of $\xt$ was drawn uniformly at random between 2 and 6 independent of other entries and other times $t$.  In Figure \ref{fig:cor} the support of $\xt$ changes as assumed in Model \ref{sbyrho} with $\rrho=2$ and $\beta=18$.
So the support of $\xt$ changes by $\frac{s}{2} = 10$ indices every 18 time instants.
When the support of $\xt$ reaches the bottom of the vector, it starts over again at the top.
This pattern can be seen in the bottom half of the figure which shows the sparsity pattern of the matrix $\bm{S} = [\bm{x}_1, \dots, \bm{x}_{t_{\max}}]$.

To form the low dimensional vectors $\lt$,  we started with an $n \times r$ matrix of i.i.d. Gaussian entries and orthonormalized the columns using Gram-Schmidt.
The first $r_0 = 10$ columns of this matrix formed $\bm{P}_{(0)}$, the next 2 columns formed $\bm{P}_{(1),\rmnew}$, and the last 2 columns formed $\bm{P}_{(2),\rmnew}$
We show two subspace changes which occur at $t_1 = 600$ and $t_2 = 8000$.
The entries of  $\bm{a}_{t,*}$ were drawn uniformly at random between -5 and 5,
and the entries of $\bm{a}_{t,\rmnew}$ were drawn uniformly at random between $-\sqrt{3v_i^{t-t_j}\lammin}$ and $\sqrt{3v_i^{t-t_j}\lammin}$ with $v_i = 1.00017$ and $\lammin = 1$ (and $q_i=1$). Thus $(\bm{\Lambda}_{t,\new})_{i,i} = v_i^{t-t_j}\lammin$ as assumed in Model \ref{exp_model}.
Entries of $\bm{a}_t$ were independent of each other and of the other $\bm{a}_t$'s.

For this simulated data we compare the performance of ReProCS and PCP.  The plots show the relative error in recovering $\lt$, that is $\|\lt-\lhatt\|_2/\|\lt\|_2$.
For the initial subspace estimate $\hat{\bm{P}}_{0}$, we used $\bm{P}_{0}$ plus some small Gaussian noise and then obtained orthonormal columns.
We set $\alpha = 800$ and $K = 6$.
For the PCP algorithm,  we perform the optimization every $\alpha$ time instants using all of the data up to that point.  So the first time PCP is performed on $[ \bm{m}_1 , \dots , \bm{m}_{\alpha}]$ and the second time it is performed on $[ \bm{m}_1 , \dots , \bm{m}_{2\alpha}]$ and so on.

Figure \ref{fig:cor} illustrates the result we have proven.  That is ReProCS takes advantage of the initial subspace estimate and slow subspace change (including the bound on $\gamma_{\rmnew}$) to handle the case when the supports of $\xt$ are correlated in time.
Notice how the ReProCS error increases after a subspace change, but decays exponentially with each projection PCA step.
For this data, the PCP program fails to give a meaningful estimate for all but a few times.
The average time taken by the ReProCS algorithm was 52 seconds, while PCP averaged over 5 minutes.  Simulations were coded in MATLAB${}^\circledR$ and run on a desktop computer with a 3.2 GHz processor.

%??? Is this figure needed - does this seem to indicate that reprocs does not work for random support changes - many people on see the figure, not read the simulation section description ??? Compare this to Figure \ref{fig:rand} where the only change in the data is that the support of $\bm{X}$ is chosen uniformly at random from all sets of size $\frac{s t_{\max}}{n}$ (as assumed in \cite{rpca}).  Thus the total sparsity of the matrix $\bm{X}$ is the same for both figures.  In Figure \ref{fig:rand}, ReProCS performs almost the same, while PCP does substantially better than in the case of correlated supports.

\begin{figure}[t]
\begin{centering}
\includegraphics[scale=.7]{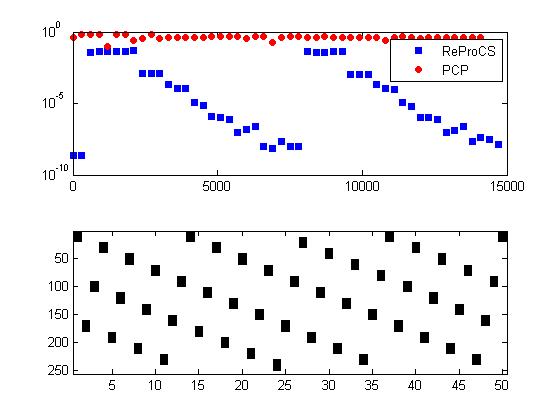}
 \caption{Comparison of ReProCS and PCP for the RPCA problem.  The top plot is the relative error $\|\lt - \lhatt\|_2 / \|\lt\|_2$.  The bottom plot shows the sparsity pattern of $\bm{S}$ (black represents a non-zero entry).  Results are averaged over 100 simulations and plotted every 300 time instants.
 \label{fig:cor}}
 \end{centering}
\end{figure}

\section{Extensions} \label{extensions}
In this section, we first give other models on changes in $\T_t$ that are special cases of the general model Model \ref{general_model} and hence can also be used in Theorem \ref{thm1_mc} or \ref{thm1}. The next three subsections discuss various other results that can also be proved using the proof techniques developed in this work.

\subsection{Other Models on Changes in $\T_t$}
We give here other models on changes in $\T_t$ that are special cases of Model \ref{general_model}.

\begin{sigmodel} \label{randsup}
Suppose that $\T_t$ consists of consecutive indices and is of size $s$ or less, i.e. $|\T_t| \le s$.
When $\T_t$ is not empty, let $\tilde{o}_t$ denote its smallest (topmost) index.
Let $\rho_1$ be an integer. We assume that $\tilde{o}_t$ satisfies the following Bernoulli-Gaussian model:%
\begin{align*}\label{probsupp}
\tilde{o}_t = \lceil o_t \mod n \rceil \text{ where }
o_t = o_{t-1} + \theta_t \left( 1.1\frac{s}{\rrho} + \varpi_t \right)
\end{align*}
where $\varpi_t \sim \mathcal{N}(0,\sigma^2)$ (Gaussian) and $\theta_t \sim Bernoulli(q)$. %Here $a \mod b$ returns the remainder when $b$ is divided by $a$.
Assume that $\{ \varpi_t \}$, $\{ \theta_t \}$ are mutually independent and independent of $\lt$'s. % for $t = 1,\dots,t_{\max}$.
Taking the mod with respect to $n$ describes the process of the set $\T_t$ starting over at $1$ when its topmost index exceeds $n$ (this models a new object appearing after the old one has disappeared; notice that at any $t$ $\T_t$ could be empty as well, i.e. there may be no object).

Assume that $s \le \frac{1.2\rrho n}{\alpha}$, $ q \geq 1 - (  \frac{n^{-10}}{2t_{\max}} )^{\frac{1}{\beta}}$ for a $\beta$ that satisfies $\rrho^2 \frac{\beta}{\alpha} \leq  0.01$, and $\sigma^2  \le \frac{s^2}{4000 \rrho^2\log(n)}$.
\end{sigmodel}

\begin{sigmodel}\label{everyframe}
Suppose that  $\T_t$ consists of $s$ consecutive indices and suppose that it moves down the vector by between 1 and $m$ indices at every time $t$. When it reaches the bottom of the vector, we  assume that it starts over at $1$. Assume that $s\leq 0.0025 \alpha$ and $m \leq \frac{n- s}{\alpha} $.
\end{sigmodel}

\begin{figure}
\centering
\includegraphics[height=3in]{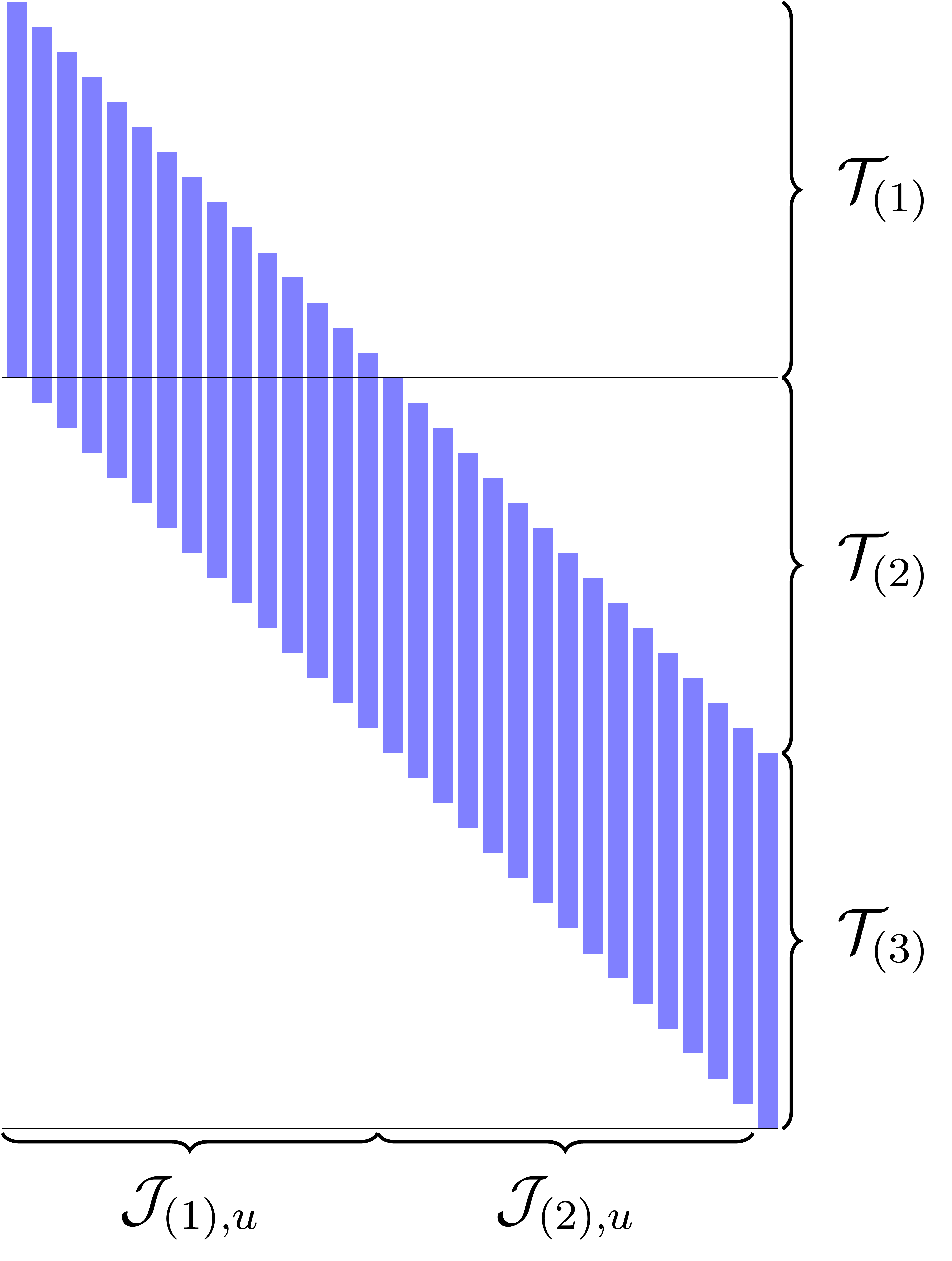}
\caption{Model \ref{everyframe} \label{fig:everyframe}}
\end{figure}

\begin{sigmodel}
In both models above we let $\T_t$ contain consecutive indices. This models a moving 1D object of length s or less that enters the scene and eventually walks out, and then another object of length $s$ or less may come in. However notice that nothing in our general model, Model \ref{general_model}, requires the indices to be consecutive or contiguous in any way. Thus in both of Models \ref{randsup} and \ref{everyframe} above, instead of one moving object, we can also have multiple moving objects as long as the union of their supports is of size at most $s$ and satisfies one of these models.  Also, with minor changes, the object(s) instead of leaving the scene can reflect back up and start moving in the other direction as well.
\end{sigmodel}

%\begin{remark}
%For both of the above models, when the support reaches the bottom of the vector, we  assume that it starts over at $1$. This models a moving 1D object of length $s$ or less that enters the scene and eventually walks out, and then another object of length $s$ or less may come in.  The requirement of consecutive indices and downward (as opposed to upward) motion are done for simplicity and ease of understanding.  Our results still hold under permutations (relabeling) of the indices.  We could also make a small modification and assume that the object is reflected back up (down) when it reaches the bottom (top).  %See Remark \ref{reflection}.
%\end{remark}
%
%\begin{remark}
%??? In all the special cases of Model \ref{general_model}, we have assumed consecutive indices because this models a moving object in video sequences. First of all, we should point out that nothing in Model \ref{general_model} requires the indices to be consecutive or contiguous in any way.  For example, we could have more than one moving object as long as the total support size is $s$ or less and the union of their individual supports is such that either Model \ref{sbyrho} or Model \ref{everyframe} hold.
%\end{remark}

\begin{lem}
If $t_{\max} \le n^{10}$, then Model \ref{randsup} is a special case of Model \ref{sbyrho} (and hence a special case of Model \ref{general_model}) with probability at least $1-n^{-10}$.
\end{lem}

\begin{proof}
The proof has three steps. (a) We first use standard arguments about a Bernoulli sequence \cite{longest_run} to prove that the object moves at least once every $\beta$ time instants with probability at least $1 - 0.5n^{-10}$. The choice of $q$ ensures that this holds. (b) Next we use a standard Gaussian tail bound argument to show that, with probability at least $1 - 0.5n^{-10}$, when it moves, it moves by at least $s/\rrho$ indices and at most $1.2s/\rrho$ indices. The bound on $\sigma^2$ ensures this.
(c) The above two claims ensure that, w.h.p., the object remains static for at most $\beta$ frames at a time and when it moves it moves by at least $s/\rrho$ indices and at most $1.2s/\rrho$ indices. Notice that all the motion is in one direction. Motion by at least $s/\rrho$ in one direction ensures that after the object moves $\rrho$ times, i.e. after $\rrho$ changes of $\T_t$, the sets are disjoint, i.e. $\T^{[k]} \cap \T^{[k + \rrho]}=\emptyset$. Motion by at most $1.2 s /\rho$ in one direction and $1.2\frac{s}{\rrho} \alpha \le n$ ensures the third condition of Model \ref{sbyrho} holds even when the object moves at every frame.
\end{proof}

\begin{lem}\label{everyframelem}
Model \ref{everyframe} is a special case of Model \ref{general_model} with $\rho = 2$ and $h^+ = s/\alpha$.
\end{lem}

See Figure \ref{fig:everyframe} for a diagram of the model and the idea behind its proof.

\begin{proof}

%BL  See if this is more clear?

For the sake of clarity, we will prove the case when the object moves exactly 1 index at every time $t$.  The only difference in the general case is the construction of the $\J_{(i),u}$.

Consider an interval $\J_u$.
Let $t_u:= (u-1)\alpha+1$ denote the first time in $\J_u$.
Without loss of generality (because we can re-label the indices) let the object start at the top of the vector.
That is $\T_{t_u} = [1,s]$.
Let $l_u = \llceil \frac{n}{s}\rrceil$.
Let $\T_{(i),u} = [(i-1)s+1, i s]$ for $i=1,2, \dots, \llfloor \frac{n}{s} \rrfloor$.  If $\frac{n}{s}$ is not an integer, also define
$\T_{\left(\llceil \frac{n}{s}\rrceil\right),u} = \left[\llfloor \frac{n}{s} \rrfloor s +1 , n \right]$.
Define $\J_{(i),u} := [t_u+(i-1)s,t_u+is-1]$ for $i=1,2, \dots, \llfloor \frac{\alpha}{s} \rrfloor$. If $\frac{\alpha}{s}$ is not an integer, also define $\J_{\left(\llceil \frac{\alpha}{s}\rrceil\right),u} = [t_u+\llfloor \frac{\alpha}{s} \rrfloor s,t_u+\alpha-1]$.

Clearly $\J_{(i),u}$ as defined above are a partition of $\J_u$. Also, by construction,  for all
$t\in\J_{(i),u}$, $\T_t \subseteq \T_{(i),u} \cup \T_{(i+1),u}$.  This follows from three facts  1) the assumption that $\T_{t_u}=[1,s]$ (which is just a renumbering of the indices to make the numbers clearer) 2) the object moves down by exactly one index at each time $t$ and 3) $m \leq \frac{n-s}{\alpha}$, so that once an index leaves $\T_t$, it will not return in the next $\alpha$ time instants.  A simpler way of stating fact 3) is that the total motion is such that $\T_t$ does not return to where it started i.e. $\T_{t_u}\cap\T_{t_u+\alpha}=\emptyset$.

Notice that $|\J_{(i),u}|\leq s$ for all $i$. (With the possible exception of the last set, they all have size exactly $s$.)  So under the assumptions of Model \ref{everyframe} $h_u^*(\alpha)\leq s$, which satisfies Model \ref{general_model} with $h^+ = \frac{s}{\alpha}\leq 0.0025\alpha = \frac{0.01\alpha}{2^2} = \frac{0.01\alpha}{\rho^2}$.
\end{proof}

%
%Proof (move to Appendix):
%Suppose that it is of the maximum size $s$.  Without loss of generality suppose that the object starts at index 1, i.e., we have $\T_1 = [1, \dots s]$. Then $\T_2$ contains consecutive indices starting at index 2 or 3 or ... $1+m$ and so on.
%Define $T_{(i),u}:= [(i-1)s+1, i s]$ for $i=1,2 \dots \alpha/s$ and

\subsection{Analyze the ReProCS algorithm that also removes the deleted directions from the subspace estimate}  \label{relax_denseness}
The tools introduced in this paper -- (a) Lemma \ref{blockdiag1} and the way it is applied to bound $\bm{\mathcal{H}}_u$ in Lemma \ref{calHk}; and (b) the detection lemma (Lemma \ref{det}), the no false detection lemma (Lemma \ref{falsedet}) and the p-PCA lemma (Lemma \ref{pPCA}) -- can also be used to get a correctness result for a practical modification of ReProCS with cluster-PCA (ReProCS-cPCA) which is  Algorithm 2 of \cite{ReProCS_IT}. This algorithm was introduced to also  remove the deleted directions from the subspace estimate. It does this by re-estimating the previous subspace at a time after the newly added subspace has been accurately estimated (i.e. at a time after $\that_j+ K \alpha$). A partial result for this algorithm was proved in \cite{ReProCS_IT}.

This result will need one extra assumption -- it will need the eigenvalues of the covariance matrix of $\lt$ to be clustered for a period of time after the subspace change has stabilized, i.e. for a period of $d_2$ frames in the interval $[t_j+d+1, t_{j+1}-1]$ -- but it will have a key advantage. It will need a much weaker denseness assumption and hence a much weaker bound on $r$ or $r_{\text{mat}}$. In particular, with this result we expect to be able to allow $r = r_{\text{mat}} \in  \bigo(n)$ with the same assumptions on $s$ and $s_{\text{mat}}$ that we currently allow. This requirement is almost as weak as that of PCP.

%(a) It will allow the deleted directions to be any direction from $\Span(\bm{P}_{t_j-1})$ for the period $[t_j, t_{j+1})$. (b) Moreover,

\subsection{Relax the independence assumption on $\ell_t$'s} \label{relax_indep}
The results in this work assume that the $\lt$'s are independent over time and zero mean; this is a valid model when background images have independent random variations about a fixed mean. Using the tools developed in this paper, a similar result can also be obtained for the more general case of $\lt$'s following an autoregressive model. This will allow the $\lt$'s to be correlated over time. A partial result for this case was obtained in \cite{zhan_reprocs}. The main change in this case will be that we will need to apply the matrix Azuma inequality from \cite{tail_bound} instead of matrix Hoeffding. This is will also require algebraic manipulation of sums and some other important modifications, as explained in \cite{zhan_reprocs}, so that the constant term after conditioning on past values of the matrix is small.

\subsection{Noisy and Undersampled Online Matrix Completion or Online Robust PCA} \label{relax_undersamp}
We expect that the tools introduced in this paper can also be used to analyze the noisy case, i.e. the case of $\mt = \xt + \lt + \bm{w}_t$ where $\bm{w}_t$ is small bounded noise. In most practical video applications, while the foreground is truly sparse, the background is only approximately low-rank. The modeling error can be handled as $\bm{w}_t$. The proposed algorithms already apply without modification to this case (see \cite{han_tsp} for results on real videos). The reason that our tools will directly extend to the noisy case is this: the sparse recovery step is already a noisy sparse recovery one, its analysis will not change if we also add in more noise due to $\bm{w}_t$. If $\lt$ and $\bm{w}_t$ are assumed independent, then there should be few simple modifications to the analysis of the p-PCA step as well.

%The proof of online robust PCA in this paper relies on the support of $\xt$ being exactly recovered. This is why it needs a fairly large lower bound on $x_{\min}$. This can be relaxed as follows. Include all entries of $\xt$ with magnitude below $x_{\min}$ into the noise $\bm{w}_t$. This way we will not need any lower bound on $x_{\min}$ but will require a stronger slow subspace change assumption (the bounds on $\gamma_\new$ and $\lambda_\new^+$ will be tighter).

Finally, we expect both the algorithm and the proof techniques to apply with simple changes to the undersampled case $\mt = \bm{A}_t \xt + \bm{B}_t \lt + \bm{w}_t$ as long as $\bm{B}_t$ is {\em not} time-varying, i.e. $\bm{B}_t = \bm{B}_0$. A partial result for this case was obtained in \cite{rrpcp_globalsip} and experiments were shown in \cite{han_tsp}.

\section{Conclusions} \label{conclusions}
In this work, we obtained correctness results for online robust PCA and for online matrix completion. Both results needed four key assumptions: (a) accurate initial subspace knowledge; (b) slow subspace change and mutual independence of the $\lt$'s according to Model \ref{exp_model}; (c) {\em some} changes in the set of missing entries (or in the set of outlier-corrupted entries) over time, one way to quantify what is needed is given in Model \ref{sbyrho}; (d) a denseness assumption on the columns of the subspace basis matrices of $\lt$; and (e) algorithm parameters are appropriately set.

Ongoing work includes obtaining the results mentioned in Sections \ref{relax_denseness}, \ref{relax_indep} and \ref{relax_undersamp}. Besides these, we expect the proof techniques developed here to apply to various other problems involving PCA with data and noise terms being correlated.

%\newpage
\appendices
\renewcommand{\thetheorem}{\thesection.\arabic{theorem}}

%%%%%%%%%%%%%%%%%%%%%%%%%%%%%%%%%%%%%%%%%%%%%%%%%%%%%%%%%%%

\section{Proof that Model \ref{sbyrho} on $\T_t$ satisfies the general Model \ref{general_model}} \label{pf_supch}
%Recall that Model \ref{sbyrho} assumes that
%\ben
%\item $\T_t = \T^{[k]}$ for all times $t \in [t^k, t^{k+1})$ with $t^{k+1} - t^k < \beta$ and $|\T^{[k]}| \le s$;
%\item let $\rrho$ be the smallest integer so that for any $k$, $\T^{[k]} \cap \T^{[k+\rrho ]} = \emptyset$; we assume that ${\rrho}^2 \beta \leq  0.01 \alpha$;
%\item for any $k$ and $i$ such that $k < i \le k+\alpha$, $(\T^{[k]} \setminus \T^{[k+1]}) \cap (\T^{[i]} \setminus \T^{[i+1]}) = \emptyset$
%\\ (an implicit requirement for this to hold is $\sum_{i=k+1}^{k+\alpha}|\T^{[i]} \setminus \T^{[i+1]}| \le n$, which will be the case if $|\T^{[i]} \setminus \T^{[i+1]}| \le \frac{s}{\rho_2}$ and $\frac{s}{\rho_2}\alpha \le n$).
%\een

\begin{proof}[Proof of Lemma \ref{spc_case}]
Consider an interval $\J_u$.
%Set $\rho=\rrho$.
We will construct one set of mutually disjoints sets $\{\mathcal{T}_{(i),u}\}_{i=1,2,\dots l_u}$ that are subsets of $\{1,2,\dots n\}$ and a partition $\{\J_{(i),u}\}_{i=1,2,\dots l_u}$ of $\J_u$ so that for all $t \in \J_{(i),u}$, \eqref{union} holds and so that $h_u(\alpha;\{\T_{(i),u}\} , \{\J_{(i),u} \}) \le \beta$ for this choice. Since $h_u^*(\alpha)$ takes the minimum over all such sets, this will imply $h_u^*(\alpha) \le \beta$. By setting $h^+ = \beta/\alpha$ and using the Model \ref{sbyrho} assumption ${\rrho}^2 \beta \le 0.01 \alpha$, we will be done.

Recall from Model \ref{sbyrho} that $\T_t = \T^{[k]}$ for all $t \in [t^k, t^{k+1})$ with $t^{k+1} - t^k < \beta$ and $|\T^{[k]}| \le s$.

Let $t_u:= (u-1)\alpha+1$ denote the first time index of $\J_u$.
Let $k_u$ be the index $k$ for which $t_u \in  [t^k, t^{k+1})$.  In other words, $\T_{t_u} = \T^{[k_u]}$.
Define $l_u$ to be the number of intervals $[t^k, t^{k+1})$ that have non-empty intersection with $\J_u$. So $l_u$ is one plus the number of times $\T_t$ changes in the interval $\J_u$.  For $i=1,2, \dots l_u-1$, define
\[
\T_{(i),u} := \T^{[k_u+i-1]} \setminus \T^{[k_u+i]},
\]
and set $\T_{(l_u),u} = \T^{[k_u+ l_u  -1]}$.
Clearly $l_u \le \alpha$. Thus, by the Model \ref{sbyrho} assumption (for any $k$ and $i$ such that $k < i \le k+\alpha$, $(\T^{[k]} \setminus \T^{[k+1]}) \cap (\T^{[i]} \setminus \T^{[i+1]}) = \emptyset$), the $\T_{(i),u}$'s are mutually disjoint.

Next, define a partition of $\J_u$ as
$$\J_{(i),u}:= [t^{k_u+i-1}, t^{k_u+i}) \cap \J_u \ \text{for} \ i=1,2,\dots l_u$$
%The above sets form a partition of $\J_u$ since (a) by definition of $l_u$, $\cup_{i=1}^{l_u} \J_{(i),u} = \J_u$; and (b) the intervals $[t^k, t^{k+1})$ are mutually disjoint.
By Model \ref{sbyrho} $1 \le t_{k+1} - t_k < \beta$ for all $k$. Since $\J_{(i),u} \subseteq [t^{k_u+i-1}, t^{k_u+i})$, $|\J_{(i),u}| < \beta$ for all $i=1,2,\dots l_u$. %and thus $h_u(\alpha;.) \le \beta$ for this choice of $\T_{(i),u}$'s and $\J_{(i),u}$'s.

Notice that for all $t \in \J_{(i),u}$, $\T_t = \T^{[k_u+i-1]}$. So if we can show that $\T^{[k_u+i-1]} \subseteq \T_{(i),u} \cup \T_{(i+1),u} \dots \cup \T_{(i+\rho-1),u}$ for all $i=1,2, \dots l_u$, we will be done since this will imply $h_u^*(\alpha) \le \beta$. %  This will then imply that $h_u^*(\alpha) \le \beta$.
% we use the Model \ref{sbyrho} assumption that $\T^{[k]}\cap \T^{[k+\rho]}=\emptyset$.
To show this, set $k = k_u+i-1$. Then,
\begin{align*}
\T^{[k]} &= \T_{(i),u} \cup [\T^{[k]} \cap \T^{[k+1]}] \\
& = \T_{(i),u} \cup [\T^{[k]} \cap\T^{[k+1]} \setminus \T^{[k+2]}] \cup [\T^{[k]} \cap \T^{[k+1]} \cap \T^{[k+2]} ] \\
&\subseteq \T_{(i),u} \cup \T_{(i+1),u} \cup [\T^{[k]} \cap \T^{[k+1]} \cap \T^{[k+2]} ]  \\
&= \T_{(i),u}  \cup \T_{(i+1),u} \cup [\T^{[k]} \cap \T^{[k+1]} \cap \T^{[k+2]} \setminus \T^{[k+3]}) \cup [\T^{[k]} \cap \T^{[k+1]} \cap \T^{[k+2]} \cap \T^{[k+3]} ]\\
& \subseteq \T_{(i),u}  \cup \T_{(i+1),u} \cup \T_{(i+2)} \cup [\T^{[k]} \cap \T^{[k+1]} \cap \T^{[k+2]} \cap \T^{[k+3]}].
\end{align*}
Continuing in the same manner as above, we get,
\begin{eqnarray} \label{union_holds}
\T^{[k]} & \subseteq & \T_{(i),u}  \cup \T_{(i+1),u} \cup \dots \cup \T_{(i+\rho-1),u} \cup [\T^{[k]} \cap \T^{[k+1]} \cap \dots \cap \T^{[k+\rho]}]  \nonumber \\
&=& \T_{(i),u}\cup \T_{(i+1),u}\cup\dots\cup\mathcal{T}_{(i+\rho-1),u}
\end{eqnarray}
The last line is because $\T^{[k]}\cap \T^{[k+\rho]}=\emptyset$ by Model \ref{sbyrho}.
%
%By Model \ref{sbyrho}, $t_{k+1} - t_k < \beta$ and thus $|\J_{(i),u}| \le \beta$ and so $h_u(\alpha;.) \le \beta$. Since $h_u^*(\alpha)$ is obtained by taking the min over all choices of the sets $\T_{(i),u}$ and all choices of the valid partitions $\J_{(i),u}$, we have
%\[
%h_u^*(\alpha) \le \beta.
%\]
\end{proof}

%\section{Proof of Lemma \ref{Dnew0_lem} and Proof of Lemma \ref{zetadecay}}% and on \texorpdfstring{$\zeta_{j,\new,k}^+$}{zeta j new k plus})}
\section{Proof of Lemma \ref{zetadecay} (bound on \texorpdfstring{$\zeta_{j,\new,k}^+$}{zeta j new k plus}) and of  Lemma \ref{Dnew0_lem}}
\label{zeta_k_section}

\begin{proof}[Proof of Lemma \ref{zetadecay}]
This proof's approach is similar to that of \cite[Lemma 6.1]{ReProCS_IT}. The details have some differences because our main result now uses different assumptions.

This lemma uses Model \ref{general_model}. As shown in Lemma \ref{spc_case}, Model \ref{sbyrho} is a special case of this general model.

Recall that  $\ds\zeta_{j,\new,k}^+ := \frac{b_{\bm{\mathcal{H}},k}}{b_{\bm{A}} - b_{\bm{A},\perp} - b_{\bm{\mathcal{H}},k}} $ with the terms on the RHS defined in Lemmas \ref{Ak}, \ref{Akperp}, \ref{calHk}.

Recall that $\epsilon = 0.01r_{\new}\zeta\lammin$.
Divide the numerator and denominator by $\lammin$.
Define
%NV - removed 2 from the et et' term both in C_k and in B_k
\begin{align*}
B_{k} := \begin{cases}
\begin{array}{l}
 \left[ \rho^2h^+(\phi^+)^2(\kappa_s^+)^2(\zeta_{j,\new,k-1}^+) + 2\kappa_s^+\phi^+  \right]\left(\frac{\lambda_{\new}^+}{\lammin}\right)
\end{array} & k = 1 \\
\vspace{-.2in}\\
\begin{array}{l}
\left[ \rho^2h^+(\phi^+)^2\zeta_{j,\new,k-1}^+ + 2\sqrt{\rho^2h^+}\phi^+  \right]\left(\frac{\lambda_{\new}^+}{\lammin}\right)
\end{array} & k\geq2
\end{cases}
\end{align*}

%NV - added \sqrt{\rho^2h^+} where we discussed; removed 2 from the et et' term both in C_k and in B_k
\begin{align*}
C_{k} := \left[ \rho^2h^+(\phi^+)^2(\zeta_{j,*}^+)r + 2\sqrt{\rho^2h^+}\phi^+(\zeta_{j,*}^+)r + 2(\zeta_{j,*}^+)r  \right]\left(\frac{\lambda^+}{\lammin}\right)   + 0.05
\end{align*}

%\begin{align*}
%C_{k} := \begin{cases}
%\begin{array}{l}
% \left[ \rho^2h^+(\phi^+)^2(\zeta_{j,*}^+)r + 2\sqrt{\rho^2h^+}\phi^+(\zeta_{j,*}^+)r + 2(\zeta_{j,*}^+)r  \right]\left(\frac{\lambda^+}{\lammin}\right)   + 0.08
% \end{array} & k=1 \\
% \vspace{-.2in}\\
% \begin{array}{l}
%\left[ \rho^2h^+(\phi^+)^2(\zeta_{j,*}^+)r + 2\sqrt{\rho^2h^+}\phi^+(\zeta_{j,*}^+)r + 2(\zeta_{j,*}^+)r  \right]\left(\frac{\lambda^+}{\lammin}\right)  + 0.08
% \end{array} & k\geq2
% \end{cases}
%\end{align*}

\begin{align*}
D_{k} := & 1  - (\zeta_{j,*}^+)^2   -
(\zeta_{j,*}^+)^2 \left(\frac{\lambda^+}{\lammin}\right)  -  \zeta_{j,\new,k-1}^+ B_k - r_{\new}\zeta (C_k+.02)
\end{align*}
Then,
\[
\zeta_{j,\new,k}^+ \leq \zeta_{j,\new,k-1}^+\frac{B_k}{D_k} + r_{\new}\zeta\frac{C_k}{D_k}.
\]

Recall that $\kappa_s^+=0.0215$ and $\phi^+ = 1.2$.
It is not difficult to see that $\zeta_{j,\new,k}^+$ is an increasing function of $\rho^2 h^+$, $r$, $\zeta$, $\zeta \frac{\lambda^+}{\lammin}$, and $\frac{\lambda_{\new}^+}{\lammin}$ and $\zeta_{j,\new,k-1}^+$.
Consider $k=1$. Using $\zeta_{j,\new,0}^+=1$ and the upper bounds assumed in Theorem \ref{thm1} on the above quantities, we get that $\zeta_{j,\new,1}^+ \le 0.18$.

Thus, $\zeta_{j,\new,1}^+ \le \zeta_{j,\new,0}^+=1$. Using this and the fact that $\zeta_{j,\new,k}^+$ is an increasing function of $\zeta_{j,\new,k-1}^+$, we can show by induction that $\zeta_{j,\new,k}^+ \le \zeta_{j,\new,k-1}^+$. Thus, $\zeta_{j,\new,k}^+ \le \zeta_{j,\new,1}^+ \le 0.18$ for all $k=1,2 \dots K$.

Using $\zeta_{j,\new,k}^+ \le 0.18$ and the bounds assumed in Theorem \ref{thm1} on the other quantities we get that
\[
\zeta_{j,\new,k}^+ \leq 0.83 \zeta_{j,\new,k-1}^+ + 0.14 r_{\new}\zeta
\]
Using this, we get
\begin{align*}
\zeta_{j,\new,k}^+  \leq 0.83\zeta_{j,\new,k-1}^{+} + 0.14 r_{\new}\zeta & \leq \zeta_{j,\new,0}^{+} (0.83)^{k} + \sum_{i = 0}^{k-1}(0.83)^i(0.14) r_{\new}\zeta \\
&\leq \zeta_{j,\new,0}^{+} (0.83)^{k} + \sum_{i = 0}^{\infty}(0.83)^i(0.14) r_{\new}\zeta \\
&\leq 0.83^k + 0.84 r_{\new}\zeta.
\end{align*}
\end{proof}

\begin{proof}[Proof of Lemma \ref{Dnew0_lem}]
Recall that $\bm{D}_{j,\new} = (\bm{I} - \Phat_{(j),*}\Phat_{(j),*}{}') \bm{P}_{(j),\new}$.
Then $\|{\I_{\mathcal{T}}}' \bm{D}_{j,\rmnew} \|_2 = \|{\I_{\mathcal{T}}}'  (\bm{I} - \Phat_{(j),*}\Phat_{(j),*}{}') \bm{P}_{(j),\new} \|_2 \leq \| {\I_{\mathcal{T}}}' \bm{P}_{(j),\new} \|_2 + \|\Phat_{(j),*}{}' \bm{P}_{(j),\new} \|_2 \leq  \kappa_s(\bm{P}_{(j),\rmnew}) + \|\Phat_{(j),*}{}'(\I - \bm{P}_{(j),*}{\bm{P}_{(j),*}}')\bm{P}_{(j),\new}\|_2 \leq \kappa_s(\bm{P}_{(j),\rmnew}) +  \zeta_{j,*}$.
The event $X_{\hat{u}_j + k-1} \in \Gamma_{j,k-1}^{\hat{u}_j}$ implies that $\zeta_{j,*} \le \zeta_{j,*}^+ \le 0.0015$. Thus, the lemma follows.
\end{proof}

\section{Proof of the Compressed Sensing (CS) Lemma (Lemma \ref{cslem})} \label{pf_cslem}
This proof's approach is similar to that of \cite[Lemma 6.4]{ReProCS_IT}. The details have some differences because our main result now uses different assumptions.
The proof uses the denseness assumption and subspace error bounds $\zeta_{j,*} \leq \zeta_{j,*}^{+}$ and $\zeta_{j,\new,k-1} \leq \zeta_{j,\new,k-1}^{+}$, that hold when $X_{\hat{u}_j+k-1}\in\Gamma_{j,k-1}^{\hat{u}_j}$ for $\hat{u}_j=u_j$ or $\hat{u}_j = u_j+1$, to obtain bounds on the restricted isometry constant (RIC) of the sparse recovery matrix $\bm{\Phi}_t$ and the sparse recovery error $\| \bm{b}_t \|_2$.
Applying the noisy compressed sensing (CS) result from \cite{candes_rip} and the assumed bounds on $\zeta$ and $\gamma$, the lemma follows.

\begin{lem} \cite[Lemma 2.10]{ReProCS_IT}\label{lemma0}\label{hatswitch}
Suppose that $\bm{P}$, $\Phat$ and $\bm{Q}$ are three basis matrices. Also, $\bm{P}$ and $\Phat$ are of the same size, $\bm{Q}' \bm{P} = \bm{0}$ and $\|(\I-\Phat \Phat{}' ) \bm{P} \|_2 = \zeta_*$. Then,
\begin{enumerate}
  \item $\|(\I-\Phat\Phat{}')\bm{P}\bm{P}'\|_2 =\|( \I - \bm{P}\bm{P}' ) \Phat \Phat{}'\|_2 =  \|( \I - \bm{P} \bm{P}' ) \Phat \|_2 = \| ( \I - \Phat \Phat{}' ) \bm{P}\|_2 =  \zeta_*$
  \item $\|\bm{P} \bm{P}' - \Phat \Phat{}'\|_2 \leq 2 \|(\I-\Phat \Phat{}')\bm{P}\|_2 = 2 \zeta_*$
  \item $\|\Phat{}' \bm{Q}\|_2 \leq \zeta_*$ \label{lem_cross}
  \item $ \sqrt{1-\zeta_*^2} \leq \sigma_i\left((\I-\Phat \Phat{}')\bm{Q}\right)\leq 1 $
\end{enumerate}
\end{lem}

%The proof of Lemma \ref{hatswitch} is straightforward and is given in \cite{ReProCS_IT}.

%Recall that $\kappa_s(\P)$ is defined in Definition \ref{def_kappa}.
We begin by first bounding the RIC of the CS matrix $\bm{\Phi}_{t}$.  We will use the notation $\kappa_s^2(\bm{P})$ to mean $\left(\kappa_s(\bm{P})\right)^2$.

\begin{lem}[{Bounding the RIC of $\bm{\Phi}_{t}$ \cite[Lemma 6.6]{ReProCS_IT}}] \label{RIC_bnd}
Recall that $\zeta_{j,*}:= \|(\I-\Phat_{(j),*}\Phat_{(j),*}{}')\bm{P}_{(j),*}\|_2$.  The following hold.
\begin{enumerate}
\item Suppose that a basis matrix $\bm{P}$ can be split as $\bm{P} = [\bm{P}_1 \ \bm{P}_2]$ where $\bm{P}_1$ and $\bm{P}_2$ are also basis matrices. Then $\kappa_s^2 (\bm{P}) = \max_{\mathcal{T}: |\mathcal{T}| \leq s} \|{\I_{\mathcal{T}}}'\bm{P}\|_2^2 \le \kappa_s^2 (\bm{P}_1) + \kappa_s^2 (\bm{P}_2)$.
\item $\kappa_s^2(\Phat_{(j),*}) \leq (\kappa_{s,*})^2 + 2\zeta_*$ for all $j$
\item $\kappa_s (\Phat_{(j),\rmnew,k}) \leq \kappa_{s,\rmnew} + \zeta_{j,\new,k} + \zeta_{j,*}$ for all $j$ and $k$.
\item For $t\in[(u_{j-1}+K)\alpha+1,(\hat{u}_j + 1)\alpha)$, $\delta_{s} (\bm{\Phi}_{t}) = \kappa_s^2 (\Phat_{(j),*}) \leq  (\kappa_{s,*})^2 + 2 \zeta_{j,*}$.
\item For $k=1,\dots,K-1$, for $t\in[(\hat{u}_j+k)\alpha+1,(\hat{u}_{j}+k+1)\alpha] $ $\delta_{s}(\bm{\Phi}_{t})  = \kappa_s^2 ([\Phat_{(j),*} \ \Phat_{(j),\rmnew,k}]) \leq \kappa_s^2 (\Phat_{(j),*}) + \kappa_s^2 (\Phat_{(j),\rmnew,k}) \leq (\kappa_{s,*})^2 + 2\zeta_{j,*} + (\kappa_{s,\rmnew} +  \zeta_{j,\new,k} + \zeta_{j,*})^2$.
\end{enumerate}
\end{lem}

\begin{proof}
\begin{enumerate}
\item Recall that $\kappa_s^2 (\bm{P}) = \max_{|\mathcal{T}| \leq s} \|{\I_\mathcal{T}}' \bm{P}\|_2^2$. Also, $\|{\I_\mathcal{T}}' \bm{P}\|_2^2 = \|{\I_\mathcal{T}}' [\bm{P}_1 \ \bm{P}_2] [\bm{P}_1 \ \bm{P}_2]' \I_\mathcal{T} \|_2 =  \|{\I_\mathcal{T}}' (\bm{P}_1 {\bm{P}_1}' + \bm{P}_2 {\bm{P}_2}') \I_\mathcal{T} \|_2 \le \|{\I_\mathcal{T}}' \bm{P}_1 {\bm{P}_1}' \I_\mathcal{T}\|_2 + \|{\I_\mathcal{T}}' \bm{P}_2 {\bm{P}_2}' \I_\mathcal{T} \|_2$. Thus, the inequality follows.

\item For any set $\mathcal{T}$ with $|\mathcal{T}| \le s$, $\|{\I_\mathcal{T}}' \Phat_{(j),*}\|_2^2  = \|{\I_\mathcal{T}}' \Phat_{(j),*} \Phat_{(j),*}{}' \I_\mathcal{T}\|_2 =\|{\I_\mathcal{T}}'( \Phat_{(j),*} \Phat_{(j),*}{}' - \bm{P}_{(j),*} {\bm{P}_{(j),*}}' + \bm{P}_{(j),*} {\bm{P}_{(j),*}}') \I_\mathcal{T}\|_2 \leq \|{\I_\mathcal{T}}'( \Phat_{(j),*} \Phat_{(j),*}{}' - \bm{P}_{(j),*} {\bm{P}_{(j),*}}' ) \I_\mathcal{T} \|_2 + \|{\I_\mathcal{T}}' \bm{P}_{(j),*} {\bm{P}_{(j),*}}' \I_\mathcal{T}\|_2 \leq 2\zeta_{j,*} + (\kappa_{s,*})^2$. The last inequality follows using Lemma \ref{lemma0} with $\bm{P} = \bm{P}_{(j),*}$ and $\hat{\bm{P}} = \hat{\bm{P}}_{(j),*}$.

\item By Lemma \ref{lemma0} with $\bm{P} = \bm{P}_{(j),*}$, $\Phat = \Phat_{(j),*}$ and $\bm{Q} = \bm{P}_{(j),\rmnew}$, $\|{\bm{P}_{(j),\new}}' \Phat_{(j),*} \|_2 \leq \zeta_{j,*}$.
By Lemma \ref{lemma0} with $\bm{P} = \bm{P}_{(j),\new}$ and $\hat{\bm{P}} = \hat{\bm{P}}_{(j),\new,k}$, $\|( \I - \bm{P}_{(j),\new} {\bm{P}_{(j),\new}}' ) \Phat_{(j),\new,k} \|_2  = \|( \I - \Phat_{(j),\new,k}\Phat_{(j),\new,k}{}') \bm{P}_{(j),\new}\|_2$.

For any set $\mathcal{T}$ with $|\mathcal{T}| \leq s$, $\|{\I_\mathcal{T}}' \Phat_{(j),\new,k}\|_2 \leq \|{\I_\mathcal{T}} '( \I - \bm{P}_{(j),\new} {\bm{P}_{(j),\new}}') \Phat_{(j),\new,k} \|_2 + \|{\I_\mathcal{T}}' \bm{P}_{(j),\new} {\bm{P}_{(j),\new}}' \Phat_{(j),\new,k}\|_2 \leq \|( \I - \bm{P}_{(j),\new}{\bm{P}_{(j),\new}}') \Phat_{(j),\new,k} \|_2 + \|{\I_\mathcal{T}}' \bm{P}_{(j),\new}\|_2 =  \|(\I - \Phat_{(j),\new,k}\Phat_{(j),\new,k}{}') \bm{P}_{(j),\new}\|_2 + \|{\I_\mathcal{T}}' \bm{P}_{(j),\new}\|_2 \leq  \|\bm{D}_{(j),\new,k}\|_2 + \| \Phat_{(j),*} \Phat_{(j),*}{}' \bm{P}_{(j),\new}\|_2 + \|{\I_\mathcal{T}}' \bm{P}_{(j),\new}\|_2 $. Taking $\max$ over $|\mathcal{T}| \le s$ the claim follows.

\item This follows using Lemma \ref{kappadelta} and the second claim of this lemma.

\item This follows using Lemma \ref{kappadelta} and the first three claims of this lemma.
\end{enumerate}

\end{proof}

\begin{corollary}\label{RICnumbnd}
\
\begin{enumerate}
\item Conditioned on $\Gamma_{j-1,\rmend}$, for $t\in[t_j,(\uhat_j + 1)\alpha]$, $\delta_s(\bm{\Phi}_{t}) \leq \delta_{2s}(\bm{\Phi}_{t})  \leq (\kappa_{2s,*})^2 + 2\zeta_{j,*}^+ < 0.1 < 0.1479$, and $\|  [ ({\bm{\Phi}_{t})_{\mathcal{T}_t}}'(\bm{\Phi}_{t})_{\mathcal{T}_t}]^{-1} \|_2 \le \frac{1}{1-\delta_s(\bm{\Phi}_{t})} < 1.2 := \phi^+$.

\item For $k=2,\dots,K$ and $\hat{u}_j = u_j$ or $\hat{u}_j = u_j + 1$, conditioned on $\Gamma_{j,k-1}^{\hat{u}_j}$, for $t\in[(\hat{u}_{j}+k-1)\alpha+1,(\hat{u}_j + k)\alpha]$, $\delta_s(\bm{\Phi}_{t}) \leq \delta_{2s}(\bm{\Phi}_{t}) \leq (\kappa_{2s,*})^2 + 2\zeta_{j,*}^+ + (\kappa_{2s,\rmnew} +  \zeta_{j,\new,k-1}^+ + \zeta_{j,*}^+)^2 < 0.1479$, and $\|  [ ({\bm{\Phi}_{t})_{\mathcal{T}_t}}'(\bm{\Phi}_{t})_{\mathcal{T}_t}]^{-1} \|_2 \le \frac{1}{1-\delta_s(\bm{\Phi}_{t})} < 1.2 := \phi^+$.

\item For $\hat{u}_j = u_j$ or $\hat{u}_j = u_j + 1$, conditioned on $\Gamma_{j,K}^{\hat{u}_j}$, for $t\in[(\hat{u}_j+K)\alpha+1,t_{j+1}-1]$,  $\delta_s(\bm{\Phi}_{t}) \leq \delta_{2s}(\bm{\Phi}_{t})  \leq (\kappa_{2s,*})^2 + 2\zeta_{j,*}^+ < 0.1 < 0.1479$, and $\|  [ ({\bm{\Phi}_{t})_{\mathcal{T}_t}}'(\bm{\Phi}_{t})_{\mathcal{T}_t}]^{-1} \|_2 \le \frac{1}{1-\delta_s(\bm{\Phi}_{t})} < 1.2 := \phi^+$.
\end{enumerate}
\end{corollary}

\begin{proof}
This follows using Lemma \ref{RIC_bnd}, the definitions of $\Gamma_{j-1,\rmend}$ and $\Gamma_{j,k}^{\hat{u}_j}$, and the bound on $\zeta_{j,\new,k-1}^+$ from Lemma \ref{zetadecay}.
\end{proof}

The following are straightforward bounds that will be useful for the proof of Lemma \ref{cslem}.

\begin{fact}\label{constants}
Under the assumptions of Theorem \ref{thm1}:
\begin{itemize}
\item $ \zeta_{j,*}^+ \gamma \le \frac{\sqrt{\zeta}}{\sqrt{r_0 + (J-1)c}} \le \sqrt{\zeta}$
\item $\zeta_{j,\new,k-1}^+ \leq 0.83^{k-1} + 0.84 r_\new \zeta$ (from Lemma \ref{zetadecay})
\item $\zeta_{j,\new,k-1}^+ \gamma_{\rmnew} \leq 0.83^{k-1} \gamma_{\rmnew} + 0.84 r_\new \zeta \gamma_{\rmnew} \leq 0.83^{k-1}\gamma_{\rmnew} + 0.3\sqrt{\zeta} $

%\item $\zeta_{k-1}^+ {\gamma_{\rmnew}}^2 \leq (0.6 \cdot 1.2^2)^{k-1} \gamma_{\rmnew}^2 + 0.4 r_\new \zeta \gamma^2 \le 0.864^{k-1}\gamma_{\rmnew}^2 + \frac{0.4}{{(r_0 + (J-1)c})^2} \le 0.864^{k-1}\gamma_{\rmnew}^2 + 0.4$
\end{itemize}
\end{fact}

\begin{proof}[Proof of Lemma \ref{cslem}]
We will prove claim 2).  The others are done in the same way.

Recall that $\Gamma_{j,k-1}^{\hat{u}_j}$ implies that $\zeta_{j,*} \leq \zeta_{j,*}^+$ and $\zeta_{j,\new,k-1}\leq \zeta_{j,\new,k-1}^+$.

\begin{enumerate}[a)]
\item For $t \in [(\uhat_j+k-1)\alpha+1, (\hat{u}_j+k)\alpha]$, $\bm{b}_t := ( \I - \Phat_{t-1} \Phat_{t-1} {}') \lt = \bm{D}_{j,*,k-1} \bm{a}_{t,*} + \bm{D}_{j,\rmnew,k-1} \bm{a}_{t,\rmnew} $. Thus, using Fact \ref{constants}
\begin{align*}
\|\bm{b}_t\|_2 & \leq \zeta_{j,*} \sqrt{r} \gamma + \zeta_{j,\new,k-1} \sqrt{r_\new} \gamma_{\rmnew} \\
& \leq \sqrt{\zeta}\sqrt{r} + (0.83^{k-1}\gamma_{\rmnew} + 0.84\sqrt{\zeta})\sqrt{r_\new} \\
& = \sqrt{r_\new} 0.83^{k-1} \gamma_{\rmnew} + \sqrt{\zeta} (\sqrt{r} + 0.84\sqrt{r_\new}) \leq \xi.
\end{align*}

\item By Corollary \ref{RICnumbnd}, $\delta_{2s} (\bm{\Phi}_{t}) < 0.15 < \sqrt{2}-1$. Given $|\mathcal{T}_t| \leq s$, $\|\bm{b}_t\|_2 \leq \xi$, by the theorem in \cite{candes_rip}, the CS error satisfies
\[
\|\hat{\bm{x}}_{t,\cs} - \xt \|_2 \leq  \frac{4\sqrt{1+\delta_{2s}(\bm{\Phi}_{t})}}{1-(\sqrt{2}+1)\delta_{2s}(\bm{\Phi}_{t})} \xi < 7 \xi.
\]

\item Using the above, $\|\hat{\bm{x}}_{t,\cs} - \xt\|_{\infty} \leq 7  \xi$. Since $\min_{i\in \mathcal{T}_t} |(\xt)_{i}| \geq x_{\min}$ and $(\xt)_{\mathcal{T}_t^c} = 0$, $\min_{i\in \mathcal{T}_t} |(\hat{\bm{x}}_{t,\cs})_i| \geq x_{\min} - 7 \xi$ and $\max_{i \in \bar{\mathcal{T}_t}} |(\hat{\bm{x}}_{t,\cs})_i| \leq 7 \xi$. If $\omega < x_{\min} - 7 \xi$, then $\hat{\mathcal{T}}_t \supseteq \mathcal{T}_t$. On the other hand, if $\omega > 7 \xi$, then $\hat{\mathcal{T}}_t \subseteq \mathcal{T}_t$. Since $\omega$ satisfies $7 \xi \leq \omega \leq x_{\min} -7 \xi$, the support of $\xt$ is exactly recovered, i.e. $\hat{\mathcal{T}}_t = \mathcal{T}_t$.

\item Given $\hat{\mathcal{T}}_t = \mathcal{T}_t$, the least squares estimate of $\xt$ satisfies $(\hat{\bm{x}}_t)_{\mathcal{T}_t} = [(\bm{\Phi}_{t})_{\mathcal{T}_t}]^{\dag} \bm{y}_t = [ (\bm{\Phi}_{t})_{\mathcal{T}_t}]^{\dag} (\bm{\Phi}_{t} \xt + \bm{\Phi}_{t} \lt)$ and $(\hat{\bm{x}}_t)_{\bar{\mathcal{T}_t}} = \bm{0}$.
Also,  ${(\bm{\Phi}_{t})_{\mathcal{T}_t}}' \bm{\Phi}_{t} = {\I_{\mathcal{T}_t}}' \bm{\Phi}_{t}$ (this follows since $(\bm{\Phi}_{t})_{\mathcal{T}_t} = \bm{\Phi}_{t} \I_{\mathcal{T}_t}$ and ${\bm{\Phi}_{t}}'\bm{\Phi}_{t} = \bm{\Phi}_{t}$).
Using this, the error $\et := \hat{\bm{x}}_t - \xt$ satisfies (\ref{etdef0}).
Thus, using Fact \ref{constants} and the bounds on $\|\bm{a}_t\|_{\infty}$ and $\|\bm{a}_{t,\rmnew}\|_{\infty}$, for $t \in [(\uhat_j+k-1)\alpha+1, (\hat{u}_j+k)\alpha]$,
\begin{align*}
\|\bm{e}_t\|_2 &\leq
 \phi^+ (\zeta_{j,*}^+ \sqrt{r}\gamma +  \zeta_{j,\new,k-1}^+ \sqrt{r_\new}\gamma_{\rmnew}) \le 1.2 \left(1.06\sqrt{\zeta} +  (0.83)^{k-1}\sqrt{r_\new}\gamma_{\rmnew}  \right)
\end{align*}
The last inequality follows from Lemma \ref{zetadecay}.
\end{enumerate}

\end{proof}

\section{Proof of Cauchy-Schwarz inequality for matrices}\label{prelim}

\begin{lem}[Cauchy-Schwarz for a sum of vectors]\label{CSsum}
For vectors $\bm{x}_t$ and $\bm{y}_t$,
\[
\left(\sum_{t=1}^{\alpha} {\bm{x}_t}'\bm{y}_t\right)^2 \leq \left( \sum_t \|\bm{x}_t\|_2^2 \right) \left( \sum_t \|\bm{y}_t\|_2^2 \right)
\]
\end{lem}

\begin{proof}
\begin{align*}
\left(\sum_{t=1}^{\alpha} {\bm{x}_t}'\bm{y}_t\right)^2 = \left( [ {\bm{x}_1}', \dots, {\bm{x}_{\alpha}}'] \left[\begin{array}{c}
\bm{y}_{1}\\ \vdots \\ \bm{y}_{\alpha}
\end{array} \right] \right)^2 \leq \left\|\left[\begin{array}{c}
\bm{x}_{1}\\ \vdots \\ \bm{x}_{\alpha}
\end{array} \right] \right\|_2^2 \left\| \left[\begin{array}{c}
\bm{y}_{1}\\ \vdots \\ \bm{y}_{\alpha}
\end{array} \right] \right\|_2^2 = \left(\sum_{t=1}^{\alpha}\|\bm{x}_t\|_2^2\right)\left(\sum_{t=1}^{\alpha}\|\bm{y}_t\|_2^2\right)
\end{align*}
The inequality is by Cauchy-Schwarz for a single vector.
\end{proof}

\begin{lem}[Cauchy-Schwarz for a sum of matrices]\label{CSmat}
For matrices $\bm{X}_t$ and $\bm{Y}_t$,
\[
\left\|\frac{1}{\alpha} \sum_{t=1}^{\alpha} \bm{X}_t {\bm{Y}_t}'\right\|_2^2 \leq \lambda_{\max}\left(\frac{1}{\alpha} \sum_{t=1}^{\alpha} \bm{X}_t {\bm{X}_t}'\right)
\lambda_{\max}\left(\frac{1}{\alpha} \sum_{t=1}^{\alpha} \bm{Y}_t {\bm{Y}_t}'\right)
\]
\end{lem}

\begin{proof}[Proof of Lemma \ref{CSmat}]
\begin{align*}
\left\| \sum_{t=1}^{\alpha} \bm{X}_t {\bm{Y}_t}'\right\|_2^2 &=
\max_{\substack{\|\bm{x}\|=1\\\|\bm{y}\|=1}} \left| \bm{x}'\left(\sum_{t}\bm{X}_t {\bm{Y}_t}'\right)\bm{y} \right|^2 \\
&= \max_{\substack{\|\bm{x}\|=1\\\|\bm{y}\|=1}} \left| \sum_{t=1}^{\alpha} ({\bm{X}_t}'\bm{x})'({\bm{Y}_t}'\bm{y})  \right|^2 \\
&\leq \max_{\substack{\|\bm{x}\|=1\\\|\bm{y}\|=1}} \left( \sum_{t=1}^{\alpha} \left\| {\bm{X}_t}'\bm{x} \right\|_2^2\right) \left( \sum_{t=1}^{\alpha} \left\| {\bm{Y}_t}'\bm{y} \right\|_2^2\right) \\
&= \max_{\|\bm{x}\|=1}  \bm{x}' \sum_{t=1}^{\alpha} \bm{X}_t{\bm{X}_t}' \ \bm{x}  \ \cdot \ \max_{\|\bm{y}\|=1} \bm{y}' \sum_{t=1}^{\alpha} \bm{Y}_t{\bm{Y}_t}' \ \bm{y}\\
&= \lambda_{\max}\left( \sum_{t=1}^{\alpha} \bm{X}_t{\bm{X}_t}' \right)\lambda_{\max}\left( \sum_{t=1}^{\alpha} \bm{Y}_t{\bm{Y}_t}' \right)
\end{align*}
The inequality is by Lemma \ref{CSsum}. The penultimate line is because $\|\bm{x}\|_2^2 = {\bm{x}'\bm{x}}$.
Multiplying both sides by $\left(\frac{1}{\alpha}\right)^2$ gives the desired result.
\end{proof}

\bibliographystyle{IEEEtran}
\bibliography{tipnewpfmt_kfcsfullpap,tipnewpfmtNIPS}
\end{document}